\documentclass[11pt]{article}
\usepackage{hyperref}
\usepackage{graphicx}
\usepackage{graphics}
\usepackage{dsfont}
\usepackage{epsfig}
\usepackage{amsmath,amssymb,amsthm,amscd}
\usepackage[parfill]{parskip}

\newcommand{\nn}{{\nonumber}}

\def\Mgn[#1]#2{{\overline{\cal M}_{#1,#2}}}
\newcommand{\ba}{\begin{eqnarray*}}
\newcommand{\ea}{\end{eqnarray*}}
\newcommand{\ban}{\begin{eqnarray}}
\newcommand{\ean}{\end{eqnarray}}
\newcommand{\be}{\begin{equation}}
\newcommand{\ee}{\end{equation}}
\newcommand{\ben}{\begin{equation}}
\newcommand{\een}{\end{equation}}
\numberwithin{equation}{section}

\newcommand{\IZ}{\mathbb{Z}}
\newcommand{\IC}{\mathbb{C}}
\newcommand{\IP}{\mathbb{P}}

\newcommand{\del}{\partial}

\newcommand{\Ext}{{\rm Ext}}

\newcommand\cC{{\mathcal C}} \newcommand\cF{{\mathcal F}}
\newcommand\cH{{\mathcal H}} \newcommand\cI{{\mathcal I}}
\newcommand\cO{{\mathcal O}} 
\newcommand{\M}{{\mathcal M}}

\newcommand\CC{\mathbb{C}}
\newcommand\LL{\mathbf{L}}
\newcommand\PP{\mathbb{P}}
\newcommand\ZZ{\mathbb{Z}}

\usepackage{mathrsfs}
\newcommand\I{\mathrm{I}^-}
\newcommand\II{\mathrm{I\hspace{-1.3pt}I}}
\newcommand\III{\mathrm{I\hspace{-1.3pt}I\hspace{-1.3pt}I}}

\newcommand\su{\mathrm{SU}(2)}

\newtheorem{conj}{Conjecture}
\newtheorem{defn}{Definition}
\newtheorem{prop}{Proposition}
\newtheorem{lem}{Lemma}

\begin{document}
\begin{titlepage}
%{}~ \hfill\vbox{ \hbox{} }\break

\begin{flushright}
\parbox[t]{1.8in}{BONN-TH-2012-16}
\end{flushright}

\vskip 2.5 cm
\centerline{\Large \bf The refined BPS index from stable pair invariants}
\vskip 0.5 cm
\renewcommand{\thefootnote}{\fnsymbol{footnote}}
\vskip 20pt \centerline{
{\large \rm Jinwon Choi\footnote{choi29@illinois.edu}, Sheldon Katz\footnote{katz@math.uiuc.edu}, and Albrecht Klemm\footnote{aklemm@th.physik.uni-bonn.de}  } } \vskip .5cm \vskip 20pt

\begin{center}
{\em  $^{*} \phantom{}^\dag$Department of Mathematics  \\ [.1 cm]
University of Illinois at Urbana-Champaign, 1409 W.\ Green St., Urbana, IL 61801}\\[3mm]
{\em  $^\ddag$Bethe Center for Theoretical Physics, Physikalisches
  Institut \\ [.1 cm]
Universit\"at Bonn, Nussallee 12, 53115 Bonn, Germany}
\end{center}

\setcounter{footnote}{0}
\renewcommand{\thefootnote}{\arabic{footnote}}
\vskip 60pt
\begin{abstract}
\end{abstract}
A refinement of the stable pair invariants of Pandharipande and
Thomas for non-compact Calabi-Yau spaces is introduced based on a
virtual Bialynicki-Birula decomposition with respect to a $C^*$ action
on the stable pair moduli space, or alternatively the equivariant
index of Nekrasov and Okounkov. This effectively calculates the refined
index for $M$-theory reduced on these Calabi-Yau geometries. Based on physical expectations we propose a product
formula for the refined invariants extending the motivic product
formula of Morrison, Mozgovoy, Nagao, and Szendroi for local $P^1$.
We explicitly compute refined invariants in low degree for local $P^2$
and local $P^1$ x $P^1$ and check that they agree with the
predictions of the direct integration of the generalized
holomorphic anomaly and with the product formula. The
modularity of the expressions obtained in the direct
integration approach allows us to relate the generating function
of refined PT invariants on appropriate geometries to
Nekrasov's partition function and a refinement of Chern-Simons
theory on a lens space.  We also relate our product formula to wallcrossing.

\end{titlepage}
\vfill \eject

%%%%%%%%%%%%%%%%%%%%%%%%%%%%%%%%%%%%%%%%%%%%%%%%%%%%%%%%%%%%%

\newpage

\baselineskip=16pt

\tableofcontents
\section{Introduction}

The BPS spectrum and its stability conditions determine
to a large extent the effective action of $N=2$ supersymmetric
theories. In rigid $N=2$ theories in four dimensions one can
define a refined BPS state counting, which records
the multiplicities $N^\Gamma_{j_L,j_R}\in \mathbb{N}$
of BPS particles with charges $\Gamma$ in the lattice
$\Lambda$ of K-theory charges of even $D$-branes and
spin quantum numbers $j_L,j_R$ of the twisted off shell Lorentz
group $\su_L\times \su_R=\mathrm{Spin}(4)$ representations~\cite{GV2,IKV,HKK}.
Many rigid N=2 theories can be constructed
by type II string compactification on a non-compact Calabi-Yau
manifolds $M$, and in fact no other examples are known at present.

The topological A-model of the type II string calculates a BPS index of these
multiplicities for the infinite subset $\Gamma'\in \Lambda'$ of charges with
one unit of D6 brane charge and arbitrary units of D2 and D0-brane
charges~\cite{GV2}.  The BPS index is a weighted sum over the right spins.  Mirror symmetry has been used to calculate
the corresponding generating functions $F(g_s,t)$ for the index
multiplicities in the B-model from the holomorphic anomaly
equations~\cite{BCOV} and appropriate  boundary conditions~\cite{HKQ}
in terms of quasi-modular forms. It was argued in~\cite{HKK} that
in the local limit the topological B-model admits a deformation
of the genus expansion by insertions of the puncture operators of
topological gravity, which captures the refinement and leads to a simple
generalization of the holomorphic anomaly equation~\cite{HK2,KK,KK2}. Together
with generalized gap conditions~\cite{KK,HK2,KK2} it allows the
efficient calculation of the deformed generating functions $F(\epsilon_1,\epsilon_2,t)$
for the $N^{\Gamma'}_{j_L,j_R}\in \mathbb{N}$ in terms of quasi-modular forms.

The purpose of this paper is twofold.  First, we extend the geometric
description of the moduli space of BPS states with charge
$\Gamma'$ given in~\cite{kkv} and extract the refined
multiplicities $N^{\Gamma'}_{j_L,j_R}$ from this description
by purely algebro-geometric methods.  The mathematical description
of the moduli spaces and virtual numbers developed in~\cite{kkv} to describe
the unrefined invariants has been developed
in~\cite{pt,ptbps} using the notion of stable pairs. This notion is closely related to
Donaldson-Thomas invariants and the refined multiplicities
are expected to capture features of the motivic Donaldson-Thomas invariants.
The second purpose is to interpret the physical description of the refined
partition function as a product formula for the refined stable pair invariants,
then checking this description for local $\PP^2$ and local $\PP^1\times\PP^1$.
Our calculations can be thought of as either geometric corroboration of
the refined B-model calculation or as evidence for a mathematical
conjecture, depending on the viewpoint of the reader.   The refined
stable pair invariants can be mathematically defined in terms of
a virtual Bialynicki-Birula decomposition, or equivalently, using
an equivariant index of M-theory \cite{no}.

The rest of the paper is organized as follows.  In Section~\ref{sec:physexp}
we describe the 5-dimensional BPS supertrace in the $\Omega$-deformation
and the Schwinger loop calculation which will allow us to relate the
B-model calculation to the $\su_L\times\su_R$ BPS invariants.  The B-model
calculation itself is done in Section \ref{sec:dirint}, leading to the explicit
calculation of the BPS invariants.  In Section~\ref{sec:enum} we review
the definition of stable pair invariants and review the relationship between
the stable pair partition function and the Gopakumar-Vafa invariants.
In Section~\ref{sec:kkv} we review the paper \cite{kkv} and update
the method, making it more rigorous by using stable pair invariants instead
of the relative Hilbert scheme of a family of curves.  In Section~\ref{pttoric}
we review the localization algorithm of \cite{ptvertex} for the PT
invariants on toric Calabi-Yau threefolds, which we have implemented on
a computer to perform the necessary calculations.  In Section~\ref{refinement}
we have defined the refined PT invariants using both a motivic approach as
well as the equivariant index of Nekrasov and Okounkov.  In
Section~\ref{refinedtobps} we express the refined PT partition function
as an infinite product depending only on the $\su\times\su$ BPS invariants,
compute the low degree terms of the partition function by geometry, and
confirm that they match the B-model calculation.
In Section~\ref{kkvrefined} we use our methods to further update the method
of \cite{kkv}, showing how the $\su\times\su$ BPS invariants can be
calculated by hand in low degree.

%--------------------------------------------------------------------
There is related work on refined invariants for surfaces
(rather than local surfaces) \cite{gs}. While this paper
was finalized there appeared three papers on the arXiv
which address  similar BPS countings, but not with stable
pairs invariants~\cite{Iqbal:2012mt}\cite{Aganagic:2012hs}\cite{Iqbal:2012xm}.
%-------------------------------------------------------------------------

\section{Physical expectations}
\label{sec:physexp}

We consider theories with eight supercharges by
compactifying M-theory on a Calabi-Yau threefold
$X$ to five dimensions or on $X\times S_r^1$ to four
dimensions. We are interested in BPS states giving holomorphic
corrections to couplings in the vector moduli space. Since the
latter decouples from the type II dilaton $\phi_{II}$,
the radius $r\sim (g^{II}_s)^\frac{2}{3}$  of the M-theory
circle with $g^{II}=\exp(\phi_{II})$ is irrelevant for these
corrections. The 5d M-theory and 4d type IIA
descriptions of these BPS states are therefore
expected to be equivalent. At the level of entropy
counts from BPS states  for 4d and 5d black holes
this has been made explicit in~\cite{GSY}.

\subsection{The refined BPS supertrace }
\label{BPSsupertrace}

M-theory reduced on a compact Calabi-Yau threefold $X$
to 5 dimensions gives rise to a 5d supergravity
theory with eight conserved supercharges.
The superalgebra is acted on by a $SU(2)_L\times
SU(2)_R\subset Sp(4)$, which is the little group of the
5d Lorentz group. If gravity can be decoupled,
i.e. a rigid limit of supergravity exists
then there emerges a further $SU(2)_{\cal R}$
symmetry acting on the algebra symmetry group\footnote{In the
supergravity theory there can be a
$U(1)_{\cal R}$ symmetry arising in the infrared,
which is not directly associated to a geometrical
symmetry of space. We would like to thank Greg
Moore for discussions about this point.}.
The corresponding states in 5d are specified by
their BPS mass $M$ or equivalently by an integer charge
vector $\Gamma$ and their spin content given by
a representation $(j_L,j_R)$ of the little group
and their representation under  $SU(2)_{\cal R}$.
Reduction on the circle leads to a four dimensional
theory with $N=2$ supersymmetry  arising from the
IIA reduction on $X$. The superalgebra and the
symmetry acting on it does not change. Only now the 4d
mass gets  shifted by a Kaluza-Klein momentum on
the circle. After this compactification
the  charge lattice of the BPS states is naturally identified
with the $K$-theory charge of the type IIA
$D_{2k}$ branes
\be
\Gamma=(q_0,q_A,p^A,p^0)\in \oplus_{i=0}^3 H^{2i}(M)\ .
\ee
For particles at rest in 4d the eigenvalues of the Hamiltonian $H$ are
the BPS masses $M=|\Gamma \cdot \Pi|$. The vector $\Pi$
is instanton corrected in the type IIA theory, but it can be
mapped by mirror symmetry to the period vector of the
holomorphic $3$-form of the mirror of $X$ and calculated
exactly. The relation between the left spin and the $D_0$ and
brane charge $q_0$ is~\cite{GSY}
\begin{equation}
q_0=2 \frac{j_L}{(p^0)^2} \ .
\end{equation}

Formally one can define a 5d  BPS supertrace
\be
Z_{{\cal BPS}}(\epsilon_L,\epsilon_R,t) =
{\rm Tr}_{{\cal BPS}} (-1)^{2 (J_L+J_R)}
e^{-2 \epsilon_L J_L} e^{-2 \epsilon_R J_R} e^{-2\epsilon_{R} J_{\cal R}} e^{\beta H}\ .
\label{supertrace}
\ee
as a refinement of the {\sl BPS index} $Z_{{\cal BPS}}(\epsilon_L,0,t)$, which was considered in~\cite{GV2}
and shown to reproduce the
holomorphic limit of the topological string partition function
$Z=\exp(F^{{\rm top.\ str.}})$ on the Calabi-Yau space $X$.
Here and in the following we denote by $J_*$ the Cartan
element $J^3_*$ of the $SU(2)_*$ and by $j_*$ an $SU(2)_*$ representation
or the eigenvalue of the Casimir. In (\ref{supertrace}) $t$ stands for all relevant
geometric parameters, see below.
The fact that $Z_{{\cal BPS}}(\epsilon_L,0,t)$
is an index implies in particular that only the left short multiplets contribute,
while the contributions of long left multiplets cancel.
The index is then expected to be invariant under complex structure
deformations of $X$. Geometrical
examples where the right spin assignments of states change
under complex deformations of $X$ while the index does not change are provided by
ruled surfaces over higher genus curves~\cite{kkv}.

It was argued in~\cite{Nekrasov:2002qd}
that under certain assumptions (\ref{supertrace})
is also an index. The main focus of the paper are
the refined multiplicities $N^\Gamma_{j_L,j_R}$ of BPS
states counted by this index (\ref{supertrace}).

In this paper we restrict ourselves to charge vectors
$\Gamma=(n,\beta,0,1)$, where in particular $\beta \in H_2(X,\mathbb{Z})$.
We denote by  $t$ the complexified K\"ahler parameters measuring the mass
of $D_2$ wrapped on holomorphic curves ${\cal C}_\beta$ with complexified volume $t$.
The argument of \cite{Nekrasov:2002qd} relies on further
assumptions that we will discuss in some detail below, see
also~\cite{IKV}\cite{Aganagic:2011mi}\cite{Aganagic:2012au}.

First of all fixing the combination $\epsilon_R
(j'_R + j_{\cal R})=\epsilon_R j_R$ in the trace allows
us to twist the assignment of the right spin content of the
theory with $J_R$, the twisted generator of the Lorentz group.
This makes in particular the eight susy generators (in the 4d language)
transform as a scalar $Q$, a selfdual two form $Q^+_{\mu \nu}$ and
a vector $Q_{\mu}$. As usual $Q$ could define a BRST
cohomology operator on any four manifold, which is however trivial
in $\mathbb{R}^4$. Instead~\cite{Nekrasov:2002qd} based
on earlier work chooses  $\tilde Q =Q+E_\alpha \Omega_{\mu}^{\alpha\ \nu}  x^\mu Q_\nu$,
and considers equivariant cohomology. Here $\Omega^\alpha_{\mu\nu}$
can be the $U(1)_{\epsilon_L}\times U(1)_{\epsilon_R}$ subgroup
of the $SO(4)$ space time rotation group and still defines
an interesting equivariant cohomology. $\tilde Q$ becomes
an equivariant differential on the moduli space of framed
instantons, which are calculated by the Atiyah-Bott localization formula and
provides a formula for the instanton partition function,
known as the Nekrasov partition function.
The argument in~\cite{Nekrasov:2002qd} starts with a supersymmetric
gauge theory in 5d, which we do not require here.
We just assume that the supergravity scale
can be decoupled, which in the geometric engineering
approach means a decompactification limit of $X$.

If the theory has additional symmetries, e.g. flavor
symmetries, one can consider a more general choice of the twisting
$\epsilon_R(j'_R + j_{\tilde {\cal R}})$, where
${\tilde {\cal R}}$ is the ${\cal R}$ symmetry action
accompanied by an $U(1)$ subgroup of the additional
symmetry. Typically these symmetries act on mass parameters and
depending on the charge $q_i$ of the mass parameter $m_i$
under the $U(1)$ one gets a shift in the
mass parameter $m_i\rightarrow m_i + q_i\epsilon_R$~\cite{Nekrasov:2011bc}.

To define (\ref{supertrace}) as a path integral
one has to realize the two twists by $J_*$
geometrically as twisted boundary conditions
for the fermions around a circle in the background
geometry outside of $\mathbb{R}^4$, or its generalizations
discussed below. One circle can be the $M$ theory $S^1$, but
for the second one needs an $S^1$ isometry inside $X$.
This is clearly not possible if $X$ is a compact
Calabi-Yau manifold. For noncompact Calabi-Yau spaces
there is such an $U(1)$ isometry and this is all that we require,
for the equivariant localization in the moduli space
of stable pairs.

Geometries realizing the two twists geometrically
are referred to as $\Omega$-backgrounds. Another
way of describing them~\cite{Nekrasov:2002qd}\cite{Nekrasov:2011bc}
is to start with a $6$-dimensional $N=1$ gauge theory and
compactifying it on a fibration $M_4\rightarrow T^2_r$, where the 4d
space time $M_4$ has at least an $U(1)_1\times U(1)_2$ isometry
and $r$ is the volume of $T^2$. The $\epsilon_{L/R}$ are considered
as complex variables and the bundle is defined by requiring
the flat connections corresponding to the isometries to have
holonomies $(\exp(\frac{r}{2}{\rm Re}\epsilon_R),
\exp(\frac{r}{2}{\rm Re}\epsilon_L))$ and $(\exp(\frac{r}{2}{\rm Im}
\epsilon_R),\exp(\frac{r}{2}{\rm Im}\epsilon_L))$
around the two cycles of the $T^2_r$.

Typical examples with sufficient isometries for the
four dimensional spacetime $M^4$ include Taub-Nut
geometries ${\rm TN}_{p^0}$ in the $M$-theory compactification on
${\rm TN}_{p^0}\times S^1_r\times X$  studied in~\cite{GSY},
which include $\mathbb{R}^4$ for Taub-Nut/D6-brane charge $p^0=1$.
The rotation angles $\epsilon_{1,2}$ of the $U(1)_i$ can
then be identified with rotation angles in the Cartan subalgebra of
$SU(2)_{L/R}$ via $\epsilon_{L/R}=\frac{1}{2}(\epsilon_1\mp \epsilon_2)$,
and $q_{L/R}=\exp(\epsilon_{L/R})$ counts the $\sigma_{L/R}^3$ spin eigenvalues.

If one has only the $M$-theory $S^1$ as in particular for
compact Calabi-Yau spaces one can view  the fibration of
$\mathbb{R}^4_\epsilon$ over this $S^1$ as a Melvin background
for $IIA$\cite{Nekrasov:2010ka}\cite{Nekrasov:2002qd},
metric\footnote{This is easily generalizable to Taub-Nut spaces.}
\begin{equation}
{\rm d} s^2=(dx^\mu + \Omega^\mu {\rm d} \theta )^2 + {\rm d}\theta^2 \ ,
\end{equation}
i.e. it  is characterized by a vev of a selfdual RR 1-form fields
$\Omega$, whose selfdual field  strength has near the
origin the form
\begin{equation}
F_L={\rm d} \Omega =
\epsilon_1 {\rm d } x_1 \wedge  {\rm d } x_2+\epsilon_2 {\rm d }x_3 \wedge  {\rm d }x_4  \ .
\end{equation}
The reduction of 5d fields which are twisted around the
$S^1$. i.e. all the fields contributing to the index
are charged under $F_L$ and give a one loop contributions
to the $F$-term $F_L^{2g} R^2_L$, which was calculated
by the Schwinger loop of~\cite{GV2}.

As explained in~\cite{Nekrasov:2002qd}\cite{Nekrasov:2011bc}
\cite{Nakayama:2011be} the metric for $T^2$ compactification
of the six dimensional theory  is given by
\begin{equation}
{\rm d} s^2=(dx^\mu + \Omega^\mu {\rm d} z + \bar \Omega {\rm d} \bar z  )^2 + {\rm d} z {\rm d} \bar z\ ,
\end{equation}
where  coordinates $(z,\bar z)$ are $T^2$ coordinates. In the $r\rightarrow 0$
limit one ends up with a selfdual $F_L$ and an anti selfdual $F_R$
field strength in 4d, spelled out in~\cite{Nakayama:2011be},
to with the twisted fields in the index couple accordingly.
The generalization of the Schwinger loop calculation~\cite{GV2} to
this coupling is straightforward and will be discussed next.

\subsection{The Schwinger loop calculation}
The Schwinger loop calculation for these F-term couplings expresses
\begin{equation}
F^{hol}(\epsilon_1,\epsilon_2,t)=\sum_{n,g\in \mathbb{Z}} (\epsilon_1+\epsilon_2)^{2 n} (\epsilon_1 \epsilon_2)^{g-1} F^{(n,g)}(t)
\label{Fe1e2}
\end{equation}
in terms of the above BPS trace
\begin{equation}
F^{hol}(\epsilon_1,\epsilon_2,t)=-\int_\epsilon^{\infty} \frac{ds}{s}\frac{\textrm{Tr}_\mathcal{BPS}(-1)^{2J_L+2J'_R}
e^{-s m^2} q_L^{s J_L} q_R^{s J_R}}{4\left(\sinh^2\left(\frac{s\epsilon_L}{2}\right)
- \sinh^2\left(\frac{s \epsilon_R}{2}\right)\right)}\ .
\label{Schwinger}
\end{equation}
Here  $\epsilon_1,\epsilon_2$ is related to the string coupling $g_s$ by
\begin{equation}
\epsilon_1=\sqrt{b} g_s, \qquad  \epsilon_1=-{1\over \sqrt{b}} g_s\
\end{equation}
and we denote by $s$ the deformation parameter $s=(\epsilon_1+\epsilon_2)^2$.

To perform the integral (\ref{Schwinger}) one considers  $M2$ branes
wrapping a curve $C_\beta$ and extending in spacetime with $n$ units
of momentum around the $M$-theory $S^1$-cycle, which leads to
the mass $m^2=\beta\cdot t+2\pi  n$. Since the vector
multiplet moduli space decouples from the $IIA$-dilaton,
one can interpolate between weak and strong coupling and view
the $M_2$ brane as a bound state of a $D2$ wrapping $C_\beta$ with
$n$ $D_0$ branes. Geometrically this corresponds to a
stable pair consisting of
a sheaf $\cF$ on $M$ of pure dimension 1 supported on $C_\beta$
together with a section
$s\in H^0(M,\cF)$ which generates $\cF$ outside a finite number
of points, i.e.\ we have
the topological data
\begin{equation}
{\rm ch}_2(\cF)=\beta, \qquad  \chi(\cF)=n \ .
\label{topinvstablepair}
\end{equation}

By summing over $n$ and using the Poisson resummation formula
$\sum_{n} \exp(-2 \pi i sn )=\sum_{k} \delta(s-k)$  one obtains up to terms
coming from the constant maps at genus 0 and 1 	
\begin{equation}
\begin{array}{rl}
F^{hol}(\epsilon_1,\epsilon_2,t)&=
\displaystyle{\sum_{{j_L,j_R=0}\atop {k=1}}^\infty \sum_{\beta\in H_2(M,\mathbb{Z})} (-1)^{2(J_L+J'_R)} \frac{N^\beta_{j_Lj_R}}{k}
\frac{\displaystyle{\sum_{m_L=-j_L}^{j_L}} q_R^{k m_L}}{2\sinh\left( \frac{k \epsilon_1}{2}\right)} \frac{\displaystyle{\sum_{m_R=-j_R}^{j_R}} q_R^{k m_R}}
{2\sinh\left( \frac{k \epsilon_2}{2}\right)}e^{-k\, \beta \cdot t}}\ .
\end{array}
\label{schwingerloope1e2}
\end{equation}
This expression is correct up to cubic  terms
$a t^3+ b t^2 + c t$ in the K\"ahler parameters multiplying
$g_s^{-2}$ and up to linear classical terms at order $g_s^{0}$
and $\frac{s}{g_s^2}$, related to classical
intersections on $X$. There are also constants terms
$t^0$ at all orders in $s,g_s^2$ obtained by setting
$N^0_{00}=\frac{\chi(X)}{2}$.

The relation between the refined and the unrefined BPS invariants is that the latter
are defined by summing over the $j_R$ spin representation with sign and their multiplicity
\be
\sum_{g=0}^\infty n^g_\beta I_L^g=\sum_{j_+} N^\beta_{j_Lj_R} (-1)^{2j_R}(2j_R+ 1)\left[\frac{j_L}{2}\right]_L  \,,
\ee
and changing the basis for the left spin representations according to
\be
I_*^n=\left(2 [0]_*+\left[\frac{1}{2}\right]_*\right)^{\otimes n}=
\sum_{i}\left(\left(2n\atop n-i\right)- \left(2n  \atop n-i-2\right) \right)\left[\frac{i}{2}\right]_* \, .
\label{defIn}
\ee

In comparing (\ref{Fe1e2}) with (\ref{schwingerloope1e2}) it is convenient to use
the identity
\begin{equation}
{\rm Tr}_{I_*^n}(-1)^{2 J_*} e^{-2 J_* s}=\left(2 {\rm sinh}\left(\frac{s}{2}\right)\right)^{2n}
\end{equation}
and express both the left and the right spin in the $I^n_*$ basis.
This yields invariants $n^\beta_{g_R,g_L}$ which are related to the $N^\beta_{j_R,j_L}$  by
\begin{equation}
\sum_{g_R,g_L} n^\beta_{g_R,g_L} I_R^{g_R}\otimes I_L^{g_L}= \sum_{j_R,j_L}
N^\beta_{j_R,j_L} \left[\frac{j_R}{2}\right]_R \otimes \left[\frac{j_L}{2}\right]_L\ .
\label{basis}
\end{equation}
The geometric interpretation implies that $N^\beta_{j_L j_R}=0$ for $\beta>\beta^{max}(j_L,j_R)$
for finite $\beta^{max}(j_L j_R)$ and the same properties hold for the $n^\Gamma_{g_R,g_L}$.
The  $n^\Gamma_{0,g}$ are the complex structure invariants unrefined BPS invariants.
Both the  $n^\Gamma_{g_R,g_L}$ and the $N^\Gamma_{j_R,j_L}$ are in $\mathbb{Z}$, but
we have the additional property $N^\Gamma_{j_R,j_L}\ge 0$.

Eq. (\ref{schwingerloope1e2}) can be exponentiated to yield the
partition function $Z=e^{F_{hol}(\epsilon_1,\epsilon_2,t)}$,
which has the form~\cite{IKV}
\begin{equation}
Z=\prod_\beta \prod_{j_{L/R}=0}^\infty \prod_{m_{L/R}=-j_{L/R}}^{j_{L/R}}\prod_{m_1,m_2=1}^\infty \left(1-q_L^{m_L} q_R^{m_r} e^{\epsilon_1(m_1-\frac{1}{2})}
e^{\epsilon_2(m_2-\frac{1}{2})} Q^{\beta}\right)^{ (-1)^{2(j_L+j_R)} N^\beta_{j_Lj_R}}\ ,
\label{productrefined}
\end{equation}
where we abbreviated $e^{-\beta \cdot t}=:Q^\beta$.

\section{The direct integration approach}
\label{sec:dirint}

In~\cite{HK2,KK} generalized holomorphic anomaly equations were
proposed\footnote{The one in~\cite{KK} contains an additional
term, which is irrelevant for the present purpose of
counting BPS states.} which take the form
\begin{eqnarray} \label{gen_hol_ano}
\bar{\partial}_{\bar{i}} F^{(n,g)}= \frac{1}{2}\bar{C}_{\bar{i}}^{jk}\big{(}D_jD_kF^{(n,g-1)}
+{\sum_{m,h} }^{\prime}  D_jF^{(m,h)}D_kF^{(n-m,g-h)}\big{)} \,, \quad n+g>1\,,
\end{eqnarray}
where the prime denotes omission of $(m,h)=(0,0)$ and $(m,h)=(n,g)$
in the sum. The first term on the right hand side is set to
zero if $g=0$. These equations together with the modular
invariance of $F^{(n,g)}$ and the  gap boundary conditions
determine~(\ref{Fe1e2}) recursively to any order in
$\epsilon_{1,2}$~\cite{HKK}. The equation (\ref{gen_hol_ano})
has been given a B-model interpretation in the local
limit~\cite{HKK} in which the deformation direction
corresponds to the puncture operator of topological
gravity coupled to the Calabi-Yau non-linear $\sigma$-model.

\subsection{Elliptic curve mirrors and closed modular expressions}

We discuss in the following the simple situation in which the
$B$-model or mirror curve for the non-compact Calabi-Yau
manifold is a family of elliptic curves. This holds for
the mirror curves of non-compact Calabi-Yau manifolds
defined as the anticanonical bundle over del Pezzo surfaces $S$,
i.e. the total space of ${\cal O}(-K_S)\rightarrow S$.

Let us denote the mirror curve ${\cal C}$ in the Weierstrass form as
\begin{eqnarray}
\label{weierstrass}
y^2 = 4x^3 -g_2(u,m) x- g_3(u,m) \ .
\end{eqnarray}
We further denote the holomorphic $(1,0)$ form
$\omega=\frac{d x}{y}$ and the complex
parameter $\tau$ that lives in the upper
halfplane by
\begin{equation}
\tau=\frac{\int_b \omega}{\int_a \omega} \ .
\end{equation}
Here $a,b$ are an integer basis of $H^1({\cal C},\mathbb{Z})$,
$u$ is the complex structure parameter of the curve and
$m$ are isomonodromic deformations. The discriminant reads
$\Delta=g_2^3-27 g_3^2$ and the $j$-function defines an
universal relation between $(u,m)$ and $\tau$
($q=\exp(2 \pi i \tau)$)
\begin{equation}
j=\frac{g_2^3}{\Delta}=\frac{1}{q} + 744 + 196884 q + 21493760 q^2+{\cal O}(q^3)\ .
\end{equation}

The main result of~\cite{HKK} is that the general form of the
higher $F^{(n,g)}$ with $n+g>1$ is is given by
\begin{equation}
\label{generalformfng}
F^{(n,g)}=\frac{1}{\Delta^{2(g+n)-2}(u,m)} \sum_{k=0}^{3g+2n-3} X^k p^{(n,g)}_k(u,m)
\end{equation}
where the $p^{(n,g)}_k(u,m)$ are completely fixed by the
holomorphic anomaly equation and behavior of $F^{(n,g})$ at
the cusp points.
Here we defined the non-holomorphic generator $X$ is by
\begin{eqnarray}
X=\frac{g_3(u,m)}{g_2(u,m)} \frac{\hat{E}_2(\tau)E_4(\tau)}{E_6(\tau)} \, .
\label{Xdef}
\end{eqnarray}
With  $\hat{E}_2$ we denoted the non-holomorphic second Eisenstein
series
\be
\hat{E}_2(\tau, \bar{\tau}) = E_2(\tau) - \frac{3}{\pi {\rm Im}(\tau)} \,.
\ee
The unhatted quantities are the usual holomorphic Eisenstein
series. We note that
\begin{equation}
\frac{E_6^2}{E_4^3}=27\frac{g_3^2}{g_2^3} \ .
\label{identity}
\end{equation}

To prove (\ref{generalformfng}) note that flat coordinate
$t$, which vanishes at a given cusp point can be integrated
from
\begin{eqnarray}
\label{nonlogperiod}
\frac{dt}{du} =\sqrt{\frac{E_6(\tau) g_2(u,m)}{E_4(\tau) g_3(u,m)}}=3^\frac{3}{4}\sqrt[4]{\frac{E_4}{g_2}}  \, .
\end{eqnarray}
Here $\frac{dt}{du}$ is a period of the holomorphic differential $\frac{dx}{y}$ over the vanishing cycle at a nodal singularity
of ${\cal C}$. The period $t(u,m)$ is a period integral of a
meromorphic differential and the constant of the $u$-integration is zero.
$t(u,m)$ can be also determined as the solution of a  third order differential equation in $u$ with polynomial coefficients in  $(m,u)$, see. \cite{HuangKlemm}.

The proof of (\ref{generalformfng})  proceeds by using (\ref{identity}),(\ref{nonlogperiod}) and the Ramanujan relations
\begin{equation}
\begin{array}{rl}
\frac{\rm d}{{\rm d}\tau}E_2 &= \frac{1}{12}(E_2^2-E_4)\ ,\\
\frac{\rm d}{{\rm d}\tau}E_4 &= \frac{1}{3}(E_2 E_4-E_6)\ , \\
\frac{\rm d}{{\rm d}\tau}E_6 &= \frac{1}{2}(E_2 E_6-E_4^2), \\
\end{array}
\end{equation}
to derive
\begin{equation}
\begin{array}{rl}
\frac{\rm d}{{\rm d} t} X &= \frac{1}{\Delta}\frac{{\rm d}u}{{\rm d} t}(A X^2+ B X+ C), \\
\frac{{\rm d}^2u}{{\rm d}^2 t}&= \frac{1}{\Delta}\left(\frac{{\rm d}u}{{\rm d} t}\right)(A X+ \frac{B}{2})\ , \\
\end{array}
\label{closing}
\end{equation}
with
\begin{equation}
A=\frac{9}{4}( 2 g_2 \partial_u g_3- 3 g_3 \partial_u g_2),\qquad B=\frac{1}{2}( g_2^2 \partial_u g_2- 18 g_3 \partial_u g_3),\qquad C=\frac{g_2 A}{3^3}\ .
\end{equation}

Using (\ref{nonlogperiod}) and the fact that the 3-point
function $C_{ttt}=\frac{\del^3 F^{(0,0)}}{\del t^3} = - \frac{2 \pi i}{c_0} \frac{d \tau}{dt}$
is given in terms of the complex modulus $\tau$
of  (\ref{weierstrass}) one can rewrite (\ref{gen_hol_ano}) as
\ban
24 \frac{\partial F^{(n,g)} }{\partial X} &=& c_0 \frac{g_2(u)}{g_3(u)} \frac{E_6}{E_4} \Big[ \left(\frac{d u}{d t}\right)^2 \frac{\del^2 F^{(n,g-1)}}{\del u^2} +\frac{d^2 u}{d t^2} \frac{\del F^{(n,g-1)}}{\del u} \nn \\
&& +\left(\frac{d u}{d t}\right)^2 {\sum_{m,h} }^{\prime} \frac{\del F^{(m,h)}}{\del u} \frac{\del F^{(n-m,g-h)}}{\del u} \Big] \,,   \label{hol_an_X}
\ean
see \cite{HKK} for more details.  It follows by
(\ref{closing},\ref{hol_an_X}) and a  simple inductive
argument that the r.h.s.\ of (\ref{hol_an_X}) is a polynomial
of $X$ of maximal degree  $2(g+n)-3$ and a rational function in
$u$ with denominator $\Delta^{2(g+n)-2}(u)$. Equation (\ref{hol_an_X})
can in particular be used to integrate  the holomorphic anomaly efficiently up to the polynomial $p^{(n,g)}_k(u)$, which is undetermined after the integration.

$F^{(0,0)}= -c_0 \int dt \int dt \tau$ can be determined up to irrelevant
constants from the complex structure $\tau$ in the upper half plane,
which in turn may be calculated using the $j$-function of the elliptic
curve. It remains to describe the boundary conditions  which
fix $p^{(n,g)}_k(u)$ and to provide the remaining initial
data $F^{(1,0)}$ and $F^{(0,1)}$ to complete the recursion
(\ref{generalformfng}).

The boundary conditions for the higher genus invariants are given
by the leading behavior of $F(\epsilon_1,\epsilon_2,t)$
at the nodes of the curve (\ref{weierstrass}). If we denote now specifically
by $t$ the vanishing coordinate at the node under investigation, then the
leading behavior is given by
\begin{eqnarray}  \label{expansionschwinger}
F(s,g_s,t)&=&\int_0^{\infty} \frac{ds}{s}\frac{\exp(-s t)}{4\sinh(s\epsilon_1/2)\sinh(s\epsilon_2/2)} +\mathcal{O} (t^0) \\
&=& \big{[}-\frac{1}{12}+\frac{1}{24} (\epsilon_1+\epsilon_2)^2 (\epsilon_1\epsilon_2)^{-1}\big{]}\log(t) \nonumber \nonumber \\ &&
+ \frac{1}{\epsilon_1\epsilon_2} \sum_{g=0}^\infty \frac{(2g-3)!}{t^{2g-2}}\sum_{m=0}^g \hat B_{2g} \hat B_{2g-2m} \epsilon_1^{2g-2m} \epsilon_2^{2m} +\ldots \nonumber \\ &&
=\big{[}-\frac{1}{12}+\frac{1}{24} s g_s^{-2}\big{]}\log(t)
+ \big{[} -\frac{1}{240}g_s^2+\frac{7}{1440} s-
\frac{7}{5760} s^2g_s^{-2} \big{]} \frac{1}{t^2} \nonumber \\ &&
+ \big{[} \frac{1}{1008}g_s^4-\frac{41}{20160} s g_s^2 +\frac{31}{26880} s^2 -\frac{31}{161280} s^3 g_s^{-2}\big{]} \frac{1}{t^4}   +\mathcal{O} (t^0)  \nonumber\\[ 7 mm]
&& +  \,\,\mbox{contributions to $2(g+n)-2>4$}\, ,  \nonumber
\end{eqnarray}
where $g_s^2= (\epsilon_1 \epsilon_2)$  and $s= (\epsilon_1 + \epsilon_2)^{2}$. Here
$\hat B_{m}=\left(\frac{1}{2^{m-1}}-1\right)\frac{B_m}{m!}$ and the Bernoulli numbers $B_m$
are defined by $t/(e^t-1)=\sum_{m=0}^\infty B_m \frac{t^m}{m!}$. The expansion
(\ref{expansionschwinger}) is simply obtained by evaluating (\ref{Schwinger}) with
the assumption that a single hypermultiplet with mass $m=t$ becomes massless at
the node.

From (\ref{expansionschwinger}) we can read the leading behavior of the $F^{(n,g)}$
\be
F^{(n,g)} = \frac{N^{(n,g)}}{t^{2(g+n)-2}} + \mathcal{O} (t^0).
\ee
For example
\begin{equation}
N^{(2,0)}=-\frac{7}{5760},\quad N^{(1,1)}=\frac{7}{1440},\quad N^{(0,2)}=-\frac{1}{240},\ldots
\label{fngconstant}
\end{equation}

The absence of subleading terms up to order $\mathcal{O}(t^0)$ is
the gap condition, which provides just enough condition to fix
$p^{(n,g)}_k(u)$~\cite{HKK}.  The genus one case, $F^{(0,1)}$ follows from the genus one holomorphic
anomaly equation and the boundary condition at the node in
(\ref{Schwinger}). $F^{(1,0)}$ is purely holomorphic and the simplest
global function compatible with its boundary conditions from
(\ref{Schwinger}) is given below.
\begin{eqnarray}
F^{(0,1)} &=& -\frac{1}{2} \log (G_{u\bar u} | u^a m^b \Delta|^\frac{1}{3})   \label{genus1a}\ , \\
F^{(1,0)} &=& \frac{1}{24}\log (u^c m^d \Delta) \label{genus1b}\ .
\end{eqnarray}
The constants $a,b,c$ and $d$ can be be determined using the known behavior
at large radius.

\subsection{The local Calabi-Yau geometries}

It is convenient to use the language of an abelian
$(2,2)$  gauged linear $\sigma$-model~\cite{Witten:1993yc}
whose vacuum  manifold describes the geometry of the local Calabi-Yau
threefolds $M$ as a symplectic quotient or as a toric variety.
For the cases at hand one considers $r$ chiral fields
$X_i$, $i=1,\ldots, r$ and a gauge group $U(1)^{(1)}
\times \ldots \times  U(1)^{(r-3)}$ under which the fields $X_i$ have
integer charges $Q^{(k)}_i$, $i=1,\ldots, r$, $i=1,\ldots, 3-r$,
subject to the anomaly condition $\sum_i Q^{(k)}_i=0$.
The vacuum manifold parametrized by the vacuum expectations
values $x_i$ of scalar components of the $X_i$ then forms the
local Calabi Yau geometry.  This can be seen as the quotient
manifold of the $x_i$ subject to the $D$-term constraints
$\sum_{i=1}^r Q^{(k)}_i|x_i|^2=r_k$ modded
out by the gauge group,  where $r_k$ are the K\"ahler
moduli which get complexified by  Fayet-Iliopoulos
terms to $t_k=r_k+i\theta_k$. In the geometric phase one has
$r_k>0$.

Equivalently in the standard toric description one describes
$M$ as
\begin{equation}
 M=(\mathbb{C}^r \setminus {\cal S}{\cal R})/(\mathbb{C}^*)^{r-3},
\end{equation}
where ${\cal S}{\cal R}$  is the vanishing locus of the Stanley-Reisner ideal
and the  $(\mathbb{C}^*)'s$ act by $x_k\rightarrow
(\mu^{(k)})^{Q^{(k)}_i} x_k$, $i=1,\ldots,r$, $k=1,\ldots, r-3$.

The mirror geometry $W$ is given by \cite{Hori:2000kt}
\begin{equation}
u v=\sum_{i=1}^r y_i=H(x,y,{\underline u})\ .
\end{equation}
Here the $y_i$ are identified under a $\mathbb{C}^*$ scaling
relation $y_i\mapsto \mu y_i$ and constrained by
$\prod_{i=1}^r y_i^{Q^{(k)}_i}=u_k$.
This allows to reduce to the $x,y$ parameters in
$H(x,y,{\underline u})$. The $u_k$ are complex
deformations of the mirror geometry and mirror
symmetry allows us in particular to explore
the $(\epsilon_1,\epsilon_2)$ refinement of the
topological partition function in non-geometric
phases as well.

For  Calabi-Yau manifolds  ${\cal O}(-K_S)\rightarrow S$
the local mirror geometry can be constructed as decompactification limits
of the mirror $W_c$ of a compact elliptic fibration $M_c$ over $S$. More precisely
in the large radius limit of the elliptic fiber the periods integrals
of the holomorphic $(3,0)$ form over the relevant 3-cycles in $W_c$
become integrals of the meromorphic form
\begin{equation}
\lambda=\log(x) \frac{{\rm d} y}{y}\ .
\label{lambda}
\end{equation}
over 1-cycles in the Riemann surface ${\cal C}$
\begin{equation}
H(x,y,{\underline z})=0\ .
\label{curve}
\end{equation}
More generally the data (\ref{lambda}, \ref{curve})
may serve as the definition of the $B$-model geometry.

In the following we give explicit examples of
local Calabi-Yau threefolds of the type
$M={\cal O}(-K_S)\rightarrow S$, where $S$
is toric.

\subsection{The local Calabi-Yau manifold ${\cal O}(-3)\rightarrow \mathbb{P}^2$}

According to the constructions of local mirror manifolds
reviewed above the $(M,W)$ geometries are described
by the charges $Q^{(i)}_k\in \mathbb{Z}$. For the
${\cal O}(-3)\rightarrow \mathbb{P}^2$ geometry one
has four chiral fields $X_i$ with  $U(1)$ charges
\begin{equation}
 Q=(-3,1,1,1)\ .
\end{equation}

To determine $F(\epsilon_1,\epsilon_2,t)$ for a local
Calabi-Yau geometry with genus one mirror curves all
we have do is to bring the genus one mirror curve
${\cal C}$ given by
\be
H(x,y;z)= y^2 + x y + y + u x^3=0
\label{mirrorp2}
\ee
of the ${\cal O}(-3)\rightarrow \mathbb{P}^2$ geometry
to Weierstrass form (\ref{weierstrass}) with
\be
g_2=3^3(1+24u) \qquad  g_3=3^3(1+36 u+216u^2)\ .
\ee
Further we note that $c_0=9$, $a=7$, $b=-1$.  Using this
information the $F^{(n,g)}$ can be very efficiently calculated as
global sections over the moduli space. In the present case
we are interested in the mirror prediction for the
$A$ model at the large volume point. We obtain
the BPS invariants as this point by equating (\ref{Fe1e2})
in the holomorphic limit ${\rm Im}(\tau)\rightarrow \infty$
with (\ref{Schwinger}). The direct integration method is
very efficient for the one parameter cases. We
calculated the $N_{j_L,j_R}^d$ up to $d=9$.

\begin{table}[h!]
%\begin{sidewaystable}[h]
\centering{{
\begin{tabular}[h]{|c|c|cccccccccccccccccccccccccccccccccccc|}
\hline
d\!&\!\!$j_L\backslash j_R$\!&\!\!0\!&\!\!\!$\frac{1}{2}$\!\!\!&\!\!\!1\!\!\!&\!\!\!$\frac{3}{2}$\!\!\!&\!\!\!2\!\!\!&\!\!\!$\frac{5}{2}$\!\!\!&\!\!\!3\!\!\!&\!\!\!$\frac{7}{2}$\!\!\!&\!\!\!4\!\!\!&
\!\!\!$\frac{9}{2}$\!\!\!&\!\!\!5\!\!\!&\!\!\!$\frac{11}{2}$\!\!\!&\!\!\!6\!\!\!&\!\!\!$\frac{13}{2}$\!\!\!&\!\!\!7\!\!\!&\!\!\!$\frac{15}{2}$\!\!\!&\!\!\!8\!\!\!&\!\!\!$\frac{17}{2}$\!\!\!
&\!\!\!9\!\!\!&\!\!\!$\frac{19}{2}$\!\!\!&\!\!\!10\!\!\!&\!\!\!$\frac{21}{2}$\!\!\!&\!\!\!11\!\!\!&\!\!\!$\frac{23}{2}$\!\!\!&\!\!\!12\!\!\!&\!\!\!$\frac{25}{2}$\!\!\!&\!\!\!13\!\!\!
&\!\!\!$\frac{27}{2}$\!\!\!&\!\!\!14\!\!\!&\!\!\!$\frac{29}{2}$\!\!\!&\!\!\!15\!\!\!&\!\!\!$\frac{31}{2}$\!\!\!&\!\!\!16\!\!\!&\!\!\!$\frac{33}{2}$\!\!\!&\!\!\!17\!\!\!&\!\!\!$\frac{35}{2}$\!\!\!\\
\hline
1\!\!\!&\!\!\!0\!\!\!&\!\!\!\!\!\!&\!\!\!\!\!\!&\!\!\!1\!\!\!&\!\!\!\!\!\!&\!\!\!\!\!\!&\!\!\!\!\!\!&\!\!\!\!\!\!&\!\!\!\!\!\!&\!\!\!\!\!\!&\!\!\!\!\!\!&\!\!\!\!\!\!&\!\!\!\!\!\!&\!\!\!\!\!\!&\!\!\!\!\!\!&\!\!\!\!\!\!&\!\!\!\!\!\!&\!\!\!\!\!\!&\!\!\!\!\!\!&\!\!\!\!\!\!&\!\!\!\!\!\!&\!\!\!\!\!\!&\!\!\!\!\!\!&\!\!\!\!\!\!&\!\!\!\!\!\!&\!\!\!\!\!\!&\!\!\!\!\!\!&\!\!\!\!\!\!&\!\!\!\!\!\!&\!\!\!\!\!\!&\!\!\!\!\!\!&\!\!\!\!\!\!&\!\!\!\!\!\!&\!\!\!\!\!\!&\!\!\!\!\!\!&\!\!\!\!\!\!&\!\!\!\\
\hline
2\!\!\!&\!\!\!0\!\!\!&\!\!\!\!\!\!&\!\!\!\!\!\!&\!\!\!\!\!\!&\!\!\!\!\!\!&\!\!\!\!\!\!&\!\!\!1\!\!\!&\!\!\!\!\!\!&\!\!\!\!\!\!&\!\!\!\!\!\!&\!\!\!\!\!\!&\!\!\!\!\!\!&\!\!\!\!\!\!&\!\!\!\!\!\!&\!\!\!\!\!\!&\!\!\!\!\!\!&\!\!\!\!\!\!&\!\!\!\!\!\!&\!\!\!\!\!\!&\!\!\!\!\!\!&\!\!\!\!\!\!&\!\!\!\!\!\!&\!\!\!\!\!\!&\!\!\!\!\!\!&\!\!\!\!\!\!&\!\!\!\!\!\!&\!\!\!\!\!\!&\!\!\!\!\!\!&\!\!\!\!\!\!&\!\!\!\!\!\!&\!\!\!\!\!\!&\!\!\!\!\!\!&\!\!\!\!\!\!&\!\!\!\!\!\!&\!\!\!\!\!\!&\!\!\!\!\!\!&\!\!\!\\
\hline
3\!\!\!&\!\!\!0\!\!\!&\!\!\!\!\!\!&\!\!\!\!\!\!&\!\!\!\!\!\!&\!\!\!\!\!\!&\!\!\!\!\!\!&\!\!\!\!\!\!&\!\!\!1\!\!\!&\!\!\!\!\!\!&\!\!\!\!\!\!&\!\!\!\!\!\!&\!\!\!\!\!\!&\!\!\!\!\!\!&\!\!\!\!\!\!&\!\!\!\!\!\!&\!\!\!\!\!\!&\!\!\!\!\!\!&\!\!\!\!\!\!&\!\!\!\!\!\!&\!\!\!\!\!\!&\!\!\!\!\!\!&\!\!\!\!\!\!&\!\!\!\!\!\!&\!\!\!\!\!\!&\!\!\!\!\!\!&\!\!\!\!\!\!&\!\!\!\!\!\!&\!\!\!\!\!\!&\!\!\!\!\!\!&\!\!\!\!\!\!&\!\!\!\!\!\!&\!\!\!\!\!\!&\!\!\!\!\!\!&\!\!\!\!\!\!&\!\!\!\!\!\!&\!\!\!\!\!\!&\!\!\!\\
\!\!\!&\!\!\!$\frac{1}{2}$\!\!\!&\!\!\!\!\!\!&\!\!\!\!\!\!&\!\!\!\!\!\!&\!\!\!\!\!\!&\!\!\!\!\!\!&\!\!\!\!\!\!&\!\!\!\!\!\!&\!\!\!\!\!\!&\!\!\!\!\!\!&\!\!\!1\!\!\!&\!\!\!\!\!\!&\!\!\!\!\!\!&\!\!\!\!\!\!&\!\!\!\!\!\!&\!\!\!\!\!\!&\!\!\!\!\!\!&\!\!\!\!\!\!&\!\!\!\!\!\!&\!\!\!\!\!\!&\!\!\!\!\!\!&\!\!\!\!\!\!&\!\!\!\!\!\!&\!\!\!\!\!\!&\!\!\!\!\!\!&\!\!\!\!\!\!&\!\!\!\!\!\!&\!\!\!\!\!\!&\!\!\!\!\!\!&\!\!\!\!\!\!&\!\!\!\!\!\!&\!\!\!\!\!\!&\!\!\!\!\!\!&\!\!\!\!\!\!&\!\!\!\!\!\!&\!\!\!\!\!\!&\!\!\!\\
\hline
4\!\!\!&\!\!\!0\!\!\!&\!\!\!\!\!\!&\!\!\!\!\!\!&\!\!\!\!\!\!&\!\!\!\!\!\!&\!\!\!\!\!\!&\!\!\!1\!\!\!&\!\!\!\!\!\!&\!\!\!\!\!\!&\!\!\!\!\!\!&\!\!\!1\!\!\!&\!\!\!\!\!\!&\!\!\!\!\!\!&\!\!\!\!\!\!&\!\!\!1\!\!\!&\!\!\!\!\!\!&\!\!\!\!\!\!&\!\!\!\!\!\!&\!\!\!\!\!\!&\!\!\!\!\!\!&\!\!\!\!\!\!&\!\!\!\!\!\!&\!\!\!\!\!\!&\!\!\!\!\!\!&\!\!\!\!\!\!&\!\!\!\!\!\!&\!\!\!\!\!\!&\!\!\!\!\!\!&\!\!\!\!\!\!&\!\!\!\!\!\!&\!\!\!\!\!\!&\!\!\!\!\!\!&\!\!\!\!\!\!&\!\!\!\!\!\!&\!\!\!\!\!\!&\!\!\!\!\!\!&\!\!\!\\
\!\!\!&\!\!\!$\frac{1}{2}$\!\!\!&\!\!\!\!\!\!&\!\!\!\!\!\!&\!\!\!\!\!\!&\!\!\!\!\!\!&\!\!\!\!\!\!&\!\!\!\!\!\!&\!\!\!\!\!\!&\!\!\!\!\!\!&\!\!\!1\!\!\!&\!\!\!\!\!\!&\!\!\!1\!\!\!&\!\!\!\!\!\!&\!\!\!1\!\!\!&\!\!\!\!\!\!&\!\!\!\!\!\!&\!\!\!\!\!\!&\!\!\!\!\!\!&\!\!\!\!\!\!&\!\!\!\!\!\!&\!\!\!\!\!\!&\!\!\!\!\!\!&\!\!\!\!\!\!&\!\!\!\!\!\!&\!\!\!\!\!\!&\!\!\!\!\!\!&\!\!\!\!\!\!&\!\!\!\!\!\!&\!\!\!\!\!\!&\!\!\!\!\!\!&\!\!\!\!\!\!&\!\!\!\!\!\!&\!\!\!\!\!\!&\!\!\!\!\!\!&\!\!\!\!\!\!&\!\!\!\!\!\!&\!\!\!\\
\!\!\!&\!\!\!2\!\!\!&\!\!\!\!\!\!&\!\!\!\!\!\!&\!\!\!\!\!\!&\!\!\!\!\!\!&\!\!\!\!\!\!&\!\!\!\!\!\!&\!\!\!\!\!\!&\!\!\!\!\!\!&\!\!\!\!\!\!&\!\!\!\!\!\!&\!\!\!\!\!\!&\!\!\!1\!\!\!&\!\!\!\!\!\!&\!\!\!\!\!\!&\!\!\!\!\!\!&\!\!\!\!\!\!&\!\!\!\!\!\!&\!\!\!\!\!\!&\!\!\!\!\!\!&\!\!\!\!\!\!&\!\!\!\!\!\!&\!\!\!\!\!\!&\!\!\!\!\!\!&\!\!\!\!\!\!&\!\!\!\!\!\!&\!\!\!\!\!\!&\!\!\!\!\!\!&\!\!\!\!\!\!&\!\!\!\!\!\!&\!\!\!\!\!\!&\!\!\!\!\!\!&\!\!\!\!\!\!&\!\!\!\!\!\!&\!\!\!\!\!\!&\!\!\!\!\!\!&\!\!\!\\
\!\!\!&\!\!\!$\frac{3}{2}$\!\!\!&\!\!\!\!\!\!&\!\!\!\!\!\!&\!\!\!\!\!\!&\!\!\!\!\!\!&\!\!\!\!\!\!&\!\!\!\!\!\!&\!\!\!\!\!\!&\!\!\!\!\!\!&\!\!\!\!\!\!&\!\!\!\!\!\!&\!\!\!\!\!\!&\!\!\!\!\!\!&\!\!\!\!\!\!&\!\!\!\!\!\!&\!\!\!1\!\!\!&\!\!\!\!\!\!&\!\!\!\!\!\!&\!\!\!\!\!\!&\!\!\!\!\!\!&\!\!\!\!\!\!&\!\!\!\!\!\!&\!\!\!\!\!\!&\!\!\!\!\!\!&\!\!\!\!\!\!&\!\!\!\!\!\!&\!\!\!\!\!\!&\!\!\!\!\!\!&\!\!\!\!\!\!&\!\!\!\!\!\!&\!\!\!\!\!\!&\!\!\!\!\!\!&\!\!\!\!\!\!&\!\!\!\!\!\!&\!\!\!\!\!\!&\!\!\!\!\!\!&\!\!\!\\
\hline
5\!\!\!&\!\!\!0\!\!\!&\!\!\!\!\!\!&\!\!\!\!\!\!&\!\!\!1\!\!\!&\!\!\!\!\!\!&\!\!\!\!\!\!&\!\!\!\!\!\!&\!\!\!1\!\!\!&\!\!\!\!\!\!&\!\!\!1\!\!\!&\!\!\!\!\!\!&\!\!\!2\!\!\!&\!\!\!\!\!\!&\!\!\!2\!\!\!&\!\!\!\!\!\!&\!\!\!2\!\!\!&\!\!\!\!\!\!&\!\!\!1\!\!\!&\!\!\!\!\!\!&\!\!\!\!\!\!&\!\!\!\!\!\!&\!\!\!\!\!\!&\!\!\!\!\!\!&\!\!\!\!\!\!&\!\!\!\!\!\!&\!\!\!\!\!\!&\!\!\!\!\!\!&\!\!\!\!\!\!&\!\!\!\!\!\!&\!\!\!\!\!\!&\!\!\!\!\!\!&\!\!\!\!\!\!&\!\!\!\!\!\!&\!\!\!\!\!\!&\!\!\!\!\!\!&\!\!\!\!\!\!&\!\!\!\\
\!\!\!&\!\!\!$\frac{1}{2}$\!\!\!&\!\!\!\!\!\!&\!\!\!\!\!\!&\!\!\!\!\!\!&\!\!\!\!\!\!&\!\!\!\!\!\!&\!\!\!1\!\!\!&\!\!\!\!\!\!&\!\!\!1\!\!\!&\!\!\!\!\!\!&\!\!\!2\!\!\!&\!\!\!\!\!\!&\!\!\!2\!\!\!&\!\!\!\!\!\!&\!\!\!3\!\!\!&\!\!\!\!\!\!&\!\!\!2\!\!\!&\!\!\!\!\!\!&\!\!\!1\!\!\!&\!\!\!\!\!\!&\!\!\!\!\!\!&\!\!\!\!\!\!&\!\!\!\!\!\!&\!\!\!\!\!\!&\!\!\!\!\!\!&\!\!\!\!\!\!&\!\!\!\!\!\!&\!\!\!\!\!\!&\!\!\!\!\!\!&\!\!\!\!\!\!&\!\!\!\!\!\!&\!\!\!\!\!\!&\!\!\!\!\!\!&\!\!\!\!\!\!&\!\!\!\!\!\!&\!\!\!\!\!\!&\!\!\!\\
\!\!\!&\!\!\!1\!\!\!&\!\!\!\!\!\!&\!\!\!\!\!\!&\!\!\!\!\!\!&\!\!\!\!\!\!&\!\!\!\!\!\!&\!\!\!\!\!\!&\!\!\!\!\!\!&\!\!\!\!\!\!&\!\!\!1\!\!\!&\!\!\!\!\!\!&\!\!\!1\!\!\!&\!\!\!\!\!\!&\!\!\!2\!\!\!&\!\!\!\!\!\!&\!\!\!2\!\!\!&\!\!\!\!\!\!&\!\!\!2\!\!\!&\!\!\!\!\!\!&\!\!\!1\!\!\!&\!\!\!\!\!\!&\!\!\!\!\!\!&\!\!\!\!\!\!&\!\!\!\!\!\!&\!\!\!\!\!\!&\!\!\!\!\!\!&\!\!\!\!\!\!&\!\!\!\!\!\!&\!\!\!\!\!\!&\!\!\!\!\!\!&\!\!\!\!\!\!&\!\!\!\!\!\!&\!\!\!\!\!\!&\!\!\!\!\!\!&\!\!\!\!\!\!&\!\!\!\!\!\!&\!\!\!\\
\!\!\!&\!\!\!$\frac{3}{2}$\!\!\!&\!\!\!\!\!\!&\!\!\!\!\!\!&\!\!\!\!\!\!&\!\!\!\!\!\!&\!\!\!\!\!\!&\!\!\!\!\!\!&\!\!\!\!\!\!&\!\!\!\!\!\!&\!\!\!\!\!\!&\!\!\!\!\!\!&\!\!\!\!\!\!&\!\!\!1\!\!\!&\!\!\!\!\!\!&\!\!\!1\!\!\!&\!\!\!\!\!\!&\!\!\!2\!\!\!&\!\!\!\!\!\!&\!\!\!1\!\!\!&\!\!\!\!\!\!&\!\!\!1\!\!\!&\!\!\!\!\!\!&\!\!\!\!\!\!&\!\!\!\!\!\!&\!\!\!\!\!\!&\!\!\!\!\!\!&\!\!\!\!\!\!&\!\!\!\!\!\!&\!\!\!\!\!\!&\!\!\!\!\!\!&\!\!\!\!\!\!&\!\!\!\!\!\!&\!\!\!\!\!\!&\!\!\!\!\!\!&\!\!\!\!\!\!&\!\!\!\!\!\!&\!\!\!\\
\!\!\!&\!\!\!2\!\!\!&\!\!\!\!\!\!&\!\!\!\!\!\!&\!\!\!\!\!\!&\!\!\!\!\!\!&\!\!\!\!\!\!&\!\!\!\!\!\!&\!\!\!\!\!\!&\!\!\!\!\!\!&\!\!\!\!\!\!&\!\!\!\!\!\!&\!\!\!\!\!\!&\!\!\!\!\!\!&\!\!\!\!\!\!&\!\!\!\!\!\!&\!\!\!1\!\!\!&\!\!\!\!\!\!&\!\!\!1\!\!\!&\!\!\!\!\!\!&\!\!\!1\!\!\!&\!\!\!\!\!\!&\!\!\!\!\!\!&\!\!\!\!\!\!&\!\!\!\!\!\!&\!\!\!\!\!\!&\!\!\!\!\!\!&\!\!\!\!\!\!&\!\!\!\!\!\!&\!\!\!\!\!\!&\!\!\!\!\!\!&\!\!\!\!\!\!&\!\!\!\!\!\!&\!\!\!\!\!\!&\!\!\!\!\!\!&\!\!\!\!\!\!&\!\!\!\!\!\!&\!\!\!\\
\!\!\!&\!\!\!$\frac{5}{2}$\!\!\!&\!\!\!\!\!\!&\!\!\!\!\!\!&\!\!\!\!\!\!&\!\!\!\!\!\!&\!\!\!\!\!\!&\!\!\!\!\!\!&\!\!\!\!\!\!&\!\!\!\!\!\!&\!\!\!\!\!\!&\!\!\!\!\!\!&\!\!\!\!\!\!&\!\!\!\!\!\!&\!\!\!\!\!\!&\!\!\!\!\!\!&\!\!\!\!\!\!&\!\!\!\!\!\!&\!\!\!\!\!\!&\!\!\!1\!\!\!&\!\!\!\!\!\!&\!\!\!\!\!\!&\!\!\!\!\!\!&\!\!\!\!\!\!&\!\!\!\!\!\!&\!\!\!\!\!\!&\!\!\!\!\!\!&\!\!\!\!\!\!&\!\!\!\!\!\!&\!\!\!\!\!\!&\!\!\!\!\!\!&\!\!\!\!\!\!&\!\!\!\!\!\!&\!\!\!\!\!\!&\!\!\!\!\!\!&\!\!\!\!\!\!&\!\!\!\!\!\!&\!\!\!\\
\!\!\!&\!\!\!3\!\!\!&\!\!\!\!\!\!&\!\!\!\!\!\!&\!\!\!\!\!\!&\!\!\!\!\!\!&\!\!\!\!\!\!&\!\!\!\!\!\!&\!\!\!\!\!\!&\!\!\!\!\!\!&\!\!\!\!\!\!&\!\!\!\!\!\!&\!\!\!\!\!\!&\!\!\!\!\!\!&\!\!\!\!\!\!&\!\!\!\!\!\!&\!\!\!\!\!\!&\!\!\!\!\!\!&\!\!\!\!\!\!&\!\!\!\!\!\!&\!\!\!\!\!\!&\!\!\!\!\!\!&\!\!\!1\!\!\!&\!\!\!\!\!\!&\!\!\!\!\!\!&\!\!\!\!\!\!&\!\!\!\!\!\!&\!\!\!\!\!\!&\!\!\!\!\!\!&\!\!\!\!\!\!&\!\!\!\!\!\!&\!\!\!\!\!\!&\!\!\!\!\!\!&\!\!\!\!\!\!&\!\!\!\!\!\!&\!\!\!\!\!\!&\!\!\!\!\!\!&\!\!\!\\
\hline
6\!\!\!&\!\!\!0\!\!\!&\!\!\!\!\!\!&\!\!\!1\!\!\!&\!\!\!\!\!\!&\!\!\!1\!\!\!&\!\!\!\!\!\!&\!\!\!3\!\!\!&\!\!\!\!\!\!&\!\!\!2\!\!\!&\!\!\!\!\!\!&\!\!\!6\!\!\!&\!\!\!\!\!\!&\!\!\!4\!\!\!&\!\!\!\!\!\!&\!\!\!8\!\!\!&\!\!\!\!\!\!&\!\!\!5\!\!\!&\!\!\!\!\!\!&\!\!\!7\!\!\!&\!\!\!\!\!\!&\!\!\!2\!\!\!&\!\!\!\!\!\!&\!\!\!2\!\!\!&\!\!\!\!\!\!&\!\!\!\!\!\!&\!\!\!\!\!\!&\!\!\!\!\!\!&\!\!\!\!\!\!&\!\!\!\!\!\!&\!\!\!\!\!\!&\!\!\!\!\!\!&\!\!\!\!\!\!&\!\!\!\!\!\!&\!\!\!\!\!\!&\!\!\!\!\!\!&\!\!\!\!\!\!&\!\!\!\\
\!\!\!&\!\!\!$\frac{1}{2}$\!\!\!&\!\!\!\!\!\!&\!\!\!\!\!\!&\!\!\!1\!\!\!&\!\!\!\!\!\!&\!\!\!2\!\!\!&\!\!\!\!\!\!&\!\!\!3\!\!\!&\!\!\!\!\!\!&\!\!\!5\!\!\!&\!\!\!\!\!\!&\!\!\!6\!\!\!&\!\!\!\!\!\!&\!\!\!9\!\!\!&\!\!\!\!\!\!&\!\!\!9\!\!\!&\!\!\!\!\!\!&\!\!\!10\!\!\!&\!\!\!\!\!\!&\!\!\!7\!\!\!&\!\!\!\!\!\!&\!\!\!5\!\!\!&\!\!\!\!\!\!&\!\!\!1\!\!\!&\!\!\!\!\!\!&\!\!\!1\!\!\!&\!\!\!\!\!\!&\!\!\!\!\!\!&\!\!\!\!\!\!&\!\!\!\!\!\!&\!\!\!\!\!\!&\!\!\!\!\!\!&\!\!\!\!\!\!&\!\!\!\!\!\!&\!\!\!\!\!\!&\!\!\!\!\!\!&\!\!\!\\
\!\!\!&\!\!\!1\!\!\!&\!\!\!\!\!\!&\!\!\!\!\!\!&\!\!\!\!\!\!&\!\!\!1\!\!\!&\!\!\!\!\!\!&\!\!\!1\!\!\!&\!\!\!\!\!\!&\!\!\!3\!\!\!&\!\!\!\!\!\!&\!\!\!3\!\!\!&\!\!\!\!\!\!&\!\!\!7\!\!\!&\!\!\!\!\!\!&\!\!\!7\!\!\!&\!\!\!\!\!\!&\!\!\!11\!\!\!&\!\!\!\!\!\!&\!\!\!9\!\!\!&\!\!\!\!\!\!&\!\!\!9\!\!\!&\!\!\!\!\!\!&\!\!\!4\!\!\!&\!\!\!\!\!\!&\!\!\!2\!\!\!&\!\!\!\!\!\!&\!\!\!\!\!\!&\!\!\!\!\!\!&\!\!\!\!\!\!&\!\!\!\!\!\!&\!\!\!\!\!\!&\!\!\!\!\!\!&\!\!\!\!\!\!&\!\!\!\!\!\!&\!\!\!\!\!\!&\!\!\!\!\!\!&\!\!\!\\
\!\!\!&\!\!\!$\frac{3}{2}$\!\!\!&\!\!\!\!\!\!&\!\!\!\!\!\!&\!\!\!\!\!\!&\!\!\!\!\!\!&\!\!\!\!\!\!&\!\!\!\!\!\!&\!\!\!1\!\!\!&\!\!\!\!\!\!&\!\!\!1\!\!\!&\!\!\!\!\!\!&\!\!\!3\!\!\!&\!\!\!\!\!\!&\!\!\!4\!\!\!&\!\!\!\!\!\!&\!\!\!7\!\!\!&\!\!\!\!\!\!&\!\!\!7\!\!\!&\!\!\!\!\!\!&\!\!\!10\!\!\!&\!\!\!\!\!\!&\!\!\!6\!\!\!&\!\!\!\!\!\!&\!\!\!4\!\!\!&\!\!\!\!\!\!&\!\!\!\!\!\!&\!\!\!\!\!\!&\!\!\!\!\!\!&\!\!\!\!\!\!&\!\!\!\!\!\!&\!\!\!\!\!\!&\!\!\!\!\!\!&\!\!\!\!\!\!&\!\!\!\!\!\!&\!\!\!\!\!\!&\!\!\!\!\!\!&\!\!\!\\
\!\!\!&\!\!\!2\!\!\!&\!\!\!\!\!\!&\!\!\!\!\!\!&\!\!\!\!\!\!&\!\!\!\!\!\!&\!\!\!\!\!\!&\!\!\!\!\!\!&\!\!\!\!\!\!&\!\!\!\!\!\!&\!\!\!\!\!\!&\!\!\!1\!\!\!&\!\!\!\!\!\!&\!\!\!1\!\!\!&\!\!\!\!\!\!&\!\!\!3\!\!\!&\!\!\!\!\!\!&\!\!\!4\!\!\!&\!\!\!\!\!\!&\!\!\!7\!\!\!&\!\!\!\!\!\!&\!\!\!6\!\!\!&\!\!\!\!\!\!&\!\!\!6\!\!\!&\!\!\!\!\!\!&\!\!\!2\!\!\!&\!\!\!\!\!\!&\!\!\!1\!\!\!&\!\!\!\!\!\!&\!\!\!\!\!\!&\!\!\!\!\!\!&\!\!\!\!\!\!&\!\!\!\!\!\!&\!\!\!\!\!\!&\!\!\!\!\!\!&\!\!\!\!\!\!&\!\!\!\!\!\!&\!\!\!\\
\!\!\!&\!\!\!$\frac{5}{2}$\!\!\!&\!\!\!\!\!\!&\!\!\!\!\!\!&\!\!\!\!\!\!&\!\!\!\!\!\!&\!\!\!\!\!\!&\!\!\!\!\!\!&\!\!\!\!\!\!&\!\!\!\!\!\!&\!\!\!\!\!\!&\!\!\!\!\!\!&\!\!\!\!\!\!&\!\!\!\!\!\!&\!\!\!1\!\!\!&\!\!\!\!\!\!&\!\!\!1\!\!\!&\!\!\!\!\!\!&\!\!\!3\!\!\!&\!\!\!\!\!\!&\!\!\!3\!\!\!&\!\!\!\!\!\!&\!\!\!5\!\!\!&\!\!\!\!\!\!&\!\!\!3\!\!\!&\!\!\!\!\!\!&\!\!\!2\!\!\!&\!\!\!\!\!\!&\!\!\!\!\!\!&\!\!\!\!\!\!&\!\!\!\!\!\!&\!\!\!\!\!\!&\!\!\!\!\!\!&\!\!\!\!\!\!&\!\!\!\!\!\!&\!\!\!\!\!\!&\!\!\!\!\!\!&\!\!\!\\
\!\!\!&\!\!\!3\!\!\!&\!\!\!\!\!\!&\!\!\!\!\!\!&\!\!\!\!\!\!&\!\!\!\!\!\!&\!\!\!\!\!\!&\!\!\!\!\!\!&\!\!\!\!\!\!&\!\!\!\!\!\!&\!\!\!\!\!\!&\!\!\!\!\!\!&\!\!\!\!\!\!&\!\!\!\!\!\!&\!\!\!\!\!\!&\!\!\!\!\!\!&\!\!\!\!\!\!&\!\!\!1\!\!\!&\!\!\!\!\!\!&\!\!\!1\!\!\!&\!\!\!\!\!\!&\!\!\!3\!\!\!&\!\!\!\!\!\!&\!\!\!3\!\!\!&\!\!\!\!\!\!&\!\!\!3\!\!\!&\!\!\!\!\!\!&\!\!\!1\!\!\!&\!\!\!\!\!\!&\!\!\!\!\!\!&\!\!\!\!\!\!&\!\!\!\!\!\!&\!\!\!\!\!\!&\!\!\!\!\!\!&\!\!\!\!\!\!&\!\!\!\!\!\!&\!\!\!\!\!\!&\!\!\!\\
\!\!\!&\!\!\!$\frac{7}{2}$\!\!\!&\!\!\!\!\!\!&\!\!\!\!\!\!&\!\!\!\!\!\!&\!\!\!\!\!\!&\!\!\!\!\!\!&\!\!\!\!\!\!&\!\!\!\!\!\!&\!\!\!\!\!\!&\!\!\!\!\!\!&\!\!\!\!\!\!&\!\!\!\!\!\!&\!\!\!\!\!\!&\!\!\!\!\!\!&\!\!\!\!\!\!&\!\!\!\!\!\!&\!\!\!\!\!\!&\!\!\!\!\!\!&\!\!\!\!\!\!&\!\!\!1\!\!\!&\!\!\!\!\!\!&\!\!\!1\!\!\!&\!\!\!\!\!\!&\!\!\!2\!\!\!&\!\!\!\!\!\!&\!\!\!1\!\!\!&\!\!\!\!\!\!&\!\!\!1\!\!\!&\!\!\!\!\!\!&\!\!\!\!\!\!&\!\!\!\!\!\!&\!\!\!\!\!\!&\!\!\!\!\!\!&\!\!\!\!\!\!&\!\!\!\!\!\!&\!\!\!\!\!\!&\!\!\!\\
\!\!\!&\!\!\!\!4\!\!&\!\!\!\!\!\!&\!\!\!\!\!\!&\!\!\!\!\!\!&\!\!\!\!\!\!&\!\!\!\!\!\!&\!\!\!\!\!\!&\!\!\!\!\!\!&\!\!\!\!\!\!&\!\!\!\!\!\!&\!\!\!\!\!\!&\!\!\!\!\!\!&\!\!\!\!\!\!&\!\!\!\!\!\!&\!\!\!\!\!\!&\!\!\!\!\!\!&\!\!\!\!\!\!&\!\!\!\!\!\!&\!\!\!\!\!\!&\!\!\!\!\!\!&\!\!\!\!\!\!&\!\!\!\!\!\!&\!\!\!1\!\!\!&\!\!\!\!\!\!&\!\!\!1\!\!\!&\!\!\!\!\!\!&\!\!\!1\!\!\!&\!\!\!\!\!\!&\!\!\!\!\!\!&\!\!\!\!\!\!&\!\!\!\!\!\!&\!\!\!\!\!\!&\!\!\!\!\!\!&\!\!\!\!\!\!&\!\!\!\!\!\!&\!\!\!\!\!\!&\!\!\!\\
\!\!\!&\!\!\!$\frac{9}{2}$\!\!\!&\!\!\!\!\!\!&\!\!\!\!\!\!&\!\!\!\!\!\!&\!\!\!\!\!\!&\!\!\!\!\!\!&\!\!\!\!\!\!&\!\!\!\!\!\!&\!\!\!\!\!\!&\!\!\!\!\!\!&\!\!\!\!\!\!&\!\!\!\!\!\!&\!\!\!\!\!\!&\!\!\!\!\!\!&\!\!\!\!\!\!&\!\!\!\!\!\!&\!\!\!\!\!\!&\!\!\!\!\!\!&\!\!\!\!\!\!&\!\!\!\!\!\!&\!\!\!\!\!\!&\!\!\!\!\!\!&\!\!\!\!\!\!&\!\!\!\!\!\!&\!\!\!\!\!\!&\!\!\!1\!\!\!&\!\!\!\!\!\!&\!\!\!\!\!\!&\!\!\!\!\!\!&\!\!\!\!\!\!&\!\!\!\!\!\!&\!\!\!\!\!\!&\!\!\!\!\!\!&\!\!\!\!\!\!&\!\!\!\!\!\!&\!\!\!\!\!\!&\!\!\!\\
\!\!\!&\!\!\!5\!\!\!&\!\!\!\!\!\!&\!\!\!\!\!\!&\!\!\!\!\!\!&\!\!\!\!\!\!&\!\!\!\!\!\!&\!\!\!\!\!\!&\!\!\!\!\!\!&\!\!\!\!\!\!&\!\!\!\!\!\!&\!\!\!\!\!\!&\!\!\!\!\!\!&\!\!\!\!\!\!&\!\!\!\!\!\!&\!\!\!\!\!\!&\!\!\!\!\!\!&\!\!\!\!\!\!&\!\!\!\!\!\!&\!\!\!\!\!\!&\!\!\!\!\!\!&\!\!\!\!\!\!&\!\!\!\!\!\!&\!\!\!\!\!\!&\!\!\!\!\!\!&\!\!\!\!\!\!&\!\!\!\!\!\!&\!\!\!\!\!\!&\!\!\!\!\!\!&\!\!\!1\!\!\!&\!\!\!\!\!\!&\!\!\!\!\!\!&\!\!\!\!\!\!&\!\!\!\!\!\!&\!\!\!\!\!\!&\!\!\!\!\!\!&\!\!\!\!\!\!&\!\!\!\\
\hline
7\!\!\!&\!\!\!0\!\!\!&\!\!\!\!\!\!&\!\!\!\!\!\!&\!\!\!6\!\!\!&\!\!\!\!\!\!&\!\!\!6\!\!\!&\!\!\!\!\!\!&\!\!\!12\!\!\!&\!\!\!\!\!\!&\!\!\!13\!\!\!&\!\!\!\!\!\!&\!\!\!19\!\!\!&\!\!\!\!\!\!&\!\!\!21\!\!\!&\!\!\!\!\!\!&\!\!\!26\!\!\!&\!\!\!\!\!\!&\!\!\!26\!\!\!&\!\!\!\!\!\!&\!\!\!26\!\!\!&\!\!\!\!\!\!&\!\!\!22\!\!\!&\!\!\!\!\!\!&\!\!\!15\!\!\!&\!\!\!\!\!\!&\!\!\!9\!\!\!&\!\!\!\!\!\!&\!\!\!4\!\!\!&\!\!\!\!\!\!&\!\!\!2\!\!\!&\!\!\!\!\!\!&\!\!\!\!\!\!&\!\!\!\!\!\!&\!\!\!\!\!\!&\!\!\!\!\!\!&\!\!\!\!\!\!&\!\!\!\\
\!\!\!&\!\!\!$\frac{1}{2}$\!\!\!&\!\!\!\!\!\!&\!\!\!4\!\!\!&\!\!\!\!\!\!&\!\!\!7\!\!\!&\!\!\!\!\!\!&\!\!\!12\!\!\!&\!\!\!\!\!\!&\!\!\!17\!\!\!&\!\!\!\!\!\!&\!\!\!24\!\!\!&\!\!\!\!\!\!&\!\!\!29\!\!\!&\!\!\!\!\!\!&\!\!\!37\!\!\!&\!\!\!\!\!\!&\!\!\!41\!\!\!&\!\!\!\!\!\!&\!\!\!45\!\!\!&\!\!\!\!\!\!&\!\!\!41\!\!\!&\!\!\!\!\!\!&\!\!\!35\!\!\!&\!\!\!\!\!\!&\!\!\!23\!\!\!&\!\!\!\!\!\!&\!\!\!13\!\!\!&\!\!\!\!\!\!&\!\!\!5\!\!\!&\!\!\!\!\!\!&\!\!\!1\!\!\!&\!\!\!\!\!\!&\!\!\!\!\!\!&\!\!\!\!\!\!&\!\!\!\!\!\!&\!\!\!\!\!\!&\!\!\!\\
\!\!\!&\!\!\!1\!\!\!&\!\!\!2\!\!\!&\!\!\!\!\!\!&\!\!\!3\!\!\!&\!\!\!\!\!\!&\!\!\!8\!\!\!&\!\!\!\!\!\!&\!\!\!11\!\!\!&\!\!\!\!\!\!&\!\!\!18\!\!\!&\!\!\!\!\!\!&\!\!\!23\!\!\!&\!\!\!\!\!\!&\!\!\!33\!\!\!&\!\!\!\!\!\!&\!\!\!40\!\!\!&\!\!\!\!\!\!&\!\!\!48\!\!\!&\!\!\!\!\!\!&\!\!\!50\!\!\!&\!\!\!\!\!\!&\!\!\!49\!\!\!&\!\!\!\!\!\!&\!\!\!39\!\!\!&\!\!\!\!\!\!&\!\!\!25\!\!\!&\!\!\!\!\!\!&\!\!\!12\!\!\!&\!\!\!\!\!\!&\!\!\!4\!\!\!&\!\!\!\!\!\!&\!\!\!1\!\!\!&\!\!\!\!\!\!&\!\!\!\!\!\!&\!\!\!\!\!\!&\!\!\!\!\!\!&\!\!\!\\
\!\!\!&\!\!\!$\frac{3}{2}$\!\!\!&\!\!\!\!\!\!&\!\!\!1\!\!\!&\!\!\!\!\!\!&\!\!\!3\!\!\!&\!\!\!\!\!\!&\!\!\!4\!\!\!&\!\!\!\!\!\!&\!\!\!9\!\!\!&\!\!\!\!\!\!&\!\!\!13\!\!\!&\!\!\!\!\!\!&\!\!\!21\!\!\!&\!\!\!\!\!\!&\!\!\!27\!\!\!&\!\!\!\!\!\!&\!\!\!38\!\!\!&\!\!\!\!\!\!&\!\!\!44\!\!\!&\!\!\!\!\!\!&\!\!\!50\!\!\!&\!\!\!\!\!\!&\!\!\!46\!\!\!&\!\!\!\!\!\!&\!\!\!38\!\!\!&\!\!\!\!\!\!&\!\!\!22\!\!\!&\!\!\!\!\!\!&\!\!\!10\!\!\!&\!\!\!\!\!\!&\!\!\!3\!\!\!&\!\!\!\!\!\!&\!\!\!1\!\!\!&\!\!\!\!\!\!&\!\!\!\!\!\!&\!\!\!\!\!\!&\!\!\!\\
\!\!\!&\!\!\!2\!\!\!&\!\!\!\!\!\!&\!\!\!\!\!\!&\!\!\!1\!\!\!&\!\!\!\!\!\!&\!\!\!1\!\!\!&\!\!\!\!\!\!&\!\!\!3\!\!\!&\!\!\!\!\!\!&\!\!\!5\!\!\!&\!\!\!\!\!\!&\!\!\!10\!\!\!&\!\!\!\!\!\!&\!\!\!14\!\!\!&\!\!\!\!\!\!&\!\!\!22\!\!\!&\!\!\!\!\!\!&\!\!\!29\!\!\!&\!\!\!\!\!\!&\!\!\!38\!\!\!&\!\!\!\!\!\!&\!\!\!41\!\!\!&\!\!\!\!\!\!&\!\!\!41\!\!\!&\!\!\!\!\!\!&\!\!\!31\!\!\!&\!\!\!\!\!\!&\!\!\!19\!\!\!&\!\!\!\!\!\!&\!\!\!7\!\!\!&\!\!\!\!\!\!&\!\!\!2\!\!\!&\!\!\!\!\!\!&\!\!\!\!\!\!&\!\!\!\!\!\!&\!\!\!\!\!\!&\!\!\!\\
\!\!\!&\!\!\!$\frac{5}{2}$\!\!\!&\!\!\!\!\!\!&\!\!\!\!\!\!&\!\!\!\!\!\!&\!\!\!\!\!\!&\!\!\!\!\!\!&\!\!\!1\!\!\!&\!\!\!\!\!\!&\!\!\!1\!\!\!&\!\!\!\!\!\!&\!\!\!3\!\!\!&\!\!\!\!\!\!&\!\!\!5\!\!\!&\!\!\!\!\!\!&\!\!\!10\!\!\!&\!\!\!\!\!\!&\!\!\!14\!\!\!&\!\!\!\!\!\!&\!\!\!22\!\!\!&\!\!\!\!\!\!&\!\!\!27\!\!\!&\!\!\!\!\!\!&\!\!\!34\!\!\!&\!\!\!\!\!\!&\!\!\!32\!\!\!&\!\!\!\!\!\!&\!\!\!26\!\!\!&\!\!\!\!\!\!&\!\!\!14\!\!\!&\!\!\!\!\!\!&\!\!\!6\!\!\!&\!\!\!\!\!\!&\!\!\!1\!\!\!&\!\!\!\!\!\!&\!\!\!\!\!\!&\!\!\!\!\!\!&\!\!\!\\
\!\!\!&\!\!\!3\!\!\!&\!\!\!\!\!\!&\!\!\!\!\!\!&\!\!\!\!\!\!&\!\!\!\!\!\!&\!\!\!\!\!\!&\!\!\!\!\!\!&\!\!\!\!\!\!&\!\!\!\!\!\!&\!\!\!1\!\!\!&\!\!\!\!\!\!&\!\!\!1\!\!\!&\!\!\!\!\!\!&\!\!\!3\!\!\!&\!\!\!\!\!\!&\!\!\!5\!\!\!&\!\!\!\!\!\!&\!\!\!10\!\!\!&\!\!\!\!\!\!&\!\!\!14\!\!\!&\!\!\!\!\!\!&\!\!\!21\!\!\!&\!\!\!\!\!\!&\!\!\!24\!\!\!&\!\!\!\!\!\!&\!\!\!26\!\!\!&\!\!\!\!\!\!&\!\!\!19\!\!\!&\!\!\!\!\!\!&\!\!\!11\!\!\!&\!\!\!\!\!\!&\!\!\!3\!\!\!&\!\!\!\!\!\!&\!\!\!1\!\!\!&\!\!\!\!\!\!&\!\!\!\!\!\!&\!\!\!\\
\!\!\!&\!\!\!$\frac{7}{2}$\!\!\!&\!\!\!\!\!\!&\!\!\!\!\!\!&\!\!\!\!\!\!&\!\!\!\!\!\!&\!\!\!\!\!\!&\!\!\!\!\!\!&\!\!\!\!\!\!&\!\!\!\!\!\!&\!\!\!\!\!\!&\!\!\!\!\!\!&\!\!\!\!\!\!&\!\!\!1\!\!\!&\!\!\!\!\!\!&\!\!\!1\!\!\!&\!\!\!\!\!\!&\!\!\!3\!\!\!&\!\!\!\!\!\!&\!\!\!5\!\!\!&\!\!\!\!\!\!&\!\!\!10\!\!\!&\!\!\!\!\!\!&\!\!\!13\!\!\!&\!\!\!\!\!\!&\!\!\!18\!\!\!&\!\!\!\!\!\!&\!\!\!18\!\!\!&\!\!\!\!\!\!&\!\!\!15\!\!\!&\!\!\!\!\!\!&\!\!\!7\!\!\!&\!\!\!\!\!\!&\!\!\!2\!\!\!&\!\!\!\!\!\!&\!\!\!\!\!\!&\!\!\!\!\!\!&\!\!\!\\
\!\!\!&\!\!\!4\!\!\!&\!\!\!\!\!\!&\!\!\!\!\!\!&\!\!\!\!\!\!&\!\!\!\!\!\!&\!\!\!\!\!\!&\!\!\!\!\!\!&\!\!\!\!\!\!&\!\!\!\!\!\!&\!\!\!\!\!\!&\!\!\!\!\!\!&\!\!\!\!\!\!&\!\!\!\!\!\!&\!\!\!\!\!\!&\!\!\!\!\!\!&\!\!\!1\!\!\!&\!\!\!\!\!\!&\!\!\!1\!\!\!&\!\!\!\!\!\!&\!\!\!3\!\!\!&\!\!\!\!\!\!&\!\!\!5\!\!\!&\!\!\!\!\!\!&\!\!\!9\!\!\!&\!\!\!\!\!\!&\!\!\!11\!\!\!&\!\!\!\!\!\!&\!\!\!13\!\!\!&\!\!\!\!\!\!&\!\!\!9\!\!\!&\!\!\!\!\!\!&\!\!\!5\!\!\!&\!\!\!\!\!\!&\!\!\!1\!\!\!&\!\!\!\!\!\!&\!\!\!\!\!\!&\!\!\!\\
\!\!\!&\!\!\!$\frac{9}{2}$\!\!\!&\!\!\!\!\!\!&\!\!\!\!\!\!&\!\!\!\!\!\!&\!\!\!\!\!\!&\!\!\!\!\!\!&\!\!\!\!\!\!&\!\!\!\!\!\!&\!\!\!\!\!\!&\!\!\!\!\!\!&\!\!\!\!\!\!&\!\!\!\!\!\!&\!\!\!\!\!\!&\!\!\!\!\!\!&\!\!\!\!\!\!&\!\!\!\!\!\!&\!\!\!\!\!\!&\!\!\!\!\!\!&\!\!\!1\!\!\!&\!\!\!\!\!\!&\!\!\!1\!\!\!&\!\!\!\!\!\!&\!\!\!3\!\!\!&\!\!\!\!\!\!&\!\!\!5\!\!\!&\!\!\!\!\!\!&\!\!\!8\!\!\!&\!\!\!\!\!\!&\!\!\!8\!\!\!&\!\!\!\!\!\!&\!\!\!7\!\!\!&\!\!\!\!\!\!&\!\!\!3\!\!\!&\!\!\!\!\!\!&\!\!\!1\!\!\!&\!\!\!\!\!\!&\!\!\!\\
\!\!\!&\!\!\!5\!\!\!&\!\!\!\!\!\!&\!\!\!\!\!\!&\!\!\!\!\!\!&\!\!\!\!\!\!&\!\!\!\!\!\!&\!\!\!\!\!\!&\!\!\!\!\!\!&\!\!\!\!\!\!&\!\!\!\!\!\!&\!\!\!\!\!\!&\!\!\!\!\!\!&\!\!\!\!\!\!&\!\!\!\!\!\!&\!\!\!\!\!\!&\!\!\!\!\!\!&\!\!\!\!\!\!&\!\!\!\!\!\!&\!\!\!\!\!\!&\!\!\!\!\!\!&\!\!\!\!\!\!&\!\!\!1\!\!\!&\!\!\!\!\!\!&\!\!\!1\!\!\!&\!\!\!\!\!\!&\!\!\!3\!\!\!&\!\!\!\!\!\!&\!\!\!4\!\!\!&\!\!\!\!\!\!&\!\!\!6\!\!\!&\!\!\!\!\!\!&\!\!\!4\!\!\!&\!\!\!\!\!\!&\!\!\!2\!\!\!&\!\!\!\!\!\!&\!\!\!\!\!\!&\!\!\!\\
\!\!\!&\!\!\!$\frac{11}{2}$\!\!\!&\!\!\!\!\!\!&\!\!\!\!\!\!&\!\!\!\!\!\!&\!\!\!\!\!\!&\!\!\!\!\!\!&\!\!\!\!\!\!&\!\!\!\!\!\!&\!\!\!\!\!\!&\!\!\!\!\!\!&\!\!\!\!\!\!&\!\!\!\!\!\!&\!\!\!\!\!\!&\!\!\!\!\!\!&\!\!\!\!\!\!&\!\!\!\!\!\!&\!\!\!\!\!\!&\!\!\!\!\!\!&\!\!\!\!\!\!&\!\!\!\!\!\!&\!\!\!\!\!\!&\!\!\!\!\!\!&\!\!\!\!\!\!&\!\!\!\!\!\!&\!\!\!1\!\!\!&\!\!\!\!\!\!&\!\!\!1\!\!\!&\!\!\!\!\!\!&\!\!\!3\!\!\!&\!\!\!\!\!\!&\!\!\!3\!\!\!&\!\!\!\!\!\!&\!\!\!3\!\!\!&\!\!\!\!\!\!&\!\!\!1\!\!\!&\!\!\!\!\!\!&\!\!\!\\
\!\!\!&\!\!\!6\!\!\!&\!\!\!\!\!\!&\!\!\!\!\!\!&\!\!\!\!\!\!&\!\!\!\!\!\!&\!\!\!\!\!\!&\!\!\!\!\!\!&\!\!\!\!\!\!&\!\!\!\!\!\!&\!\!\!\!\!\!&\!\!\!\!\!\!&\!\!\!\!\!\!&\!\!\!\!\!\!&\!\!\!\!\!\!&\!\!\!\!\!\!&\!\!\!\!\!\!&\!\!\!\!\!\!&\!\!\!\!\!\!&\!\!\!\!\!\!&\!\!\!\!\!\!&\!\!\!\!\!\!&\!\!\!\!\!\!&\!\!\!\!\!\!&\!\!\!\!\!\!&\!\!\!\!\!\!&\!\!\!\!\!\!&\!\!\!\!\!\!&\!\!\!1\!\!\!&\!\!\!\!\!\!&\!\!\!1\!\!\!&\!\!\!\!\!\!&\!\!\!2\!\!\!&\!\!\!\!\!\!&\!\!\!1\!\!\!&\!\!\!\!\!\!&\!\!\!1\!\!\!&\!\!\!\\
\!\!\!&\!\!\!$\frac{13}{2}$\!\!\!&\!\!\!\!\!\!&\!\!\!\!\!\!&\!\!\!\!\!\!&\!\!\!\!\!\!&\!\!\!\!\!\!&\!\!\!\!\!\!&\!\!\!\!\!\!&\!\!\!\!\!\!&\!\!\!\!\!\!&\!\!\!\!\!\!&\!\!\!\!\!\!&\!\!\!\!\!\!&\!\!\!\!\!\!&\!\!\!\!\!\!&\!\!\!\!\!\!&\!\!\!\!\!\!&\!\!\!\!\!\!&\!\!\!\!\!\!&\!\!\!\!\!\!&\!\!\!\!\!\!&\!\!\!\!\!\!&\!\!\!\!\!\!&\!\!\!\!\!\!&\!\!\!\!\!\!&\!\!\!\!\!\!&\!\!\!\!\!\!&\!\!\!\!\!\!&\!\!\!\!\!\!&\!\!\!\!\!\!&\!\!\!1\!\!\!&\!\!\!\!\!\!&\!\!\!1\!\!\!&\!\!\!\!\!\!&\!\!\!1\!\!\!&\!\!\!\!\!\!&\!\!\!\\
\!\!\!&\!\!\!7\!\!\!&\!\!\!\!\!\!&\!\!\!\!\!\!&\!\!\!\!\!\!&\!\!\!\!\!\!&\!\!\!\!\!\!&\!\!\!\!\!\!&\!\!\!\!\!\!&\!\!\!\!\!\!&\!\!\!\!\!\!&\!\!\!\!\!\!&\!\!\!\!\!\!&\!\!\!\!\!\!&\!\!\!\!\!\!&\!\!\!\!\!\!&\!\!\!\!\!\!&\!\!\!\!\!\!&\!\!\!\!\!\!&\!\!\!\!\!\!&\!\!\!\!\!\!&\!\!\!\!\!\!&\!\!\!\!\!\!&\!\!\!\!\!\!&\!\!\!\!\!\!&\!\!\!\!\!\!&\!\!\!\!\!\!&\!\!\!\!\!\!&\!\!\!\!\!\!&\!\!\!\!\!\!&\!\!\!\!\!\!&\!\!\!\!\!\!&\!\!\!\!\!\!&\!\!\!\!\!\!&\!\!\!1\!\!\!&\!\!\!\!\!\!&\!\!\!\!\!\!&\!\!\!\\
\!\!\!&\!\!\!$\frac{15}{2}$\!\!\!&\!\!\!\!\!\!&\!\!\!\!\!\!&\!\!\!\!\!\!&\!\!\!\!\!\!&\!\!\!\!\!\!&\!\!\!\!\!\!&\!\!\!\!\!\!&\!\!\!\!\!\!&\!\!\!\!\!\!&\!\!\!\!\!\!&\!\!\!\!\!\!&\!\!\!\!\!\!&\!\!\!\!\!\!&\!\!\!\!\!\!&\!\!\!\!\!\!&\!\!\!\!\!\!&\!\!\!\!\!\!&\!\!\!\!\!\!&\!\!\!\!\!\!&\!\!\!\!\!\!&\!\!\!\!\!\!&\!\!\!\!\!\!&\!\!\!\!\!\!&\!\!\!\!\!\!&\!\!\!\!\!\!&\!\!\!\!\!\!&\!\!\!\!\!\!&\!\!\!\!\!\!&\!\!\!\!\!\!&\!\!\!\!\!\!&\!\!\!\!\!\!&\!\!\!\!\!\!&\!\!\!\!\!\!&\!\!\!\!\!\!&\!\!\!\!\!\!&\!\!\!\!1\\
\hline
d\!&\!\!$j_L\slash j_R$\!&\!\!0\!&\!\!\!$\frac{1}{2}$\!\!\!&\!\!\!1\!\!\!&\!\!\!$\frac{3}{2}$\!\!\!&\!\!\!2\!\!\!&\!\!\!$\frac{5}{2}$\!\!\!&\!\!\!3\!\!\!&\!\!\!$\frac{7}{2}$\!\!\!&\!\!\!4\!\!\!&
\!\!\!$\frac{9}{2}$\!\!\!&\!\!\!5\!\!\!&\!\!\!$\frac{11}{2}$\!\!\!&\!\!\!6\!\!\!&\!\!\!$\frac{13}{2}$\!\!\!&\!\!\!7\!\!\!&\!\!\!$\frac{15}{2}$\!\!\!&\!\!\!8\!\!\!&\!\!\!$\frac{17}{2}$\!\!\!
&\!\!\!9\!\!\!&\!\!\!$\frac{19}{2}$\!\!\!&\!\!\!10\!\!\!&\!\!\!$\frac{21}{2}$\!\!\!&\!\!\!11\!\!\!&\!\!\!$\frac{23}{2}$\!\!\!&\!\!\!12\!\!\!&\!\!\!$\frac{25}{2}$\!\!\!&\!\!\!13\!\!\!
&\!\!\!$\frac{27}{2}$\!\!\!&\!\!\!14\!\!\!&\!\!\!$\frac{29}{2}$\!\!\!&\!\!\!15\!\!\!&\!\!\!$\frac{31}{2}$\!\!\!&\!\!\!16\!\!\!&\!\!\!$\frac{33}{2}$\!\!\!&\!\!\!17\!\!\!&\!\!\!$\frac{35}{2}$\!\!\!\\
\hline
\end{tabular}}}
\caption{Non vanishing BPS numbers $N^d_{j_L,j_R}$ of local ${\cal O}(-3)\rightarrow \mathbb{P}^2$ up to $d=7$.}
%\end{sidewaystable}
\label{bpstable}
\end{table}
Up to small typos in \cite{IKV} the results up to $d=5$ agree with the results of the generalized
vertex~\cite{IKV}.

\subsection{The local Calabi-Yau
manifold ${\cal O}(-2,-2) \rightarrow \mathbb{P}^1\times \mathbb{P}^1$}

Here we describe explicitly the refinement of the five dimensional
index  for the local Calabi-Yau manifold ${\cal O}(-2,-2)
\rightarrow \mathbb{P}^1\times \mathbb{P}^1$. The geometry is physically
very interesting as it contains the refinement of the 4d $N=2$ Seiberg
Witten gauge theory~\cite{KKV1}, the refinement of the 3d Chern Simons
theory on the lens space $L(1,2)$\cite{Aganagic:2002wv} and a potential
refinement of the ABJM theory~\cite{Drukker:2010nc}\cite{Marino:2009jd} .

For the case at hand we have five chiral fields
$X_i$, $i=1,\ldots, 5$ and a gauge group $U(1)^{(1)}
\times U(1)^{(2)}$ under which the fields have charges
\begin{equation}
Q^{(1)}=(-2,1,1,0,0),\qquad  Q^{(2)}=(-2,0,0,1,1)\ ,
\end{equation}
respectively. The vanishing locus of the Stanley-Reisner ideal is
${\cal Sr}=\{ x_1=x_2=0\} \cup \{ x_3=x_4=0\}$

The elliptic curve of the mirror is given as
\begin{equation}
H(x,y)=1+x + \frac{u_1}{x} +y + \frac{u_2}{y}=0\ .
\end{equation}

These periods integrals are annihilated
by the two Picard-Fuchs operators
$\theta_i=u_i \frac{d}{d u_i}$
\begin{equation}
 \begin{array}{rl}
{\cal L}^{(1)}=& \theta_1^2-2(\theta_1 +\theta_2-1) (2 \theta_1+2 \theta_2-1)u_1\\
{\cal L}^{(2)}=& \theta_2^2-2(\theta_1 +\theta_2-1) (2 \theta_1+2 \theta_2-1)u_2 \ ,
 \end{array}
\label{PFp1p1}
\end{equation}
which have a constant solution and two logarithmic solutions
$t_1=\log(u_1)+ \Sigma(u_1,u_2)$ and  $t_2=\log(u_2)+ \Sigma(u_1,u_2)$.
This suggests to change parameters and introduce $u=u_1$
and
\begin{equation}
\Lambda_s=\log(u_1)-\log(u_2),
\end{equation}
which is a trivial solution.
We can now separate the derivatives in the operators
(\ref{PFp1p1}) and capture the system by one differential
operator of third order in $u$, where we understand
$m=e^{\Lambda_s}$ now as a deformation parameter.
This situation is similar to rank one $N=2$ Seiberg-Witten
(gauge) theories. The latter have one
coupling constant related to the complex structure
and thereby to elliptic integrals of the curve
and up to 9 mass parameters for hypermultiplet fields.
Indeed the geometry of ${\cal O}(-2,-2)
\rightarrow \mathbb{P}^1\times \mathbb{P}^1$ there
has a Seiberg-Witten limit~\cite{KKV1}, with a $SU(2)$
gauge group and without hypermultiplets. In the above
parametrization it is  at $(m,u)=(0,1/4)$ and the
decoupling of the mass scale $m$ becomes very simple in
the $(m,u)$ variables.

After same changes of variables~\cite{Klemm:2012ii} we can
parametrize the curve (\ref{curve}) as
\begin{equation}
y^2+x^2-y-\frac{xy}{\sqrt{u}}-m x^2 y=0
\end{equation}
and bring it into Weierstrass form (\ref{weierstrass})
using Nagells algorithm, with
\begin{equation}
\begin{array}{rl}
g_2=& 27 u^4 \left(16 u^2 \left(m^2-m+1\right)-8 u (m+1)+1\right)\\ [2 mm]
g_3=&  -27 u^6 (-1 + 4 u (1 + m)) (1 - 8 u (1 + m) + 8 u^2 (2 - 5 m + 2 m^2))\ .
\end{array}
\end{equation}
This yields a $J$-invariant
\begin{equation}
J=\frac{\left(16 \left(m^2-m+1\right) u^2-8 (m+1) u+1\right)^3}{m^2 u^4 \left(16 (m-1)^2 u^2-8 (m+1) u+1\right)}\ .
\label{jfunction}
\end{equation}
The coefficients in (\ref{genus1a},\ref{genus1b}) are  given by  $a=7,b=\frac{7}{2},c=-2$ and
$d=-1$.

With this information the direct integration determines
\begin{equation}
\begin{array}{rl}
\Delta^2F^{(0,2)}=&
\frac{2560 X^3}{81 u^6}+ \frac{16 X^2 \left(3 n_1^2 u^2+5 (m+1) u-8\right)}{27 u^4}+\\
&\frac{X \left(\left(m^2+50 m+1\right) u^2+6 n_1^2 (m+1) u^3-20 (m+1) u+13\right)}{54 u^2}+ \\
& \frac{9 n_1^2 \left(2 m^2+5 m+2\right) u^4-3 \left(5 m^2+76 m+5\right) u^2-\left(31 m^3-168 m^2-168 m+31\right) u^3+51 (m+1) u-23}{6480} , \\
\Delta^2F^{(1,1)}=&
\frac{32 X^2 n_2}{27 u^4}+\frac{X \left(\left(m^2-10 m+1\right) u^2-n_1^2 (m+1) u^3+(m+1) u-1\right)}{9 u^2}- \\
&\frac{n_1^2\left(43 m^2+130 m+43\right) u^4+\left(90 m^2-548 m+90\right) u^2-4 \left(29 m^3-127 m^2-127 m+29\right) u^3+4 (m+1)
u+13}{8640}\ , \\
\Delta^2F^{(2,0)}=& \frac{X n_2^2}{54 u^2}-
\frac{n_1^2 \left(17 m^2-370 m+17\right) u^4+2 \left(75 m^2+514 m+75\right) u^2-\left(84 m^3+508 m^2+508 m+84\right) u^3-116 (m+1)u+33}{34560}
\end{array}
\label{f2}
\end{equation}
and all higher genus amplitudes. Here we introduced $n_1=(m-1)$ and $n_2= (1-m u-u)$ and
rescaled $u$ by $u\rightarrow u/4$.

Here we note that in order to implement the gap condition
we introduce the conifold variable $\tilde u$ by
\begin{equation}
u=\frac{1}{4 \left(\sqrt{m}+1\right)^2}-\frac{\tilde u}{16}\
\label{conifoldparameter}
\end{equation}
and expand around small $\tilde u$, which means close to the conifold.
One property of the coordinate $\tilde u$ that follows from (\ref{jfunction})
is that ${\rm lim}_{\tilde u\rightarrow 0} \frac{1}{J(m,\tilde u)}=0$ independent
of $m$. As a consequence we can invert (\ref{jfunction})
near $\tilde u\sim 0$ and $q\sim 0$ for $q(\tilde m,\tilde u)$ and obtain
(\ref{nonlogperiod}) and (\ref{Xdef}) as expansions in $\tilde u$, whose coefficients
are exact rational functions in $m^{1/4}$.

\subsubsection{Nekrasov's 4d partition function at weak and strong coupling}
If we change coordinates to~\cite{Kachru:1995fv,KKV1}
\begin{equation}
u=\frac{1}{4}(1-\varepsilon^2 u_{sw}), \quad m=\frac{e^{\Lambda_s}}
{\Lambda_{Sw}^4(1-\varepsilon^2 u_{sw})}
\label{seibergwittenparameter}
\end{equation}
we obtain in leading order in $\epsilon$ from (\ref{nonlogperiod})
the Seiberg-Witten $a$ period and the leading $\varepsilon$
order of $Z$ reproduced exactly Nekrasov's partition function.
Similarly using the variable (\ref{conifoldparameter}) and expanding
near $m=0$ in $m$ and near $u=\mp \Lambda^2$  we obtain to leading
order in  $\varepsilon$ the partition function $Z$ of $N=2$
Seiberg-Witten theory, in the strong coupling region, i.e. at
the monopole and dyon point. The relation between refined string
theory on del Pezzo surfaces and $N=2$ field theory
is  more interesting for higher del Pezzo surfaces  and
will be further discussed in~\cite{HuangKlemm}.

\subsubsection{BPS invariants for ${\cal O}(-2,-2) \rightarrow \mathbb{P}^1\times \mathbb{P}^1$
in the large volume limit}

We did the recursion up to genus 9 and observe for the refined BPS invariants and report
first some of the $n_{g_R,g_L}^\beta$

\begin{table}[!h]
\centering
\begin{tabular}[h]{|c|ccccccc|}
\hline
      & $d_1$ &   0 &     1 &     2 &     3 &     4 &        5   \\
\hline
$d_2$ &       &     &       &       &       &        &          \\
0     &       &     &    -2 &     0 &     0 &      0 &       0   \\
1     &       &  -2 &    -4 &    -6 &    -8 &    -10 &     -12  \\
2     &       &   0 &    -6 &   -32 &  -110 &   -288 &    -644   \\
3     &       &   0 &    -8 &  -110 &  -756 &  -3556 &  -13072   \\
4     &       &   0 &   -10 &  -288 & -3556 & -27264 & -153324  \\
5     &       &   0 &   -12 &  -644 &-13072 &-153324 &-1252040  \\
\hline
\end{tabular}
\caption{Instanton numbers $n^{d_1d_2}_{0,0}$ of ${\cal O}(-K)\rightarrow \mathbb{P}^1\times\mathbb{P}^1$.}
\end{table}

\begin{table}[!h]
\centering
\begin{tabular}[h]{|c|ccccccc|}
\hline
      & $d_1$ &   0 &     1 &     2 &     3 &     4 &        5   \\
\hline
$d_2$ &       &     &       &       &       &        &          \\
0     &       &     &     1 &     0 &     0 &      0 &       0   \\
1     &       &   1 &    10 &    35 &    84 &    165 &     286  \\
2     &       &   0 &    35 &    359&   1987&	7620&	23414   \\
3     &       &   0 &    84 &   1987&	20554&	134882&	657672	\\
4     &       &   0 &   165 &   7620&	134882&	1392751& 10110954\\
5     &       &   0 &   286 &   23414& 657672	& 10110954&104334092\\
\hline
\end{tabular}
\caption{Instanton numbers $n^{d_1d_2}_{1,0}$ of ${\cal O}(-K)\rightarrow \mathbb{P}^1\times\mathbb{P}^1$.}
\end{table}

\begin{table}[!h]
\centering
\begin{tabular}[h]{|c|ccccccc|}
\hline
      & $d_1$ &   0 &     1 &     2 &     3 &     4 &        5   \\
\hline
$d_2$ &       &     &       &       &       &        &          \\
0     &       &     &       &       &       &       &          \\
1     &       &     &       &   9 &     68 &	 300&	   988  \\
2     &       &     &       & 68&	1016&	7792&	41376	\\
3     &       &     &       & 300&	7792&	95313&	760764\\
4     &       &     &       & 988&	41376&	760764&	8695048\\
5     &       &     &       & 2698&	172124&	4552692& 71859628\\
\hline
\end{tabular}
\caption{Instanton numbers $n^{d_1d_2}_{0,1}$ of ${\cal O}(-K)\rightarrow \mathbb{P}^1\times\mathbb{P}^1$.}
\end{table}

\begin{table}[!h]
\centering
\begin{tabular}[h]{|c|ccccccc|}
\hline
      & $d_1$ &   0 &     1 &     2 &     3 &     4 &        5   \\
\hline
$d_2$ &       &     &       &       &       &        &          \\
0     &       &     &       &       &       &        &           \\
1     &       &     &    -6&	-56&	-252&	-792&	-2002 \\
2     &       &     &    -56&	-1232&	-11396&	-65268&	-278564\\
3     &       &     &  -252&	-11396&	-184722&-1726770& -11307496\\
4     &       &     &   -792&	-65268&	-1726770& -24555200&-233289152\\
5     &       &     &   -2002&	-278564&-11307496& -233289152& -3087009512\\
\hline
\end{tabular}
\caption{Instanton numbers $n^{d_1d_2}_{2,0}$ of ${\cal O}(-K)\rightarrow \mathbb{P}^1\times\mathbb{P}^1$.}
\end{table}

\begin{table}[!h]
\centering
\begin{tabular}[h]{|c|ccccccc|}
\hline
      & $d_1$ &   0 &     1 &     2 &     3 &     4 &        5   \\
\hline
$d_2$ &       &     &       &       &       &        &          \\
0     &       &     &       &       &       &        &          \\
1     &       &     &       &       &       &        &          \\
2     &       &     &       &   -120 &	-1484&	-9632&	-43732 \\
3     &       &     &       & -1484&	-33856&	-364908&-2580992\\
4     &       &     &       &   -9632&	-364908&-6064608& -62822028\\
5     &       &     &       &    -43732& -2580992&-62822028& -912904128\\
\hline
\end{tabular}
\caption{Instanton numbers $n^{d_1d_2}_{1,1}$ of ${\cal O}(-K)\rightarrow \mathbb{P}^1\times\mathbb{P}^1$.}
\end{table}

\begin{table}[!h]
\centering
\begin{tabular}[h]{|c|ccccccc|}
\hline
      & $d_1$ &   0 &     1 &     2 &     3 &     4 &        5   \\
\hline
$d_2$ &       &     &       &       &       &        &          \\
0     &       &     &      &      &      &       &          \\
1     &       &    &       &      &      &       &  \\
2     &       &    &     &   	& -12	& -116	& -628\\
3     &       &     &   &   -12	& -580	& -8042	& -64624\\
4     &       &     &   &    -116& 	-8042& 	-167936	& -1964440\\
5     &       &     &   &    -628& 	-64624& 	-1964440& -32242268 \\
\hline
\end{tabular}
\caption{Instanton numbers $n^{d_1d_2}_{0,2}$ of ${\cal O}(-K)\rightarrow \mathbb{P}^1\times\mathbb{P}^1$.}
\end{table}

Changing the basis according to (\ref{basis}) yields
\begin{equation}
 N^{(1,d)}=\begin{cases}
1 & {\rm if  }\ \ j_L=0, j_R=\frac{1}{2}+ d \\
0 & {\rm otherwise}
\end{cases}
\end{equation}
as well as the refined invariants reported in Table~\ref{bpstable2}.

\begin{table}[h!]
%\begin{sidewaystable}[h]
\centering{{
\begin{tabular}[h]{|c|c|cccccccccccccccccccc|}
\hline
$(d_1,d_2)$\!&\!\!$j_L\backslash j_R$\!&\!\!0\!&\!\!\!$\frac{1}{2}$\!\!\!&\!\!\!1\!\!\!&\!\!\!$\frac{3}{2}$\!\!\!&\!\!\!2\!\!\!&\!\!\!$\frac{5}{2}$\!\!\!&\!\!\!3\!\!\!&\!\!\!$\frac{7}{2}$\!\!\!&\!\!\!4\!\!\!&
\!\!\!$\frac{9}{2}$\!\!\!&\!\!\!5\!\!\!&\!\!\!$\frac{11}{2}$\!\!\!&\!\!\!6\!\!\!&\!\!\!$\frac{13}{2}$\!\!\!&\!\!\!7\!\!\!&\!\!\!$\frac{15}{2}$\!\!\!&\!\!\!8\!\!\!&\!\!\!$\frac{17}{2}$\!\!\!
&\!\!\!9\!\!\!&\!\!\!$\frac{19}{2}$\!\!\! \\
\hline
$(2,2)$\!&\!\!0\!&\!\! \!&\!\!\!\!\!\!&\!\!\!\!\!\!&\!\!\!\!\!\!&\!\!\!\!\!\!&\!\!\!1\!\!\!&\!\!\!\!\!\!&\!\!\!1\!\!\!&\!\!\!\!\!\!&
\!\!\!\!\!\!&\!\!\!\!\!\!&\!\!\!\!\!\!&\!\!\!\!\!\!&\!\!\!\!\!\!&\!\!\!\!\!\!&\!\!\!\!\!\!&\!\!\!\!\!\!&\!\!\!\!\!\!
&\!\!\!\!\!\!&\!\!\!\!\!\! \\
\!&\!\!$\frac{1}{2}$\!&\!\! \!&\!\!\!\!\!\!&\!\!\!\!\!\!&\!\!\!\!\!\!&\!\!\!\!\!\!&\!\!\!\!\!\!&\!\!\!\!\!\!&\!\!\!\!\!\!&\!\!\!1\!\!\!&
\!\!\!\!\!\!&\!\!\!\!\!\!&\!\!\!\!\!\!&\!\!\!\!\!\!&\!\!\!\!\!\!&\!\!\!\!\!\!&\!\!\!\!\!\!&\!\!\!\!\!\!&\!\!\!\!\!\!
&\!\!\!\!\!\!&\!\!\!\!\!\! \\
\hline
\hline
$(2,3)$\!&\!\!0\!&\!\! \!&\!\!\!\!\!\!&\!\!\!\!\!\!&\!\!\!\!\!\!&\!\!\!\!\!\!&\!\!\!1\!\!\!&\!\!\!\!\!\!&\!\!\!1\!\!\!&\!\!\!\!\!\!&
\!\!\!2\!\!\!&\!\!\!\!\!\!&\!\!\!\!\!\!&\!\!\!\!\!\!&\!\!\!\!\!\!&\!\!\!\!\!\!&\!\!\!\!\!\!&\!\!\!\!\!\!&\!\!\!\!\!\!
&\!\!\!\!\!\!&\!\!\!\!\!\! \\
\!&\!\!$\frac{1}{2}$\!&\!\! \!&\!\!\!\!\!\!&\!\!\!\!\!\!&\!\!\!\!\!\!&\!\!\!\!\!\!&\!\!\!\!\!\!&\!\!\!\!\!\!&\!\!\!\!\!\!&\!\!\!1\!\!\!&
\!\!\!\!\!\!&\!\!\!1\!\!\!&\!\!\!\!\!\!&\!\!\!\!\!\!&\!\!\!\!\!\!&\!\!\!\!\!\!&\!\!\!\!\!\!&\!\!\!\!\!\!&\!\!\!\!\!\!
&\!\!\!\!\!\!&\!\!\!\!\!\! \\
\!&\!\!$1$\!&\!\! \!&\!\!\!\!\!\!&\!\!\!\!\!\!&\!\!\!\!\!\!&\!\!\!\!\!\!&\!\!\!\!\!\!&\!\!\!\!\!\!&\!\!\!\!\!\!&\!\!\!\!\!\!&
\!\!\!\!\!\!&\!\!\!\!\!\!&\!\!\!1\!\!\!&\!\!\!\!\!\!&\!\!\!\!\!\!&\!\!\!\!\!\!&\!\!\!\!\!\!&\!\!\!\!\!\!&\!\!\!\!\!\!
&\!\!\!\!\!\!&\!\!\!\!\!\! \\
\hline
\hline
$(3,3)$\!&\!\!0\!&\!\! \!&\!\!\!\!\!\!&\!\!\!\!\!\!&\!\!\!1\!\!\!&\!\!\!\!\!\!&\!\!\!1\!\!\!&\!\!\!\!\!\!&\!\!\!3\!\!\!&\!\!\!\!\!\!&
\!\!\!3\!\!\!&\!\!\!\!\!\!&\!\!\!4\!\!\!&\!\!\!\!\!\!&\!\!\!\!\!\!&\!\!\!\!\!\!&\!\!\!\!\!\!&\!\!\!\!\!\!&\!\!\!\!\!\!
&\!\!\!\!\!\!&\!\!\!\!\!\! \\
\!&\!\!$\frac{1}{2}$\!&\!\! \!&\!\!\!\!\!\!&\!\!\!\!\!\!&\!\!\!\!\!\!&\!\!\!\!\!\!&\!\!\!\!\!\!&\!\!\!1\!\!\!&\!\!\!\!\!\!&\!\!\!2\!\!\!&
\!\!\!\!\!\!&\!\!\!3\!\!\!&\!\!\!\!\!\!&\!\!\!3\!\!\!&\!\!\!\!\!\!&\!\!\!\!1\!\!&\!\!\!\!\!\!&\!\!\!\!\!\!&\!\!\!\!\!\!
&\!\!\!\!\!\!&\!\!\!\!\!\! \\
\!&\!\!$1$\!&\!\! \!&\!\!\!\!\!\!&\!\!\!\!\!\!&\!\!\!\!\!\!&\!\!\!\!\!\!&\!\!\!\!\!\!&\!\!\!\!\!\!&\!\!\!\!\!\!&\!\!\!\!\!\!&
\!\!\!1\!\!\!&\!\!\!\!\!\!&\!\!\!2\!\!\!&\!\!\!\!\!\!&\!\!\!3\!\!\!&\!\!\!\!\!\!&\!\!\!\!\!\!&\!\!\!\!\!\!&\!\!\!\!\!\!
&\!\!\!\!\!\!&\!\!\!\!\!\! \\
\!&\!\!$\frac{3}{2}$\!&\!\! \!&\!\!\!\!\!\!&\!\!\!\!\!\!&\!\!\!\!\!\!&\!\!\!\!\!\!&\!\!\!\!\!\!&\!\!\!\!\!\!&\!\!\!\!\!\!&\!\!\!\!\!\!&
\!\!\!\!\!\!&\!\!\!\!\!\!&\!\!\!\!\!\!&\!\!\!1\!\!\!&\!\!\!\!\!\!&\!\!\!1\!\!\!&\!\!\!\!\!\!&\!\!\!\!\!\!&\!\!\!\!\!\!
&\!\!\!\!\!\!&\!\!\!\!\!\! \\
\!&\!\!$2$\!&\!\! \!&\!\!\!\!\!\!&\!\!\!\!\!\!&\!\!\!\!\!\!&\!\!\!\!\!\!&\!\!\!\!\!\!&\!\!\!\!\!\!&\!\!\!\!\!\!&\!\!\!\!\!\!&
\!\!\!\!\!\!&\!\!\!\!\!\!&\!\!\!\!\!\!&\!\!\!\!\!\!&\!\!\!\!\!\!&\!\!\!\!\!\!&\!\!\!1\!\!\!&\!\!\!\!\!\!&\!\!\!\!\!\!
&\!\!\!\!\!\!&\!\!\!\!\!\! \\
\hline
\hline
$(3,4)$\!&\!\!0\!&\!\! \!&\!\!\!1\!\!\!&\!\!\!\!\!\!&\!\!\!1\!\!\!&\!\!\!\!\!\!&\!\!\!3\!\!\!&\!\!\!\!\!\!&\!\!\!4\!\!\!&\!\!\!\!\!\!&
\!\!\!7\!\!\!&\!\!\!\!\!\!&\!\!\!6\!\!\!&\!\!\!\!\!\!&\!\!\!7\!\!\!&\!\!\!\!\!\!&\!\!\!1\!\!\!&\!\!\!\!\!\!&\!\!\!1\!\!\!
&\!\!\!\!\!\!&\!\!\!\!\!\! \\
\!&\!\!$\frac{1}{2}$\!&\!\! \!&\!\!\!\!\!\!&\!\!\!\!\!\!&\!\!\!\!\!\!&\!\!\!\!\!\!&\!\!\!\!\!\!&\!\!\!1\!\!\!&\!\!\!\!\!\!&\!\!\!2\!\!\!&
\!\!\!\!\!\!&\!\!\!4\!\!\!&\!\!\!\!\!\!&\!\!\!6\!\!\!&\!\!\!\!\!\!&\!\!\!\!8\!\!&\!\!\!\!\!\!&\!\!\!2\!\!\!&\!\!\!\!\!\!
&\!\!\!\!\!\!&\!\!\!\!\!\! \\
\!&\!\!$1$\!&\!\! \!&\!\!\!\!\!\!&\!\!\!\!\!\!&\!\!\!\!\!\!&\!\!\!\!\!\!&\!\!\!\!\!\!&\!\!\!\!\!\!&\!\!\!1\!\!\!&\!\!\!\!\!\!&
\!\!\!2\!\!\!&\!\!\!\!\!\!&\!\!\!5\!\!\!&\!\!\!\!\!\!&\!\!\!6\!\!\!&\!\!\!\!\!\!&\!\!\!7\!\!\!&\!\!\!\!\!\!&\!\!\!1\!\!\!
&\!\!\!\!\!\!&\!\!\!\!\!\! \\
\!&\!\!$\frac{3}{2}$\!&\!\! \!&\!\!\!\!\!\!&\!\!\!\!\!\!&\!\!\!\!\!\!&\!\!\!\!\!\!&\!\!\!\!\!\!&\!\!\!\!\!\!&\!\!\!\!\!\!&\!\!\!\!\!\!&
\!\!\!\!\!\!&\!\!\!\!\!\!&\!\!\!\!\!\!&\!\!\!1\!\!\!&\!\!\!\!\!\!&\!\!\!2\!\!\!&\!\!\!\!\!\!&\!\!\!4\!\!\!&\!\!\!\!\!\!
&\!\!\!1\!\!\!&\!\!\!\!\!\! \\
\!&\!\!$2$\!&\!\! \!&\!\!\!\!\!\!&\!\!\!\!\!\!&\!\!\!\!\!\!&\!\!\!\!\!\!&\!\!\!\!\!\!&\!\!\!\!\!\!&\!\!\!\!\!\!&\!\!\!\!\!\!&
\!\!\!\!\!\!&\!\!\!\!\!\!&\!\!\!\!\!\!&\!\!\!\!\!\!&\!\!\!1\!\!\!&\!\!\!\!\!\!&\!\!\!2\!\!\!&\!\!\!\!\!\!&\!\!\!3\!\!\!
&\!\!\!\!\!\!&\!\!\!\!\!\! \\
\!&\!\!$\frac{5}{2}$\!&\!\! \!&\!\!\!\!\!\!&\!\!\!\!\!\!&\!\!\!\!\!\!&\!\!\!\!\!\!&\!\!\!\!\!\!&\!\!\!\!\!\!&\!\!\!\!\!\!&\!\!\!\!\!\!&
\!\!\!\!\!\!&\!\!\!\!\!\!&\!\!\!\!\!\!&\!\!\!\!\!\!&\!\!\!\!\!\!&\!\!\!\!\!\!&\!\!\!\!\!\!&\!\!\!1\!\!\!&\!\!\!\!\!\!
&\!\!\!1\!\!\!&\!\!\!\!\!\! \\
\!&\!\!$3$\!&\!\! \!&\!\!\!\!\!\!&\!\!\!\!\!\!&\!\!\!\!\!\!&\!\!\!\!\!\!&\!\!\!\!\!\!&\!\!\!\!\!\!&\!\!\!\!\!\!&\!\!\!\!\!\!&
\!\!\!\!\!\!&\!\!\!\!\!\!&\!\!\!\!\!\!&\!\!\!\!\!\!&\!\!\!\!\!\!&\!\!\!\!\!\!&\!\!\!\!\!\!&\!\!\!\!\!\!&\!\!\!\!\!\!
&\!\!\!\!\!\!&\!\!\!1\!\!\! \\
\hline
\end{tabular}}}
\caption{Non vanishing BPS numbers $N^{(d_1,d_2)}_{j_L,j_R}$ of local ${\cal O}(-2,-2)\rightarrow \mathbb{P}^1\times \mathbb{P}^1$}
%\end{sidewaystable}
\label{bpstable2}
\end{table}
These results of Table~\ref{bpstable2} agree with the ones of~\cite{IKV}.

\subsubsection{The refinement of perturbative CS theory on the Lens space $L(2,1)$}

In the above parametrization the ABJM slice~\cite{Marino:2009jd}\cite{Drukker:2010nc}
\cite{Klemm:2012ii} is given by $m=1$ and in particular the Chern-Simons
theory on the  lens space  point $L(2,1)$~\cite{Aganagic:2002wv} is
at $(m,u)=(1,\infty)$. The analysis and the choice of
variables is quite similar to~\cite{Aganagic:2002wv} except
that we do not have to solve differential equations, as
we can infer all properties of the local cusp expansions from the
universal relation (\ref{jfunction}) between the complex
structure  parameter $\tau$, defining the periods up
to normalization, and (\ref{nonlogperiod}).

Unlike at the conifold (\ref{conifoldparameter}) one has to
evaluate expressions like (\ref{f2}) at a point in the $u$-plane
where $1/J$ does not vanish for generic $m$, but only for $m=1$,
which is therefore the cusp point. To get the correct double
scaling limit near the orbifold point, the local parameters
$(\tilde m,\tilde u)$ can be defined as
\begin{equation}
m=1-\tilde m, \qquad u=\frac{1}{\tilde u^2 \tilde m^2}\ .
\end{equation}
Since $1/J$ small we can invert (\ref{jfunction}) for
$(\tilde m, \tilde u)$.

We further need to express the $\tilde m$, $\tilde u$  in
terms of the flat coordinates, which are given by the periods
$\Lambda_s$ and the period $a(\tilde u,\tilde m)$,
which can be calculated from (\ref{nonlogperiod}).
There is a subtlety in the latter calculation, because
as we mentioned we have normalized $g_2$ and $g_3$ so
that $a(u,m)$ is a solution to the system (\ref{PFp1p1}) at
$(u,m)=(0,0)$. Clearly scaling $g_2\rightarrow f^2(u,m) g_2$ and
$g_3\rightarrow f^3(u,m) g_3$ does not change
the $J$-function and hence the relation between $q$ and $u,m$.
From (\ref{nonlogperiod}) is is however clear that the
scaling  changes the normalization of the
period $\frac{{\rm d} t}{{\rm d} u}$, that vanished at
the cusp, by the factor $1/f^\frac{1}{2}$.
To get the correctly normalized solution we set
$f(u,m)=\frac{1}{\tilde u^6 \tilde m^4}$. That yields
\begin{equation}
a(\tilde u,\tilde m)=\tilde m \tilde u+\frac{1}{4} {\tilde m}^2 \tilde u+\frac{9}{64}  {\tilde m}^3 u+\frac{25}{256}
{\tilde m}^4 u+ \left(\frac{1225 {\tilde m}^5 \tilde u}{16384}+\frac{{\tilde m}^3 {\tilde u}^3}{192}\right)+\ldots
\end{equation}
as one solution and we chose $-\Lambda_s$ as the second.
We get then
\begin{equation}
\begin{array}{rl}
\tilde m&=1-\exp(-\Lambda_s)=\Lambda_s+{\cal O}(\Lambda_s^2)\\
\tilde u&=\frac{a }{\Lambda}+\frac{a}{4}+\frac{1}{192} a  \Lambda_s -\frac{1}{256} \left(a \Lambda_s^2\right)- \left(\frac{a^3}{192 \Lambda_s}
+\frac{49 a \Lambda_s^3}{737280}\right)+
\left(\frac{17 a \Lambda_s^4}{196608}-\frac{a^3}{768}\right)+O\left(\epsilon^7\right)\ .
\end{array}
\label{orbifoldperiods}
\end{equation}
Because the small parameters are $\tilde m\sim \tilde u\sim \epsilon $ we have $a\sim \epsilon^2$
while $\Lambda_s\sim \epsilon$. This defines the order and convergence of (\ref{orbifoldperiods}).
Defining as in~\cite{Aganagic:2002wv}
\begin{equation}
N_1=S_1=\frac{1}{4}(\Lambda + a),\qquad  N_2=S_2=\frac{1}{4}(\Lambda - a)
\end{equation}
we get the genus 0 partition function from integrating $F_0=\int {\rm d} a \int {\rm d} a \tau $
\begin{equation}
\begin{array}{rl}
F^{(0,0)}=&\frac{1}{2}(S_1^2 \log(S_1) + S_2^2 \log(S_2)) + \frac{1}{288}\left( S_1^4 + 6 S_1^3 S_2 + 18 S_1^2 S_2^2+ \ldots\right)-\\ [2 mm]
    & \frac{1}{345600} \left( 4 S_1 + 45 S_1^5 S_2 + 225 S_1^4 S_2^2 + 1500 S_1^3 S_2^3+\ldots \right)  + {\cal O }(S^6)\ ,
\end{array}
\end{equation}
which agrees with the results~\cite{Aganagic:2002wv}. Here and in the following
the $\ldots$ mean addition of symmetric terms in $S_1$ and $S_2$. To match with the
matrix model, the genus counting parameters of the topological  string
and the Chern Simons matrix model was related~\cite{Aganagic:2002wv} by
$g_s^{top}=2 i \hat g_s$. We extend this to the refined model by
by setting
\begin{equation}
\epsilon^{top}_i = \sqrt{2 i} \epsilon^m_i \qquad i=1,2\ .
\end{equation}
The rational for this is to reproduce $F^{(0,g)}$ and more general
the leading terms (\ref{fngconstant}) of the $F^{(n,g)}$.

With this definition we can extract the refined amplitudes
and calculate first
\begin{equation}
\begin{array}{rl}
F^{(1,0)}=&\frac{1}{24}(\log(S_1) +\log(S_2)) + \frac{1}{576} (S_1^2+30 S_2 S_1+S_2^2)-\\ [2 mm]
          &\frac{1}{138240}(2 S_1^4-255 S_2 S_1^3+1530 S_2^2 S_1^2+\ldots )+ \\ [2 mm]
          &\frac{1}{34836480}(8 S_1^6+945 S_2 S_1^5-43470 S_2^2 S_1^4+150570 S_2^3 S_1^3-\ldots )+{\cal O}(S^8)\ .
\end{array}
\end{equation}
The result for $F^{(0,1)}$ agrees with~\cite{Aganagic:2002wv}
\begin{equation}
\begin{array}{rl}
F^{(0,1)}=&-\frac{1}{12}(\log(S_1) +\log(S_2)) - \frac{1}{288} \left(S_1^2-6 S_2 S_1+S_2^2\right)+\\ [2 mm]
& \frac{1}{69120} (S_1^4+105 S_2 S_1^3-90 S_2^2 S_1^2+\ldots)+ {\cal O}(S^6)\ .
\end{array}
\end{equation}
For $n+g=2$ we obtain the results for the refinement
\begin{equation}
\begin{array}{rl}
F^{(2,0)}=&-\frac{7}{5760}\left( \frac{1}{S_1^2}+ \frac{1}{S_2^2}\right)+\frac{1}{192} + \frac{6047 S_1^2-26430 S_1 S_2 +6047 S_2^2}{5529600}+\\[2 mm]
&\frac{3653 S_1^4-78912 S_1^3 S_2+216054 S_1^2 S_2^2-\ldots}{39813120}+\\[2 mm]
&\frac{193952 S_1^6-15472305 S_1^5 S_1+161797725 S_1^4 S_2-351759000 S_1^3 S_2^3+\ldots}{63700992000}\ .
\end{array}
\end{equation}
as well as
\begin{equation}
\begin{array}{rl}
F^{(1,1)}=&\frac{7}{1440}\left( \frac{1}{S_1^2}+ \frac{1}{S_2^2}\right)+\frac{1}{96} +\frac{2053 S_1^2-5970 S_1 S_2+2053 S_2^2}{1382400} +\\[2 mm]
&\frac{1207 S_1^4-13428 S_1^3S_2+17226 S_1^2 S_2^2-\ldots}{9953280}+ \\ [2 mm]
&\frac{65248 S_1^6-2704095 S_1^5 S_2+8059275 S_2^4 S_2^2+1839000 S_2^3 S_1^3+\ldots}{15925248000}\ .
\end{array}
\end{equation}
The result for  $F^{(0,2)}$
\begin{equation}
\begin{array}{rl}
F^{(0,2)}=&-\frac{1}{240}\left( \frac{1}{S_1^2}+\frac{1}{S_2^2}\right)-\frac{S_1^2+60 S_1 S_2+S_2^2}{57600} +\\[2 mm]
& \frac{S_1^4+126 S_1^3 S_2+378 S_1^2 S_2^2+126 S_1 S_2^3 +S_2^4}{1451520}- \\ [2 mm]
& \frac{64 S_1^6+38385 S_2 S_1^5+334575 S_2^2 S_1^4+124500 S_2^3 S_1^3+\ldots}{2654208000}          \ .
\end{array}
\end{equation}
agrees with the result of~\cite{Aganagic:2002wv}.
Using our exact results for $F^{(n,g)}(X,m,u)$  these expansions
are available up to $n+g=9$. We note the  $F^{(0,g)}$
have no constant terms. This is expected from
the matrix model description of~\cite{Marino:2002fk} and
its large $N$-expansion. In fact these constants are canceled
if we set $\chi=4$ in $N^0_{00}$ contributing
via (\ref{schwingerloope1e2}) to the constants at
infinity. Similarly the classical terms at large radius  at
order $g_s^{-2}$ and  $g_s^{0}$ can be fixed from the
matrix model expansion. In the $\epsilon_R=0$ slice we
also checked to higher order that the perturbative expansion of the
Chern-Simons matrix~\cite{Marino:2002fk} agrees with
the topological string according to the expectations
in~\cite{Aganagic:2002wv}.

The constant terms in $F^{(n>0,g)}$ are not zero.
This is already to be expected from the fact that the refined
Chern-Simons matrix model~\cite{Aganagic:2012au} on $S^3$ involves
a shift in the K\"ahler parameter, relative to the refined
topological string, due to a different choice of ${\cal R}$.
We have evidence that the shift $\Lambda_s+\epsilon_+$ leads
to right parameters to compare with the matrix model
description. A more precise parameter map for the
full model is under investigation. It is also
noticeable that for the choices $b=1/2$ and $b=2$
the expansions simplify. Such specialization of the refined ensemble were
recently studied in~\cite{Krefl:2012ei}.

%---------------------------------------------------------------------------
\section{Enumerative invariants of Calabi-Yau threefolds}
%---------------------------------------------------------------------------
\label{sec:enum}

In this section, we review the enumerative invariants of Calabi-Yau threefolds
that we will use: the stable pair invariants of Pandharipande and Thomas,
and the Gopakumar-Vafa invariants.

%\subsection{Generalized Donaldson-Thomas invariants}
%---------------------------------------------------------------------------
\subsection{Pandharipande-Thomas invariants}\label{ptinvariants}
%---------------------------------------------------------------------------

We begin by explaining the theory
of stable pairs due to Pandharipande and Thomas \cite{pt,ptbps}.  Stable
pairs
clarify the assertions made in \cite{kkv} and and also provide mathematical
proofs.  We then return to the refined invariants with the benefit of stable pairs.

\begin{defn}\label{def:stablepair}
A {\em stable pair\/} on a smooth
threefold $X$ consists of a sheaf $\cF$ on $X$ and a
section $s\in H^0(\cF)$ such that
\begin{itemize}
\item $\cF$ is pure of dimension 1
\item $s$ generates $\cF$ outside of a finite set of points
\end{itemize}
\end{defn}

A stable pair is a D6-D2-D0 brane bound state, and can be written as a complex
\[\cI^\bullet:\cO_X\stackrel{s}{\to}\cF.\]

Let $P_n(X,\beta)$ denote the moduli space of stable pairs with
$\mathrm{ch}_2(\cF)=\beta,\ \chi(\cF)=n$.  Then if $X$ is Calabi-Yau,
$P_n(X,\beta)$ supports a {\em symmetric\/} obstruction theory.  See
\cite{behrend} for the definitions and basic properties of symmetric
obstruction theories.

There are only a few things that we need to know about symmetric obstruction
theories.  The basic idea of a symmetric obstruction theory is that the obstructions are dual to their deformations.
For stable pairs, the space of first order deformations is
$\mathrm{Ext}^1(\cI^\bullet,\cI^\bullet)$ and the space of obstructions is
$\mathrm{Ext}^2(\cI^\bullet,\cI^\bullet)$.  These are dual by Serre
duality.

An important feature of symmetric obstruction theories is that they have
virtual dimension~0, since deformations and obstructions have the same
dimension.

If $M$ is the moduli space associated with a symmetric obstruction theory
and $M$ is smooth, then the corresponding virtual number is
$(-1)^{\mathrm{dim}(M)}e(M)$, where $e(M)$ is the topological euler
characteristic.  This is because the bundle describing the
deformations is the tangent bundle of $M$, so the obstruction bundle must
be the cotangent bundle of $M$, and the euler class of the cotangent bundle
is $(-1)^{\mathrm{dim}(M)}e(M)$.

In general, the virtual number is a weighted euler
characteristic.  See \cite{behrend} for more details.

Now let $X$ be Calabi-Yau and let
$P_n(X,\beta)$ be the moduli space of stable pairs with
$\mathrm{ch}_2(F)=
\beta$ and $\chi(F)=n$, and let $P_{n,\beta}$ be the associated invariant, i.e.\
the degree of the virtual fundamental class of $P_n(X,\beta)$.  These
invariants can be arranged in a generating function
\[
Z_{PT}=\sum_{n,\beta}P_{n,\beta}q^nQ^\beta.
\]

We let $Z_{GW}$ be the generating function for disconnected Gromov-Witten invariants:
\[
Z_{GW}=\mathrm{exp} \left(F'_{GW}(\lambda, Q)\right), \hspace{1em} F'_{GW}(\lambda, Q)=\sum_{\beta\ne 0} \sum_{g}N_{g,\beta}\lambda^{2g-2}Q^\beta,
\]
where $N_{g,\beta}$ is the Gromov-Witten invariant.
The fundamental conjecture from which everything will follow is

\begin{conj}
After the change of variables $q=-e^{i\lambda}$, we have
$Z_{PT}=Z_{GW}$.
\label{conj:ptgw}
\end{conj}

Conjecture~\ref{conj:ptgw} is known to be true in the toric case \cite{ptvertex}.

In low degree, the stable pair moduli spaces have simpler descriptions, as they are isomorphic to
relative Hilbert schemes.

First of all, on a smooth surface $S$, the stable pair moduli spaces are isomorphic to relative Hilbert
schemes.  Let $\beta\in H_2(S,\ZZ)$ and let $p_a$ be the arithmetic genus of curves of class $\beta$.  Let
$\cC^{[n]}$ be the relative Hilbert scheme parametrizing curves $C$ of class $\beta$ and $n$ points on $C$ (more
precisely, a subscheme $Z\subset C$ of length $n$).

\begin{prop} \cite{ptbps} $P_{1-p_a+n}(S,\beta)\simeq\cC^{[n]}$
for any $n\ge0$.
\end{prop}

Next, we claim that if $S$ is Fano, then for each $\beta\in H_2(S,\ZZ)$,
stable pairs on the total space $X$ of $K_S$ are identified with stable
pairs on $S$, for small holomorphic euler characteristic.  We state the result for $\PP^2$.

\begin{prop} $P_{1-p_a+n}(\IP^2,d)=P_{1-p_a+n}(X,d)$
for $n\le d+2$.
\label{prop:justp2}
\end{prop}

To prove Proposition~\ref{prop:justp2},
we first make the following claim:

\medskip\noindent
{\bf Claim}: If $C\subset X$ is a
Cohen-Macaulay curve of degree $d$ which is not contained in $\IP^2$
scheme-theoretically, then $\chi(\cO_C)\ge 1-p_a+(d+3)$.

\smallskip
To prove this claim, we first establish some notation.
let $J\subset \cO_X$ be the
ideal sheaf of $C$ and let $I\subset\cO_X$ be
the ideal sheaf of $\IP^2$.  For later use, $I$ can be generated by
a single section $p\in\cO_X(-3)$ which vanishes precisely along $\IP^2$.
Note that $J$ must contain $I^{k+1}$ for
some $k$, so that $J+I^{k+1}=J$.

From the filtration
\[
\cdots J+I^{n+1}\subset J+I^n\subset\cdots\subset J+I^2\subset J+I\subset
\cO_X
\]
we get exact sequences
\[
0\to \frac{J+I^n}{J+I^{n+1}}\to \frac{J+I^{n-1}}{J+I^{n+1}}\to
\frac{J+I^{n-1}}{J+I^{n}}\to 0
\]
which allow us to write
\[
\chi(\cO_C)=\chi(\cO_X/J)=\sum_{n=0}^k \chi\left(
\frac{J+I^n}{J+I^{n+1}} \right),
\]
Fixing $n$, we get a map
\begin{equation}
\phi_n:\frac{I^n}{I^{n+1}}\to (J+I^n)/(J+I^{n+1}).
\label{eq:ordern}
\end{equation}
Since $(J+I^n)/(J+I^{n+1})$ is generated by the image of $I^n$, we see that
$\phi_n$ is surjective.
Since $I^n/I^{n+1}\simeq\cO_{\IP^2}(3n)$ is locally free on $\IP^2$, we can
form $(I^n/I^{n+1})^*\otimes\ker\phi_n\subset \cO_{\IP^2}$, which is
necessarily an ideal sheaf $K_n$ of a (not necessarily Cohen-Macaulay)
plane curve.  We conclude from
(\ref{eq:ordern}) that $(J+I^n)/(J+I^{n+1})\simeq (\cO_{\IP^2}/K_n)(3n)$.

Let $d_n$ be the degree of the plane curve defined by $K_n$, so that
$d=\sum d_n$.  By Riemann-Roch we have
\[
\chi((J+I^n)/(J+I^{n+1}))\ge 3nd_n+1-(d_n-1)(d_n-2)/2.
\]
Then
\begin{equation}
\chi(\cO_C)=\sum_n \chi((J+I^n)/(J+I^{n+1}))\ge
\sum_n\left(
3nd_n+1-(d_n-1)(d_n-2)/2
\right).
\label{eq:boundchi}
\end{equation}
If $d_2=0$, then $C$ is a plane curve.  If $d_2>0$, then the smallest that the
bound (\ref{eq:boundchi})
can be is if $d_1=d-1$ and $d_2=1$, giving
$\chi(\cO_C)\le 1-(d-2)(d-3)/2+4=1-p_a+d+3$.

The bound is sharp, as can be
seen from the example $J=(p^2,px,x^{d-1})$, where $p\in\cO_X(-3)$ is a
section vanishing on $\IP^2$ and $x$ is a homogeneous coordinate on $\IP^2$.
Then $J+I=(p,x^{d-1})$ and $J+I^2=(p^2,px,x^{d-1})=J$,
$\cO_X/(J+I)\simeq \cO_{\IP^2}/(x^{d-1})$ is just a line $L$ with multiplicity
$d-1$, and $(J+I)/J\simeq\cO_L(3)$, giving $\chi(\cO_X/J)=\chi(\cO_{L^{d-1}})
+\chi(\cO_L(3))=1-(d-2)(d-3)/2+4=1-p_a+d+3$.

\medskip
Proposition~\ref{prop:justp2} follows immediately.  If $n\le d+2$ then by the claim, $C$ must be
supported on $\PP^2$ scheme-theoretically, and so a stable pair $\cO_X\to \cF$ with $F$ supported on $C$ can
be functorially identified with a stable pair $\cO_{\PP^2}\to \cF$.

\bigskip\noindent
{\bf Corollary.} $P_{1-p_a+n}(X,d)\simeq\cC^{[n]}$
for $n\le d+2$.

\subsection{Gopakumar-Vafa invariants}
By the BPS state counts in M-theory, integer-valued \emph{Gopakumar-Vafa invariants} $n_\beta^g$ of $X$ are proposed in \cite{GV2}. These are  are related to the Gromov-Witten invariants by the formula
\begin{equation}\label{eq:gvformula}
\sum_{\beta, g} N_{g,\beta} \lambda^{2g-2} Q^\beta
=\sum_{\substack{\beta, g, k\\ \beta\ne 0}}n^g_\beta \frac{1}{k}\left(2
\sin\left(\frac{k\lambda}{2}\right)\right) ^{2g-2}Q^{k\beta}.
\end{equation}
A priori, $n_\beta^g$ defined by above formula are rational numbers because the Gromov-Witten invariants are rational numbers. The \emph{integrality conjecture} is the assertion that the Gopakumar-Vafa invariants defined recursively via (\ref{eq:gvformula}) are integers.

If Conjecture \ref{conj:ptgw} holds, we can write \cite{katz04}
\begin{equation}
Z_{PT}=\prod_\beta\left(
\prod_{j=1}^\infty\left(1+(-1)^{j+1}q^j Q^\beta\right)^{jn^0_\beta}
\prod_{g=1}^\infty\prod_{k=0}^{2g-2}\left(1+\left(-1\right)^{g-k}q^{g-1-k}Q^\beta
\right)^{\left(-1\right)^{k+g}n^g_\beta {2g-2 \choose k}}\right)\ .
\label{productGV}
\end{equation}
Hence, Gopakumar-Vafa invariants can be deduced from Pandharipande-Thomas invariants.
See \cite{ptbps} for more details on this approach.

According to its origin in string theory, the GV invariants $n^g_\beta$ may be thought of as a virtual number of genus $g$ Jacobians inside the moduli space of
stable sheaves $\cF$ on $X$ of pure dimension~1 with $\mathrm{ch}_2(\cF)=\beta$.  This viewpoint led to a computational method for the GV invariants which we
will review in the next section and will refine in Section~\ref{refinement}.  Using a symmetric obstruction theory on this moduli space, the genus~0 GV invariants
$n^0_\beta$ can be directly defined mathematically as the associated virtual number~\cite{katz08}.

%---------------------------------------------------------------------------
\section{KKV approach}
%---------------------------------------------------------------------------
\label{sec:kkv}

In this section, we review the method of \cite{kkv} for the
geometric computation of the Gopakumar-Vafa invariants.  In
Section~\ref{geometricrefined} we will show
that the method readily extends to compute the refined invariants, using
the refinement of the Pandharipande-Thomas invariants which we will describe
in Section~\ref{refinement}.  Furthermore, the refined invariants can be used
to compute the $\su\times\su$ BPS invariants.
The computation will be implemented for
local $\IP^2$ in Section~\ref{refinedtobps}.

\subsection{Generalities}
\label{generalities}

The idea of \cite{kkv} was to compare the cohomology of the Hilbert schemes
$C^{[k]}$ of length $k$ subschemes (i.e.\ $k$ points counted with multiplicity)
of a smooth curve $C$ with the cohomology of its Jacobian $J(C)$.  The comparison
can be carried out either geometrically via the Abel-Jacobi mapping, or
representation-theoretically
by the associated Lefschetz actions of $SU(2)$ on the cohomologies of the
Hilbert scheme and the Jacobian.  It was further proposed that the relative
Hilbert scheme could be used to extend the comparison to families under
certain hypotheses.

Pandharipande and Thomas observed in
\cite{ptbps} that the methods of \cite{kkv} could be
generalized and made more rigorous using the moduli space of stable pairs in place of the relative
Hilbert scheme.  We will begin with the ideas of \cite{kkv} and then will
reformulate these ideas in the language of stable pairs, thereby supplying the
details of the observation of \cite{ptbps}.

We follow the conventions of physics and denote by $[k]$ the (half-integer) spin $k$ representation of $\su$, so that $\dim[k]=2k+1$.

Denote by $I_1=[\frac12]+2g[0]$ the $SU(2)$ content of the standard
Lefschetz decomposition of the cohomology $H^*(C)$ of a genus $g$
Riemann surface $C$.

It is easy to see that the Lefschetz action on a
genus $g$ Jacobian $J(C)$ is then
\begin{equation}
\label{jaclefschetz}
H^*(J(C))=I_g := \left(I_1\right)^{\otimes g}=
\bigoplus_{i=0}^g \left\{ {2g\choose g-i}-{2g\choose g-i-2} \right\}\left[\frac{i}2\right],
\end{equation}
as is easily proven by induction.

Let $C^{[k]}$ denote the Hilbert scheme of
length $k$ subschemes of $C$, which is
just the $k$-fold symmetric product of $C$, by the smoothness of $C$.
Noting that $[\frac12]$ corresponds to the even
cohomology of $C$ while $2g[0]$ corresponds
to the odd cohomology, it follows that
as $\su$ representations
\begin{equation}
\label{cklefschetz}
%\begin{array}{rl}
H^*\left(C^{[k]}\right)%&=\displaystyle{\mathrm{Sym}^k\left(\left[\frac12\right]+2g[0]\right)
=\bigoplus_i \mathrm{Sym}^i\left[\frac12\right]\otimes
\wedge^{k-i}\left(2g[0]\right)%} \\[4 mm]
%&
=\displaystyle{\bigoplus_i{2g\choose k-i}\left[\frac{i}2\right].}
%\end{array}
\end{equation}
Comparing (\ref{jaclefschetz}) and (\ref{cklefschetz}), we see that we have
as $\su$ representations
\begin{equation}
\label{cgidentity}
H^*(C^{[g]})=H^*(J(C))\oplus H^*(C^{[g-2]}).
\end{equation}

More identities can be inferred by comparing (\ref{jaclefschetz}) and
(\ref{cklefschetz}), but we need to establish some notation and conventions
first.

We introduce a linear operator $\theta$ on the representation ring of $\su$,
defined on generators as
\begin{equation}
\label{raising}
\theta\left(\left[\frac{k}2\right]\right)=\left\{
\begin{array}{cl}
\left[\frac{k-1}2\right] & k>0\\
0 & k=0\\
\end{array}\right. .
\end{equation}
In \cite{kkv}, $\theta$ was described in terms of the $\su$ raising operator,
which acts on $H^*(J(C))$ as cup product with the cohomology class of the
theta divisor of $J(C)$.

Again by comparing (\ref{jaclefschetz}) and (\ref{cklefschetz}), we
get equalities of $\su$ representations
\begin{equation}
\label{ajidentity}
H^*(C^{[k]})=\theta^{g-k}H^*(J(C))\oplus H^*(C^{[k-2]})
\end{equation}
for each $0\le k\le g$, where we understand $H^*(C^{[k-2]})=0$ for $k<2$ and
$C^{[0]}$ to be a point.
The case $k=g$ is just (\ref{cgidentity}).

By the structure of massive 5-dimensional BPS representations, the Hilbert
space of BPS states associated to M2-branes wrapping a homology class
$\beta\in H_2(X,\ZZ)$ can be written as
\begin{equation}
\label{hbetadef}
\left[\left(\frac12,0\right)\oplus 2\left(0,0\right)\right]\otimes
\hat{\cH}_\beta
\end{equation}
\noindent
for some $\su\times\su$ representation $\hat{\cH}_\beta$.  The $\su\times\su$
BPS invariants $N^\beta_{j_L,j_R}$ are then defined as the multiplicities of the
representations $[(j_L,j_R)]$ in $\hat{\cH}$:
\begin{equation}
\label{lrmults}
\hat{\cH}_\beta=\oplus_{j_L,j_R}N^\beta_{j_L,j_R}\left[\left(j_L,j_R\right)\right].
\end{equation}
The Gopakumar-Vafa invariants can be deduced from (\ref{lrmults}) by
\begin{equation}
\label{gvdef}
\mathrm{Tr}(-1)^{F_R}\hat{\cH}_\beta=\oplus_g n^g_\beta I_g,
\end{equation}
where as usual the operator $(-1)^{F_R}$ is the identity on integer spin
representations of $\su_R$ and is minus the identity on half-integer spin
representations.  Explicitly, we have
\be
\sum_{g=0}^\infty n^g_\beta I_g=\mathrm{Tr}(-1)^{F_R}\hat{\cH}_\beta=\sum_{j_L} N^\beta_{j_Lj_R} (-1)^{2j_R}(2j_R+ 1)\left[\frac{j_L}{2}\right].
\ee

\bigskip
For the rest of this paper, we consider
$X$ to be a local toric Calabi-Yau threefold, the total space of
the canonical bundle of a toric Fano surface $S$.  Let $\beta\in H_2(S,\ZZ)$
be an effective class, and  let $p_a$ be the arithmetic genus of the curves in
the divisor class $\beta$.

The ansatz of \cite{kkv} was that (\ref{ajidentity})
holds in families as follows.
Let $\cC$ be the universal curve of class $\beta$.  Let
$\cC^{[k]}$ denote the relative Hilbert scheme of $k$ points in the curves of
the family.  Let us further suppose that $\cC^{[k]}$ is smooth, so that its
cohomology supports an $\su$ representation via Lefschetz. Then the assertion
was that
\begin{equation}
H^*\left(\cC^{[k]}\right)=\left(\theta^{p_a-k}\hat{\cH}_\beta\right)_{\su_{\Delta}}\oplus H^*\left(\cC^{[k-2]}\right)+{\rm\ correction\ terms},
\label{eq:kkvid}
\end{equation}
where the correction terms arise from reducible curves.
These correction terms will
be made precise below from the theory of stable pairs.
In \cite{kkv}, the GV invariants were deduced from (\ref{eq:kkvid}) by
applying $\mathrm{Tr}(-1)^F$.

The identity (\ref{eq:kkvid}) is to be understood as an identity of
$\su$ representations.  The $\su$ representations on $H^*(\cC^{[k]})$
and $H^*(\cC^{[k-2]})$ are just the respective Lefschetz actions as
before so we just have to explain the meaning of
$(\theta^{p_a-k}\hat{\cH}_\beta)_{\su_{\Delta}}$.  On representations of
$\su_L\times\su_R$, we define $\theta$ by via the $\su_L$
representation, i.e.\
\[
\theta\left(\left[\left(j_L,j_R\right)\right]\right)
=\left(\theta\left(\left[j_L\right]\right)\right)\otimes\left[j_R\right].
\]
The subscript $\su_\Delta$ denotes that the resulting $\su\times\su$
representation should be restricted to the diagonal $\su\subset\su\times\su$.

We now update the method of \cite{kkv},
rigorously explaining the computation of the GV
invariants from the PT invariants,
assuming the product formula (\ref{productGV}).  While the
method of \cite{kkv} is not needed to compute the GV invariants from (\ref{productGV}), the use of this method will serve as a warm-up for our handling of
the refined invariants in Section~\ref{refinedtobps}.

In Section~\ref{ptinvariants} we saw
that if $\cC$ denotes the universal
curve of class $\beta$ in $S$, then we have
\[
P_{k+1-p_a}(X,\beta)=\cC^{[k]}
\]
for sufficiently small $k$.  For $S=\PP^2$ and curves of degree $d$, the bound is $k\le d+2$.
For general $k$, we need to use $P_{k+1-p_a}(X,\beta)$ in place of $\cC^{[k]}$ in
(\ref{eq:kkvid}).

Putting $r=p_a-k$ and continuing to assume for the
moment that the PT moduli spaces are smooth, (\ref{eq:kkvid}) becomes
\begin{equation}
\label{updatedkkv}
H^*\left(P_{1-r}(X,\beta)\right)=\left(\theta^{r}\hat{\cH}_\beta\right)_{\su_{\Delta}}
\oplus H^*\left(P_{-1-r}(X,\beta)\right)+{\rm\ correction\ terms}.
\end{equation}

If $P_{n}(X,\beta)$ is smooth, then we have\footnote{If we have a symmetric
obstruction theory on a smooth projective variety
$\M$ (e.g.\ the obstruction theory on the moduli
space of stable pairs) then the virtual fundamental class is just
$(-1)^{\dim \M}(e(\M))$.  In terms of the Lefschetz action, this is
$\mathrm{Tr}(-1)^FH^*(\M)$.}
\[
\mathrm{Tr}(-1)^FH^*(P_{n}(X,\beta))=P_{n,\beta}.
\]
Applying $\mathrm{Tr}(-1)^F$ to (\ref{eq:kkvid}), using (\ref{gvdef}),
and replacing the euler
characteristic of a smooth $P_{n}(X,\beta)$ more generally with
$P_{n,\beta}$, we arrive at a precise statement that can be
proven.

\medskip
\begin{prop}
\label{kkvprop}
There are explicit identities
\[
P_{1-r,\beta}-P_{-1-r,\beta}=\mathrm{Tr}(-1)^F\left(\theta^{r}\hat{\cH}_\beta\right)_{\su_{\Delta}}+O\left(\left(n^{h}_\gamma\right)^2\right)=\]
\[\sum_{g\ge1} n^g_\beta \left({2g-2\choose g-2+r}-{2g-2\choose g+r}\right)+\delta_{r,0}n^0_\beta+
O\left(\left(n^{h}_\gamma\right)^2\right),
\]
where the omitted terms $O\left(\left(n^{h}_\gamma\right)^2\right)$ are
explicit nonlinear terms in the $\{n^{h}_\gamma\}$ arising from the expansion
of (\ref{productGV}).
\end{prop}
The terms $O\left(\left(n^{h}_\gamma\right)^2\right)$ are precisely the correction
terms of \cite{kkv}.  Thus, Proposition~\ref{kkvprop} is a mathematically
rigorous formulation of the validity of the KKV method.  The $\delta_{r,0}$
in the statement of the Proposition is the usual Kronecker delta.

\smallskip
Note that the smoothness of the PT moduli spaces is no longer assumed.

\bigskip
The proof of Proposition~\ref{kkvprop} is straightforward.  We focus on the
terms of class $\beta$ by writing
\[
Z_{\mathrm{PT}}=\sum_\beta Z_{\mathrm{PT}}^\beta Q^\beta.
\]
Then from (\ref{productGV})  we obviously have
\begin{equation}
\label{linearterms}
Z_{\mathrm{PT}}^\beta=\sum_{g\ge1}
n^g_\beta \left(\sum_k{2g-2\choose g-1-k}q^k\right)+\sum_{j=1}^\infty
(-1)^{j+1}jn^0_\beta q^j +O\left(\left(n^{h}_\gamma\right)^2\right).
\end{equation}

From (\ref{ajidentity}) and the definition (\ref{gvdef}) of the
$n^g_\beta$, we compute that
\begin{equation}
\label{traceinvariant}
\mathrm{Tr}(-1)^F\left(\theta^{r}\hat{\cH}_\beta\right)_{\su_{\Delta}}=
\sum_{g\ge1} n^g_\beta \left\{ {2g-2\choose g-r}-{2g-2\choose g-r-2} \right\}
+\delta_{0,r}n^0_\beta.
\end{equation}
The lemma follows immediately from (\ref{linearterms}) and
(\ref{traceinvariant}), with the $O\left(\left(n^{h}_\gamma\right)^2\right)$
terms in the statement of the lemma being precisely the difference of
the two explicit $O\left(\left(n^{h}_\gamma\right)^2\right)$ terms in the expansions of
$P_{1-r,\beta}$ and $P_{-1-r,\beta}$ arising from the
product formula (\ref{productGV}).

\subsection{Local $\PP^2$}
\label{localp2}
We now illustrate the low degree cases with $X$ equal to local $\PP^2$.
For curves of degree $d$, we have $p_a=p_a(d)=(d-1)(d-2)/2$. We will
assume that $n^g_d=0$ for $g>p_a(d)$.  As
explained in Section~\ref{ptinvariants},
the PT-moduli spaces will be
equal to the relative Hilbert schemes in the cases discussed below, so the
PT invariants can be calculated by hand.

\medskip\noindent
$d=1$.  Since lines have genus~0, we set $n^1_g=0$ for $g>0$ in the
generating function (\ref{productGV}), which gives
\[
Z_{PT}=\prod_{j=1}^\infty\left(1+(-1)^{j+1}q^jQ\right)^{jn^0_1}+O(Q^2)=1+Q
\left(n^0_1\sum_{n=0}^\infty(-1)^n(n+1)q^{n+1}\right)+O(Q^2)
\]
Comparing coefficients of $qQ$, we see that $P_{1,1}=n^0_1$ and there are no
correction terms, matching Proposition~\ref{kkvprop} with $\beta=1$, $r=0$,
and $\hat{\cH}_1=[0,1]$.  But $P_1(X,1)$
is just the moduli space of lines in $\PP^2$, itself a $\PP^2$.  So
$P_{1,1}=+e(P_1(X,1))=3$, so that $n^0_1=3$.

{$d=2$.}  Since degree~2 curves have arithmetic genus~0, we set
$n^g_2=0$ for $g>0$, and then the
expansion of $Z_{PT}$ has the form
\[
\prod_{d=1}^2
\prod_{j=1}^\infty\left(1+(-1)^{j+1}q^jQ^d\right)^{jn^0_d}+O(Q^3)=1+Q(\cdots)
+Q^2(n^0_2q+O(q^2))+O(Q^3).
\]
The coefficient of $qQ^2$ gives $P_{1,2}=n^0_2$.  However,
$P_1(X,2)$ is the space of conics in $\IP^2$ with no point, parametrized
by $\IP^5$.  Thus
$P_{1,2}=-e(\PP^5)=-6$, giving $n^0_2=-6$.  This is in agreement with
Proposition~\ref{kkvprop} with $d=2$ and $r=0$, and $\hat{\cH}_2=[0,5/2]$.

\medskip\noindent
{$d=3$.}  Since degree~3 curves have arithmetic genus~1, we set
$n^g_3=0$ for $g>1$, and then the
expansion of $Z_{PT}$ gives
$$
\left(\prod_{d=1}^3
  \prod_{j=1}^\infty\left(1+(-1)^{j+1}q^jQ^d\right)^{jn^0_d}\right)
\left(1-Q^3\right)^{-n^1_3}+O(Q^4)$$ $$=1+Q(\cdots)
+Q^2(\cdots)+Q^3(n^1_3+n^0_3q+O(q^2))+O(Q^4).$$
The coefficients of $Q^3$ and $qQ^3$ give
$P_{0,3}=n^1_3,\ P_{1,3}=n^0_3$.

Now $P_0(X,3)$ is the moduli space of cubics in $\IP^2$, which is
isomorphic to
$\IP^9$, so $n^0_3=P_{0,3}=-10$.
Next, $P_1(X,3)$ is the
universal cubic curve $\cC$.  Consider the map
$\cC\to\IP^2$, which forgets the curve and remembers the point,
$(C,p)\mapsto p$.  This exhibits $\cC$ as a
$\IP^8$-bundle over $\IP^2$, since the space of cubics through any point
of $\IP^2$ is a codimension~1 linear subspace of $\IP^9$, i.e.\ a $\IP^8$.
In particular $\cC$ is smooth.
Taking the euler characteristic gives $n^1_3=e(\cC)=(9)(3)=27$.
These results are consistent with Proposition~\ref{kkvprop} for $d=3$ and $r=0,1$
and $\hat{\cH}_3=[1/2,9/2]+[0,3]$.

\medskip\noindent
{$d=4$.} Since degree~4 curves have arithmetic genus~3, we set
$n^g_4=0$ for $g>3$, and then the
expansion of $Z_{PT}$ has the form
\[
1+Q(\cdots) +\ldots+
Q^4\left(n^3_4q^{-2}+\left(n^2_4+4n^3_4\right)q^{-1}+\right.
\]
\[\left.\left(n^1_4+2n^2_4+6n^3_4\right)
+\left(n^0_4+n^2_4+4n^3_4 +n^0_1n^1_3\right)q+\ldots \right)+O(Q^5).
\]

Comparing coefficients, we see
\begin{equation}
P_{-2,4}=n^3_4,\
P_{-1,4}=n^2_4+4n^3_4,\ P_{0,4}=n^1_4+2n^2_4+6n^3_4,\
P_{1,4}=n^0_4+n^2_4+4n^3_4 +n^0_1n^1_3
\label{deg4equations}
\end{equation}
This gives immediately
\begin{equation}
\label{deg4kkvequations}
P_{0,4}-P_{-2,4}=n^1_4+2n^2_4+5n^3_4,\ P_{1,4}-P_{-1,4}=n^0_4+n^0_1n^1_3.
\end{equation}

Note in particular the natural occurrence of $n^0_1n^1_3$ in the last equation
of (\ref{deg4kkvequations}).  This product was explained in \cite{kkv} as a
``correction term'' arising from quartic curves which factor into a line and a
cubic, but the equation makes its role very clear.

We now compute the degree~4 GV invariants.

Now $P_{-2}(X,4)$ is the moduli space of quartic plane curves, which is
isomorphic to
$\IP^{14}$; hence $P_{-2,4}=15$.

$P_{-1}(X,4)$ is the universal curve, a $\IP^{13}$ bundle over
$\IP^2$; hence $P_{-1,4}=-(3)(14)=-42$.

$P_{0}(X,4)$ is the
relative Hilbert scheme of length 2 subschemes.  Since any length 2 subscheme of $\IP^2$ (including the degenerate case of a single point with multiplicity 2) imposes
independent conditions on the space of degree~4 curves, $\cC^{[2]}$ is a $\IP^{12}$ bundle over the (smooth) Hilbert scheme
$(\IP^2)^{[2]}$ of length 2 subschemes $\IP^2$, hence smooth itself.

The euler numbers of the Hilbert scheme are
computed by
\begin{equation}
\sum_n e\left(
\left(\IP^2\right)^{[n]}\right) s^n = \eta(s)^3
\label{eulerhilb}
\end{equation}
and in particular $e((\IP^2)^{[2]})=9$. So $P_{0,4}=(9)(13)=117$.

Next, $P_1(X,4)$ is
the relative Hilbert scheme $\cC^{[3]}$, a $\IP^{11}$ bundle
over $(\IP^2)^{[3]}$, which has euler characteristic 22.  So
$P_{1,4}=-(22)(12)=-264$.

We then solve (\ref{deg4equations}) or equivalently (\ref{deg4kkvequations})
for the $n^g_4$ to get
\[
n^3_4=15
\]
\[
n^2_4=-42-4(15)=-102
\]
\[
n^1_4= 117-2(-102)-6(15)=231
\]
\[
n^0_4=-264-(-102)-4(15)-(3)(-10)=-192
\]

This is all consistent with Proposition~\ref{kkvprop} with $\hat{\cH}_4=[3/2,7]+
[1,11/2]+[1/2,6+5+4]+[0,13/2+9/2+7/2]$, comparing with the first two equalities
in (\ref{deg4equations}) (for $r=2,3$) and with (\ref{deg4kkvequations})
(for $r=0,1$).

The method can be applied without difficulty for $d=5$ and for $d=6,\ g>2$.
However, for $d=6$ and $g=2$, there is a problem because $P_{0}(X,6)$, the
relative Hilbert scheme $\cC^{[8]}$, is not obviously smooth.  This is not
an obstacle, since the PT invariants can be calculated anyway by localization
\cite{ptvertex}.  We will review the calculation in Section~\ref{pttoric}.

To see the problem with smoothness,
we project onto $(\IP^2)^{[8]}$ and compute the fibers.  If the
8 points are general, then they impose independent conditions and the fiber is
$\IP^{27-8}=\IP^{19}$.  But if the 8 points are contained in a line $L$,
then if they are also contained in
a degree 6 curve $C$, they are necessarily contained in the intersection $C\cap
L$. But if $C\cap L$ were finite, it could be no more than 6 points; hence
$C\cap L$ is infinite and $C$ must contain $L$.  We learn that the fiber
consists of all degree 6 curves $L\cup D$, where $D$ is a degree 5 curve.
But the space of all $D$ is a $\IP^{20}$ and so the fiber dimension jumps, and
we can no longer conclude smoothness.  \footnote{If 7 points are chosen to be contained in a line $L$ instead of 8, the same argument shows $C$ again contains $L$, but there will not be any jump in the fiber
dimension in that case.}

In general, for the family of degree $d$ curves, $\cC^{[k]}$ is smooth whenever $k\le d+1$.  We confirm this by checking the worst case, when $k=d+1$.  As usual, we have a
fibration $\cC^{[d+1]}\to\IP^{[d+1]}$ and want to check that the fibers are projective
spaces of the same dimension.  The generic fibers are projective spaces of dimension
$d(d+3)/2-(d+1)=(d+2)(d-1)/2$, since general points will impose $d+1$ independent
conditions on the space of all degree $d$ curves.  The worst possible case is when the $d+1$ points
are contained in a line $L$.  As in the degree 4 case, if these points are also contained in a degree $d$ curve, we conclude that $L$ is contained in $C$. Therefore $C$ is the union of $L$ and an arbitrary curve of degree $d-1$.  But curves of degree $d-1$ are parametrized by a projective space of dimension $(d+2)(d-1)/2$, and the fibers all
have the same dimension, as claimed.

\section{Pandharipande-Thomas for toric Calabi-Yau manifolds}
\label{pttoric}

\subsection{Combinatorics of the $T$-fixed loci}\label{toriccombinatorics}
In this section, we review the classification of torus fixed locus of $P_n(X,d)$ via box configurations studied in \cite{ptvertex}.
Let $s\colon \cO_X \to \cF$ be a torus fixed stable pair. Consider the associated exact sequence
\[0\to I_C\to \cO_X \to \cF\to Q\to 0. \]
Then the curve $C$ is a torus invariant curve in $X$ and the zero-dimensional cokernel $Q$ of $s$ is supported on torus fixed points. We start with a description of the torus invariant curves.

We may restrict our attention to affine torus invariant open sets containing a unique fixed point. For example, in local $\IP^2$ there are three fixed points and three affine open sets containing each of these fixed points.
Let $x_1, x_2, x_3$ be the coordinate functions on a affine torus invariant open set $U$ such that the torus $T\simeq (\IC^*)^3$ acts by \[(t_1,t_2,t_3)\cdot x_i=t_ix_i.\]

Since $C$ is $T$-fixed, it is defined by a monomial ideal $I$ of the polynomial ring $R=\IC[x_1,x_2,x_3]$. By the purity of $\cF$, $C$ is a Cohen-Macaulay curve, i.e.\ a pure one-dimensional curve with no embedded point.  Therefore
the ring $R/I$ is of pure dimension one.  The set of such monomial ideals $I$ is in one-to-one correspondence with the set of triples of three outgoing partitions as follows.

A monomial ideal $I$ in $R$ is associated to a three dimensional partition $\pi$ by considering a union of boxes corresponding to the weights of $R/I$ in the group of characters of $T$, identified with $\IZ^3$ in the usual way. The localizations
\[  (I)_{x_i}\subset \IC[x_1,x_2,x_3]_{x_i},\]
for $i=1,2,3$ are all $T$-fixed, and hence each corresponds to a two-dimensional partition $\pi^i$. One can think of $\pi^i$ as a cross-section of the three dimensional partition $\pi$ by a plane $x_i=c$ for a large integer $c$. We will call $\pi^i$ the outgoing partition of $\pi$. Conversely, given a triple $(\pi^1,\pi^2,\pi^3)$ of outgoing partitions, the monomial ideal $I$ is defined by a unique minimal three dimensional partition with outgoing partition $(\pi^1,\pi^2,\pi^3)$. The minimality assumption is due to the Cohen-Macaulay property of the curve $C$. We denote the curve corresponding to the outgoing partition $\vec{\pi}=(\pi^1,\pi^2,\pi^3)$ by $C_{\vec{\pi}}$.

\begin{prop}[\cite{pt}]\label{prop:ptequivalenttosheaf}
Let $\mathfrak{m}\subset \cO_C$ be the ideal of a zero dimensional subscheme of a Cohen-Macaulay curve $C$. A stable pair $(\cF,s)$ with support $C$ satisfying \[{\rm Support}^{red}(Q)\subset {\rm Support}(\cO_C/\mathfrak{m})\] is equivalent to a subsheaf of $\mathscr{H}om(\mathfrak{m}^r,\cO_C)/\cO_C$, for $r\gg 0$.
\end{prop}

We have inclusions
\[ \mathscr{H}om(\mathfrak{m}^r,\cO_C)\to \mathscr{H}om(\mathfrak{m}^{r+1},\cO_C)\]
by the purity of $\cO_C$. Hence, by Proposition \ref{prop:ptequivalenttosheaf}, we may consider a stable pair as a subsheaf of the limit \[\varinjlim_{r} \mathscr{H}om(\mathfrak{m}^r,\cO_C)/\cO_C.\]
Under the equivalence of Proposition~\ref{prop:ptequivalenttosheaf}, the subsheaf of $\displaystyle \varinjlim_{r} \mathscr{H}om(\mathfrak{m}^r,\cO_C)/\cO_C$ is precisely the cokernel $Q$.

Let $\pi^1[x_2,x_3]$ be the monomial ideal of $\IC[x_2,x_3]$ defined by a partition $\pi^1$, and let
\[ M_1= \IC[x_1,x_1^{-1}]\otimes\left(\IC[x_2,x_3]/\pi^1[x_2,x_3]\right). \]
We define $M_2$ and $M_3$ similarly. Hence, $M_i$ may be viewed in the space of $T$-characters as an infinite cylinder $\textrm{Cyl}_i\in \IZ^3$ along the $x_i$ axis with cross-section $\pi^i$.

Then we have
\[
\varinjlim_{r} \mathscr{H}om(\mathfrak{m}^r,\cO_C)\simeq \bigoplus_{i=1}^3 M_i=: M.
\]
The submodule $\cO_C$ in $M$
is generated by $(1,1,1)$. Hence, the $T$-fixed stable pair $(\cF|_U, s|_U)$ corresponds to a finite dimensional $T$-invariant submodule of $M/\langle (1,1,1)\rangle$. In \cite{ptvertex}, this submodule is described by box configurations as follows.

There are three types of $T$-weights of $M/\langle (1,1,1)\rangle$:
\begin{itemize}
\item[(i)]  weights which are contained in exactly one cylinder $\textrm{Cyl}_i$ and have negative $i$-th coordinate. The set of all weights of these type is denoted by $\I$.
\item[(ii)] weights which are contained in exactly two and three cylinders. The sets of weights of these types are denoted by $\II$ and $\III$ respectively.
\end{itemize}
Let $\IC_w$ be the one dimensional weight space with the weight $w$. Then, we have by definition
\begin{equation}\label{eq:Mweights}
 M/\langle (1,1,1)\rangle = \bigoplus_{w\in \I\cup \II} \IC_w \oplus \bigoplus_{w\in \III} (\IC_w)^2.
\end{equation}
The $R$-module structure on $M/\langle (1,1,1)\rangle$ is such that multiplication by $x_i$ increases the $i$-th coordinate of the weight vector by one. See \cite{ptvertex} for more precise statements.  Therefore, the $T$-fixed pair on $U$ can be described by a $T$-invariant $R$-submodule of \eqref{eq:Mweights}, which by taking its $T$-weights, yields a \emph{labeled box configuration} in $\I\cup \II \cup \III$.

A labeled box configuration is a collection of a finite number of boxes supported on $\I\cup \II \cup \III$, where a box at a type $\III$ weight $w$ may be labeled by a one dimensional subspace of $(\IC_w)^2$ in \eqref{eq:Mweights}. A box indicates the inclusion of the corresponding $T$-weight in $Q$.  An unlabeled type $\III$ box indicates the inclusion of the entire two dimensional space $(\IC_w)^2$ in $Q$. Fixed point loci corresponding to box configurations not containing any labeled type $\III$ boxes are isolated.  Since box configurations correspond to submodules of $M/\langle (1,1,1)\rangle$, it must be invariant under multiplication by $x_i$. We refer to \cite{ptvertex} for more detail.

We return to the case where $X$ is local $\IP^2$. Local $\IP^2$ has three $T$-fixed points, denoted by $p_0$, $p_1$, and $p_2$. Let $L_{ij}$ be the $T$-invariant line connecting $p_i$ and $p_j$. Then, as there is no compact $T$-invariant line along the fiber direction, we do not have any type $\III$ boxes.\footnote{The same result holds for any local Calabi-Yau threefold based on a toric surface.}  It follows
that all $T$-fixed points in $P_n(X,d)$ are isolated. Since there are no labeled type $\III$ boxes, the condition that the labeled box configuration defines a $T$-invariant $R$-submodule is simply as follows.

\begin{itemize}
  \item[$(\dagger)$] For $w=(w_1,w_2,w_3)\in \I\cup \II$, if any of \[(w_1-1,w_2,w_3), (w_1,w_2-1,w_3), (w_1,w_2,w_3-1)\] support a box then $w$ must support a box.
\end{itemize}

\begin{figure}
\center
\includegraphics[width=10cm]{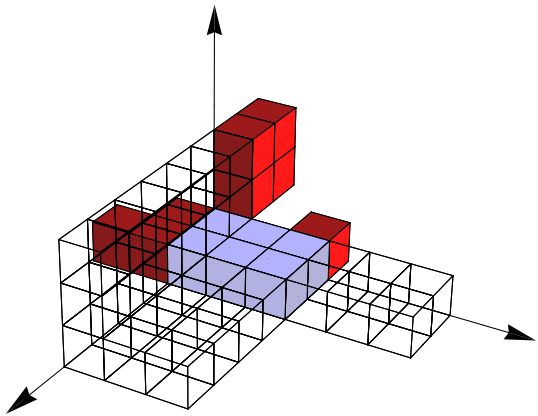}
\caption{The figure shows boxes of type $\I$ in red, boxes of type $\II$ in blue (light)
and two uncolored positive cylinders corresponding to the partitions
$(3,1,1)$ and $(1,1)$. In local toric Calabi-Yau manifolds there are only two cylinders and no boxes of type $\III$.}
\label{mon}
\end{figure}

The \emph{length} of the box configuration is defined by the dimension of corresponding submodule of $M/\langle (1,1,1)\rangle$ as a vector space, which in the case of local $\IP^2$, is the same as the number of boxes. Then by \cite{ptvertex}, we have the following.

\begin{prop}\label{prop:ptp2counting}
Torus fixed points in $P_n(X,d)$ are in one-to-one correspondence with tuples $(B_0, B_1, B_2)$ of three box configurations satisfying the rule $(\dagger)$ such that for some triple $\vec{\lambda}=(\lambda^{01},\lambda^{02},\lambda^{12})$ of partitions with $|\lambda^{01}|+|\lambda^{02}|+|\lambda^{12}|=d$, the outgoing partitions of $B_0, B_1, B_2$ are $(\lambda^{01},\lambda^{02},\emptyset)$, $(\lambda^{12},\lambda^{01},\emptyset)$, and $(\lambda^{02},\lambda^{12},\emptyset)$ respectively, and the sum of the lengths of $B_0, B_1, B_2$ is equal to $n-\chi(\cO_{C(\vec{\lambda})})$, where $C(\vec{\lambda})$ is the torus fixed curve on $\IP^2$ defined by the partitions $\lambda^{ij}$ along $T$-invariant line $L_{ij}$.
\end{prop}

\begin{proof}
A $T$-fixed stable pair $(\cF,s)$ is supported on a $T$-invariant curve of degree $d$, which is given by $C(\vec{\lambda})$ for some triple $\vec{\lambda}$ of partitions.
By the exact sequence
\[ 0 \to \cO_{C(\vec{\lambda})} \to \cF \to Q\to 0,\]
we have
\[n=\chi(\cF) = \chi(Q) + \chi(\cO_{C(\vec{\lambda})}).\]
Moreover, $Q$ must be supported on fixed points, and by the above discussion, at each fixed point $Q$ corresponds to a box configuration satisfying the rule $(\dagger)$.
\end{proof}

For a triple $\vec{\pi}=(\pi^1,\pi^2,\pi^3)$ of partitions, define the \emph{renormalized volume} $|\vec{\pi}|$ by
\[|\vec{\pi}|= \#\left\{\pi \cap [0,\dots,N]^3 \right\}-
(N+1) \sum_{i=1}^3 |\pi^i| \,, \quad N\gg 0, \]
where $\pi$ is the three dimensional partition corresponding to the curve $C_{\vec{\pi}}$.

\begin{lem}In Proposition \ref{prop:ptp2counting},
\begin{align*}
  \chi(\cO_{C(\vec{\lambda})})=& \sum_{\lambda\in \{\lambda^{01},\lambda^{02},\lambda^{12}\}}\left(\sum_{j=1}^{l(\lambda)} \left(\frac{\lambda_j(3\lambda_j+1)}{2}-j\lambda_j\right)\right) \\
  &+|(\lambda^{01},\lambda^{02},\emptyset)|+|(\lambda^{12},\lambda^{01},\emptyset)|+|(\lambda^{02},\lambda^{12},\emptyset)|,
\end{align*}
where the partition $\lambda$ is written as $\lambda= \lambda_1+\cdots+\lambda_{l(\lambda)}$ with $\lambda_1\ge\lambda_2\ge\cdots\lambda_{l(\lambda)}$ and the
two-dimensional partitions $\lambda$ are oriented so that $\lambda_1$ represents
the order of thickening of the associated line in the non-compact direction normal to $\PP^2$.
\end{lem}
\begin{proof}
Each torus fixed line $L_{ij}$ is isomorphic to $\IP^1$, and has normal bundle $\cO_{\IP^1}(1)\oplus\cO_{\IP^1}(-3)$. This follows from an elementary computation by applying \cite[Lemma 5]{mnop1}.
\end{proof}

The case when $X$ is local $\PP^1\times \PP^1$ is similar. Local $\PP^1\times \PP^1$ has four $T$-fixed points $p_0$, $p_1$, $p_2$ and $p_3$, and four $T$ invariant lines $L_{ij}$ connecting them. The torus fixed points are given similarly by tuples of box configurations at four points whose outgoing partitions are given by partitions along four invariant lines. We state the results.

\begin{prop}\label{prop:p1p1}
Let $X$ be local $\PP^1\times \PP^1$. Torus fixed points in $P_n(X,(d_1,d_2))$ are in one-to-one correspondence with tuples $(B_0, B_1, B_2, B_3)$ of four box configurations satisfying the rule $(\dagger)$ that for some quadruple $\vec{\lambda}=(\lambda^{01},\lambda^{12},\lambda^{23},\lambda^{30})$ of partitions with $(|\lambda^{01}|+|\lambda^{23}|, |\lambda^{12}|+|\lambda^{30}|)=(d_1,d_2)$, the outgoing partitions of $B_0, B_1, B_2,$ and $B_3$ are $(\lambda^{30},\lambda^{01},\emptyset)$, $(\lambda^{01},\lambda^{12},\emptyset)$, $(\lambda^{12},\lambda^{23},\emptyset)$ and $(\lambda^{23},\lambda^{30},\emptyset)$ respectively, and the sum of the lengths of $B_0, B_1, B_2,$ and $B_3$ is equal to $n-\chi(\cO_{C(\vec{\lambda})})$, where $C(\vec{\lambda})$ is the torus fixed curve on $\IP^1\times \IP^1$ defined by the partition $\lambda^{ij}$ along $T$-invariant line $L_{ij}$.
\end{prop}

\begin{lem}In Proposition \ref{prop:p1p1},
\begin{align*}
  \chi(\cO_{C(\vec{\lambda})})=& \sum_{\lambda\in \{\lambda^{01},\lambda^{12},\lambda^{23},\lambda^{30}\}}\left(\sum_{j=1}^{l(\lambda)}  \lambda_j^2\right) \\
  &+|(\lambda^{30},\lambda^{01},\emptyset)|+|(\lambda^{01},\lambda^{12},\emptyset)|+|(\lambda^{12},\lambda^{23},\emptyset)|+|(\lambda^{23},\lambda^{30},\emptyset)|.
\end{align*}
As before, the
two-dimensional partitions $\lambda$ are oriented so that $\lambda_1$ represents
the order of thickening of the associated line in the non-compact direction
normal to $\PP^1\times\PP^1$.
\end{lem}

\subsection{Virtual equivariant tangent obstruction theory}
\label{virteqto}

The virtual tangent space of $P_n(X,d)$ at a point corresponding to a pair $\cI^\bullet:\cO_X\stackrel{s}{\to}\cF$ is given by $$\mathcal{T}_{\cI^\bullet}=\Ext^1(\cI^\bullet,\cI^\bullet)-\Ext^2(\cI^\bullet,\cI^\bullet).$$

Consider
\[\chi(\cI^\bullet,\cI^\bullet)=\sum_{i=0}^3(-1)^i \Ext^i(\cI^\bullet,\cI^\bullet).\]
The $T$-action on $\Ext^0(\cI^\bullet,\cI^\bullet)\simeq \IC$ is trivial. By Serre duality, $\Ext^3(\cI^\bullet,\cI^\bullet)$ is the one dimensional $T$-representation with $T$-weight $\delta^{-1}$, where $\delta$ is the $T$-weight of Calabi-Yau form. Hence, in the representation ring of the torus $T$, we have
\[\mathcal{T}_{\cI^\bullet}= 1- \delta^{-1} - \chi(\cI^\bullet,\cI^\bullet).\]
It is enough to compute the representation of $\chi(\cI^\bullet,\cI^\bullet)$.

For each affine invariant open set $U_\alpha$, we denote the $T$-character of $\Gamma(U_\alpha,\cF)$ by $F_\alpha$. Let $U_{\alpha\beta}$ be the intersection of $U_\alpha$ and $U_\beta$. Then, after reordering the indices if necessary, the $T$-character $F_{\alpha\beta}$ of $\Gamma(U_{\alpha\beta},\cF)$ is of the form
$$F_{\alpha\beta}= \delta(t_1)F_{\alpha\beta}(t_2,t_3),$$
where $t_1$, $t_2$ and $t_3$ are $T$-weights of three coordinate axis of $U_\alpha$.
Here $\delta(t_1)$ is the formal delta function at $t_1=1$
\[
\delta(t_1)=\displaystyle\sum_{n=-\infty}^{\infty}t_1^n=\frac{1}{1-t_1}+\frac{t_1^{-1}}{1-t_1^{-1}}
\]
and $F_{\alpha\beta}(t_2,t_3)= \displaystyle\sum_{(k_2,k_3)\in \pi^{\alpha\beta}} t_2^{k_2} t_3^{k_3},$ where $\pi^{\alpha\beta}$ is corresponding outgoing partition. These can be easily obtained from the description of $T$-fixed stable pair in Section \ref{toriccombinatorics}.

In \cite{ptvertex}, the representation of $\chi(\cI^\bullet,\cI^\bullet)$ is computed in terms of $F_\alpha$ and $F_{\alpha\beta}$ via the local-to-global spectral sequence and \v{C}ech complex with respect to open cover $\{U_\alpha\}$. It can be easily seen that we only need contributions from $U_\alpha$ and $U_{\alpha\beta}$.

Define the bar operation
\[
Q\in \IZ((t_1,t_2,t_3)) \mapsto \overline{Q}\in \IZ((t_1,t_2,t_3))
\]
by $t_i\mapsto t_i^{-1}$ on variables.
Let
\[
G_\alpha=F_\alpha-\frac{\overline{F}_\alpha}{t_1t_2t_3}+F_\alpha\overline{F}_\alpha\frac{(1-t_1)(1-t_2)(1-t_3)}{t_1t_2t_3}.
\]
This is the contribution to the representation of $\chi(\cI^\bullet,\cI^\bullet)$ from open set $U_\alpha$. In general, $G_\alpha$ is not a finite Laurent polynomial. To get a finite vertex contribution, a part of the edge contribution needs to be subtracted.
Let
\[
G_{\alpha\beta}=-F_{\alpha\beta}-\frac{\overline{F}_{\alpha\beta}}{t_2t_3}+F_{\alpha\beta}\overline{F}_{\alpha\beta}\frac{(1-t_2)(1-t_3)}{t_2t_3}.
\]
Then the edge contribution from the open set $U_{\alpha\beta}$ is $\delta (t_1) G_{\alpha\beta}$. We define a vertex character
\[
V_\alpha = F_\alpha + \sum_{\beta_i} \frac{G_{\alpha\beta_i}}{1-t_i},
\]
where the summation runs over the three neighboring vertices $\beta_i$. We denote the remaining edge contribution by $E_{\alpha\beta}$. Then it is shown in \cite{ptvertex} that $V_\alpha$ and $E_{\alpha\beta}$ are finite Laurent polynomials. The former is called the vertex contribution and the latter the edge contribution.  Hence we obtain the $T$-character of $\mathcal{T}_{\cI^\bullet}$.

\begin{prop}
In the representation ring of the torus $T$, we have
\begin{equation}\label{eq:Trep}
\mathcal{T}_{\cI^\bullet}= \sum_{\alpha} V_\alpha +\sum_{\alpha\beta}E_{\alpha\beta}.
\end{equation}
\end{prop}

We restrict to a subtorus $\{\delta=1\}\subset T$ preserving the Calabi-Yau form. Then, \eqref{eq:Trep} is of the form $\mathcal{T}_{\cI^\bullet}= Tan_{\cI^\bullet} - Obs_{\cI^\bullet}$, where $Tan_{\cI^\bullet}$ and $Obs_{\cI^\bullet}$ are $T$-representations of $\Ext^1(\cI^\bullet,\cI^\bullet)$ and $\Ext^2(\cI^\bullet,\cI^\bullet)$ respectively and are dual to each other. In the next section, we use $Tan_{\cI^\bullet}$ to define the {\em virtual Bialynicki-Birula decomposition\/} of $P_n(X,d)$.

We have written a Mathematica program to compute the equivariant obstruction
theory of the $P_n(X,d)$.

\section{Refinement of the Pandharipande-Thomas invariants}
\label{refinement}

In this section, we define the refined PT invariants, first explaining the
motivic viewpoint, and then the equivariant viewpoint.  In
Section~\ref{refinedtobps} we will see that applying the procedure of
Section~\ref{sec:kkv} to the refined invariants leads to a geometric
calculation of the $\su\times\su$ BPS invariants, assuming the equivariant
product formula (\ref{productrefined}) for the refined invariants.
Alternatively, the verification of
relationships between refined PT invariants using wallcrossing methods can
lead, at least in principle, to a proof of the equivariant product formula.

\subsection{Virtual Bialynicki-Birula decomposition}

We first recall the Bialynicki-Birula decomposition.  Let $M$ be a smooth
$n$-dimensional projective variety over $\CC$ admitting a $\CC^*$ action
with finitely many fixed points.
For each $p\in M^{\CC^*}$, there
is an induced action of $\CC^*$ on $T_pM$, leading to a decomposition into
eigenspaces of the characters $\chi$ of $\CC^*$
\begin{equation}
T_pM=\bigoplus_{\chi\in X(\CC^*)}T_p^\chi .
\label{eigenspaces}
\end{equation}
Put $T_p^+=\oplus_{\chi>0}T_p^\chi$ and $T_p^-=\oplus_{\chi<0}T_p^\chi$ and let
$d_p^+=\dim T_p^+$, $d_p^-=\dim T_p^-$.  Zero is not an eigencharacter since
$p$ is an isolated fixed point.

Now let
\begin{equation}
U_p=\left\{x\in M\mid \lim_{t\to0}t\cdot x =p\right\}.
\label{BBcell}
\end{equation}
Then $U_p$ is a cell of dimension $d^+_p$. The collection of all cells $U_p$
constitutes a cell decomposition, the Bialynicki-Birula decomposition.  It
follows that the absolute motive $[M]$ of $M$ is given by
\begin{equation}
[M]=\sum_{p\in M^{\CC^*}}\LL^{d^+_p},
\label{bbabsmotive}
\end{equation}
a polynomial in the absolute motive $\LL=[\CC]$ of $\CC$.

Let's now endow $M$ with its canonical symmetric obstruction theory, with
obstruction bundle $T^*M$.  We can view this as the symmetric obstruction
theory associated to the trivial superpotential $W=0$.  Then the construction
of \cite{bbs} associates to $(M,W=0)$ a virtual motive\footnote{Our sign
convention is chosen to match both the conventions of physics and the work
\cite{no} of Nekrasov and Okounkov to be discussion in Section~\ref{eqindex}.}
\begin{equation}
[M]^{\mathrm{vir}}=\left(-\LL^{-1/2}\right)^n[M]=\left(-\LL^{-1/2}\right)^n\sum_{p\in M^{\CC^*}}\LL^{d^+_p}
=\sum_{p\in M^{\CC^*}}\left(-\LL^{-1/2}\right)^{d^+_p-d^-_p}.
\label{bbvirmotive}
\end{equation}

\noindent
{\bf Remark}.  (\ref{bbvirmotive}) shows that $[M]^{\mathrm{vir}}$ is independent
of the choice of $\CC^*$ action on $M$ with finitely many fixed points.  It is
an open question as to whether the virtual motive depends on the choice of a
superpotential on a smooth variety containing $M$ which induces the same
symmetric obstruction theory.  It is shown in \cite{bbs} that the relative
virtual motive is $(-\LL^{-1/2})^n$, independent of choices.  Furthermore,
the argument comes close to proving the desired independence for the virtual
motive itself.

We now return to the situation where $X$ is a
local toric Calabi-Yau threefold.  Consider the PT moduli space $P_n(X,\beta)$
and let $P_n(X,\beta)^T$ be its fixed-point locus.  In
Section~\ref{toriccombinatorics}, we saw that $P_n(X,\beta)^T$ is
isolated.  Pick a 1-parameter subgroup $\CC^*\subset T$ sufficiently generic so
that $P_n(X,\beta)^{\CC^*}=P_n(X,\beta)^{T}$.  If $P_n(X,\beta)$ is smooth,
there is an associated
Bialynicki-Birula decomposition, leading to a formula for the absolute motive
$[P_n(X,\beta)]$ as a polynomial in $\LL$, independent of the choice of
$\CC^*\subset T$.
The virtual
motive associated with $(P_n(X,\beta),\ W=0)$ is then
\begin{equation}
\label{smoothvirmotive}
[P_n(X,\beta)]^\mathrm{vir}=\sum_{Z\in P_n(X,\beta)^{\CC^*}}
\left(-\LL^{-1/2}\right)^{d^+_Z-d^-_Z},
\end{equation}
a Laurent polynomial in $\LL^{1/2}$, and independent of the choice of
$\CC^*\subset T$.

We will associate to any $P_n(X,\beta)$, whether smooth or not, a {\em virtual
Bialynicki-Birula decomposition\/}, leading to a virtual Bialynicki-Birula
motive $[P_n(X,\beta)]^{vir}$
given by (\ref{smoothvirmotive}).

In the smooth case, each term in the
virtual motive is equal to the
absolute motive of
the corresponding cell in the Bialynicki-Birula decomposition times a power
of $-\LL^{-1/2}$.  In general, the exponent of $\LL^{-1/2}$
is the dimension of the tangent space to
$P_n(X,\beta)$ at $Z$.  Since $P_n(X,\beta)$ is not assumed smooth, the
exponents can be different for different fixed points.

We will see in Section~\ref{eqindex} that $[P_n(X,\beta)]^{vir}$ is independent
of the choice of $\CC^*\subset T_0\subset T$, where $T_0\subset T$ is the
subtorus which preserves the holomorphic 3-form.  We see this indirectly,
by translating from the motivic to the
equivariant setting and using \cite{no}.

We remark that this sort of calculation was first attempted in \cite{DS}.

\bigskip\noindent
{\bf Remark}.  A natural question to ask is whether there is a
systematic way to define $P_n(X,\beta)$ as the
critical point locus of a superpotential $W$ on a smooth space.  This would
lead to an intrinsic definition of the virtual motive $[P_n(X,\beta)]^{vir}$
and would make it feasible to find a motivic proof of its independence of the
choice of $\CC^*\subset T_0$.

\subsection{Nekrasov and Okounkov's equivariant virtual index}
\label{eqindex}

In this section, we review the results of \cite{no}, which will allow
us to present our results in a more rigorous manner.

We begin with a quick
discussion of the virtual holomorphic Euler characteristic and virtual
Riemann-Roch.  A reference for this part is \cite{fg}.

Let $M$ admit a perfect obstruction theory.  Then $M$ has a virtual
structure sheaf $\mathcal{O}_M^\mathrm{vir}$, a class in $K^0(M)$.  Given
any class $V\in K^0(M)$, its virtual holomorphic euler characteristic
can be defined as
\begin{equation}
\chi^\mathrm{vir}(M,V):=\chi(M,V\otimes\mathcal{O}_M^\mathrm{vir}).
\label{virholeuler}
\end{equation}

The right hand side of (\ref{virholeuler}) can be computed by Riemann-Roch on
$M$, which can be recast as virtual Riemann-Roch on $M$ taken together with
its obstruction theory.

Now suppose that $M$ admits a perfect obstruction theory, so that in particular
its virtual dimension is~0 and $M$ has a Donaldson-Thomas-type invariant which
we write as
$\#^\mathrm{vir}(M)$.  In this situation, the virtual canonical bundle $K_M^{\mathrm{vir}}$
has a square root $(K_M^{\mathrm{vir}})^{1/2}$
(for example, if $M$ is
smooth, then $(K_M^{\mathrm{vir}})^{1/2}$ is just the ordinary canonical
bundle of $M$).  Then the {\em virtual index\/} of $M$ is
\begin{equation}
\chi^\mathrm{vir}(M,(K_M^{\mathrm{vir}})^{1/2})=\#^\mathrm{vir}(M).
\label{virchivir}
\end{equation}

Now suppose that $M$ and its obstruction theory $E^\bullet$
are $G$-equivariant, with $(E^\bullet)^\vee\simeq E^\bullet\otimes\delta$,
for some character $\delta$ of $G$.  While $K_M^{\mathrm{vir}}$ inherits
inherits a $G$ action, $(K_M^{\mathrm{vir}})^{1/2}$ need not and so we may need
to go to a double cover $\tilde{G}$ of $G$ to define an equivariant virtual
index, a representation of $\tilde{G}$ rather than just a number.  The
fundamental result is that
\begin{equation}
\label{equivariantindex}
\chi^{\mathrm{vir}}(M,(K_M^{\mathrm{vir}})^{1/2})\in \ZZ[\delta^{1/2},\delta^{-1/2}],
\end{equation}
where $\delta^{1/2}$ is a character of $\tilde{G}$ which is a
square root of $\delta$.

Now suppose that we have a 1-parameter subgroup $\sigma:\CC^*\hookrightarrow G$
contained in the kernel of $\delta$ such that $M^{\CC^*}$ is finite.  Then
localization is particularly simple.

\begin{prop}
\begin{equation}
\chi^{\mathrm{vir}}(M,(K_M^{\mathrm{vir}})^{1/2})=\sum_{p\in M^{\CC^*}}
\left(-\delta^{1/2}\right)^{d^+_p-d^-_p}.
\label{locindex}
\end{equation}
\label{eqloc}
\end{prop}

Note that (\ref{locindex}) is precisely the same as (\ref{bbvirmotive}), with
$\delta$ in place of $\LL$.  Since the equivariant virtual index has been
defined without a choice of 1~parameter subgroup, we see that the right hand
side of (\ref{eqloc}) is independent of the choice of 1-parameter subgroup
$\sigma$ contained in the kernel of $\delta$.  Thus the virtual
BB motive is
independent of the choice of $\sigma$ as well.

\smallskip
We now apply this to the situation of the PT invariants of a local toric
Calabi-Yau threefold.  Since each $P_n(X,\beta)$ is compact, the theory of
\cite{no} applies.  We have $G=T=(\mathrm{\CC^*})^3$ and we let $\delta$ be
the character of the natural $T$-action on the $T$-invariant holomorphic
3-form on $X$ as in Section~\ref{virteqto}.
Then the PT obstruction theory $E^\bullet$ on $P_n(X,\beta)$
satisfies $(E^\bullet)^\vee\simeq E^\bullet\otimes\delta$.  We can therefore
define and compute the equivariant index for PT invariants by the results of
this section.  This index, the refined PT invariant,
will be denoted by $P^\mathrm{r}_{n,\beta}(X)\in\ZZ[\delta^{1/2},\delta^{1/2}]
\simeq\ZZ[\LL^{1/2},\LL^{-1/2}]$.

\subsection{Geometric calculation of refined PT-invariants}
\label{geometricrefined}

The calculation of the refined PT invariants is now straightforward.
We choose a generic 1-parameter subgroup $\CC^*\subset T_0$ and apply
localization using the results of Section~\ref{pttoric}.  The refined
invariants are then given by Proposition~\ref{eqloc} in the
equivariant language or (\ref{bbvirmotive}) in the motivic language.
We have implemented these calculations on a computer.

\section{Calculation of the BPS invariants from the refined PT-invariants}
\label{refinedtobps}

The basic idea is to repeat the argument of Section~\ref{sec:kkv}
equivariantly, using (\ref{productrefined}) in place of (\ref{productGV}).
For clarity, we rewrite (\ref{productrefined}) here in more geometric
language.
\begin{equation}
Z^r_{PT}=\prod_{\beta,j_L,j_R}\prod_{m_{L/R}=-j_{L/R}}^{j_{L/R}}
\prod_{m=1}^\infty\prod_{j=0}^{m-1} \left(1 - \LL^{-m/2 + 1/2+j-m_R}
(-q)^{m-2m_L}Q^\beta\right)^{(-1)^{2(j_L+j_R)}N^\beta_{j_L,j_R}},
\label{motivicproduct}
\end{equation}
where we have made the changes of variables
\[
q_L^{1/2}\to -q^{-1},\ q_R\to \LL^{-1},\ e^{\epsilon_1}\to \LL^{-1/2}\left(-q\right),\ e^{\epsilon_2}\to \LL^{1/2}\left(-q\right)
\]
starting from (\ref{productrefined}).  We have also changed the index variables by
identifying $m_1+m_2$ and $m_1-m_2$ in (\ref{productrefined}) with $m+1$ and $m+1-2j$,
respectively.

Clearly the $\su\times\su$ BPS invariants $N^\beta_{j_L,j_R}$
together with (\ref{motivicproduct}) determine the refined PT invariants.  In
this section, we show how to reverse the procedure, determining the
$\su\times\su$ BPS invariants from the knowledge of the refined PT invariants
and (\ref{productrefined}).

\smallskip\noindent
{\em Claim:\/} The $N^\beta_{j_L,j_R}$ are uniquely determined from the refined PT
invariants and (\ref{motivicproduct}).

For convenience, we put
\[
[j_R]_\LL:=\LL^{-j_R}+\LL^{-j_R+1}+\ldots+\LL^{j_R-1}+\LL^{j_R}.
\]
Note that the $\{[j_R]_\LL\}$ form an additive basis for the vector space of Laurent
polynomials in $\LL$ and $\LL^{-1}$ which are symmetric under the $\ZZ_2$
symmetry $\LL\leftrightarrow\LL^{-1}$.  The refined PT invariants always respect
this symmetry.
Furthermore, since
$[j_R]_\LL$ is just the character of the $[j_R]$ representation of $\su$,
the map $[j_R]\mapsto[j_R]_\LL$
from the representation ring of $\su$ to the ring of polynomials
$\ZZ[\LL^{1/2},\LL^{-1/2}]$ is a ring homomorphism.  This allows one to
easily multiply the expressions $[j_R]_\LL$.  For example, we have
\[
\left[j_R\right]_\LL\left[\frac12\right]_\LL=
\left[j_R+\frac12\right]_\LL+\left[j_R-\frac12\right]_\LL
\]
by the analogous identity in the representation ring of $\su$.

We prove the claim by induction on $\beta$ and $k$.  We start from small $\beta$,
choosing any ordering of the $\beta$ refining
the partial ordering
$\beta_1\ge\beta_2$ if $\beta_1-\beta_2$ is effective, and proceed inductively
using this ordering of the $\beta$.

The for fixed $\beta$, we look at which factors in (\ref{motivicproduct}) can contribute to
$P^\mathrm{r}_{k,\beta}$.  We have $k=m-2m_L$.  The smallest value
of $k$ that can occur is $k=1-p_a$, for $m=1$ and $m_L=p_a/2$.
We then compare coefficients of $q^{1-p_a}Q^\beta$ in
(\ref{motivicproduct}) and get
\begin{equation}
PT^r_{1-p_a,\beta}=\sum_{j_R}\left(-1\right)^{2j_R}N^\beta_{p_a/2,j_R}[j_R]_{\LL},
\label{topspinrefined}
\end{equation}
where we have combined the relevant terms in the sum over $m_R$ in
\ref{motivicproduct} to get the factor of $[j_R]_\LL$ in
(\ref{topspinrefined}).
After rewriting $PT^r_{1-p_a,\beta}$ in the $[j_R]_\LL$ basis, the $N^\beta_{p_a/2,j_R}$
are read off immediately from (\ref{topspinrefined}).

Continuing with this fixed $\beta$, we now proceed with a downward induction on $j_L$.
Suppose we have found $N^{\beta'}_{j_L,j_R}$ for all $\beta'<\beta$ and for
$\beta'=\beta$ and all $j_L >J_L$.  We put $k=1-2J_L$ in (\ref{motivicproduct})
and see that
\begin{equation}
PT^r_{1-2J_L,\beta}=\sum_{j_R}\left(-1\right)^{2j_R}N^\beta_{J_L,j_R}[j_r]_{\LL}+S,
\label{allspinrefined}
\end{equation}
where the omitted terms $S$ only depend on the
$N^{\beta'}_{j_L,j_R}$ for all $\beta'<\beta$ and for
$\beta'=\beta$ and all $j_L >J_L$, which are known by the inductive procedure.  After
substituting in these known quantities and expressing $PT^r_{1-2J_L,\beta}$ in the
$[j_R]_\LL$ basis, the $N^\beta_{J_L,j_R}$
are read off immediately from (\ref{allspinrefined}).

\smallskip\noindent
{\bf Remark.}  It is not immediately obvious from the geometric procedure why the
$N^\beta_{j_L,j_R}$ should be nonnegative integers.

\smallskip
For local $\PP^1$,  it is well known that the only nonvanishing BPS invariant
is $N^1_{0,0}=1$.  In that case, (\ref{motivicproduct}) simplifies to
the motivic product formula of \cite{mmns}.

By the computer program, we have checked \eqref{motivicproduct} for local $\PP^2$ when
\begin{align*}
  d=3 \text{ and } n\le 10, \\
  d=4 \text{ and } n\le 9,\\
  d=5 \text{ and } n\le 1,
\end{align*}
and for local $\PP^1\times\PP^1$ for degree $(d_1,d_2)$ with $d_1+d_2\le 5$ and $n\le 6$.

\bigskip
We now take up the case where $X$ is
local $\IP^2$ and show how the refined PT invariants can be
used to deduce the BPS invariants.

\subsection{Refined invariants for local $\mathbb{P}^2$}

We now illustrate the method for low degree in $\PP^2$.
For curves of degree $d$, we have $p_a=p_a(d)=(d-1)(d-2)/2$. We will
assume that $n^g_{d}=0$ for $g>p_a(d)$.  As
explained in Section~\ref{ptinvariants},
the PT-moduli spaces will be
equal to the relative Hilbert schemes in the cases discussed below, so the
PT invariants can be calculated by hand.  We will also repeatedly use
the fact that
the moduli space of curves of degree $d$ on $\PP^2$ is
isomorphic to $\PP^{(d^2+3d)/2}$, with virtual Bialynicki-Birula
decomposition or equivariant index
$[(d^2+3d)/4]_\LL$.

\bigskip\noindent
$d=1$.  We have $p_a=0$, so we get from (\ref{topspinrefined}),
an application of
(\ref{smoothvirmotive}) and either $P_1(X,1)\simeq\PP^2$ or the
computation of Section~\ref{pttoric}
\[
\left[1\right]_\LL=P^{\mathrm{r}}_{1,1}=\sum_{j_R}(-1)^{2j_R}N^1_{0,j_R}[j_R]_\LL,
\]
which leads immediately to
\[
N^1_{0,j_R}=\left\{
\begin{array}{cl}
1 & j_R=1\\
0 & {\rm otherwise}
\end{array}\right. ,
\]
in agreement with the B-model result.

\bigskip\noindent
$d=2$.  We have $p_a=0$, so we get from (\ref{topspinrefined}), an
application of
(\ref{smoothvirmotive}), and either $P_1(X,2)\simeq\PP^5$ or the computation of Section~\ref{pttoric}
\[
-\left[\frac52\right]_\LL=P^{\mathrm{r}}_{1,2}=\sum_{j_R}\left(-1\right)^{2j_R}N^2_{0,j_R}[j_R]_\LL,
\]
which leads immediately to
\[
N^2_{0,j_R}=\left\{
\begin{array}{cl}
1 & j_R=5/2\\
0 & {\rm otherwise}
\end{array}\right.
\]
in agreement with the B-model results.

\bigskip\noindent
$d=3$.  We have $p_a=1$, so we get from (\ref{topspinrefined}),
an application of
(\ref{smoothvirmotive}), and either $P_0(X,3)\simeq\PP^9$ or the computation of Section~\ref{pttoric}
\[
-\left[\frac92\right]_\LL=P^{\mathrm{r}}_{0,3}=\sum_{j_R}\left(-1\right)^{2j_R}
N^3_{1/2,j_R}[j_R]_\LL,
\]
which leads immediately to
\begin{equation}
N^3_{1/2,j_R}=\left\{
\begin{array}{cl}
1 & j_R=9/2\\
0 & {\rm otherwise}
\end{array}\right. .
\label{d3spin12}
\end{equation}
Next, expanding the coefficient of $qQ^3$ in (\ref{motivicproduct}) gives for
(\ref{allspinrefined})
\begin{equation}
[5]_\LL+[4]_\LL+[3]_\LL=P^{\mathrm{r}}_{1,3}=
\sum_{j_R}\left(\left(-1\right)^{2j_R}N^3_{0,j_R}[j_R]_\LL+
\left(-1\right)^{2j_R+1}N^3_{1/2,j_R}[1/2]_\LL[j_R]_\LL\right),
\label{d3chi1first}
\end{equation}
where the factor of $[1/2]_\LL$ comes from combining the terms of the sum
over $m_L$ in (\ref{motivicproduct}) when $j_L=1/2$.
The left hand side of (\ref{d3chi1first}) arises since $P_1(X,3)$ is a
$\PP^8$ bundle over $\PP^2$ and $[4][1]=[5]+[4]+[3]$, or by the computations of
Section~\ref{pttoric}.  Then (\ref{d3chi1first})
simplifies to
\begin{equation}
[5]_\LL+[4]_\LL+[3]_\LL=
\sum_{j_R}\left(-1\right)^{2j_R}N^3_{0,j_R}[j_R]_\LL+
\left[\frac12\right]_\LL\left[\frac92\right]_\LL,
\label{d3chi1}
\end{equation}
where we have used (\ref{d3spin12}) to get (\ref{d3chi1}).  Since
$[1/2][9/2]=[5]+[4]$, we conclude that
\begin{equation}
N^3_{0,j_R}=\left\{
\begin{array}{cl}
1 & j_R=3\\
0 & {\rm otherwise}
\end{array}\right. ,
\label{d3spin0}
\end{equation}
in agreement with the B-model results.

Note that the last equality in (\ref{d3chi1}) agrees with (\ref{allspinrefined}) with
$S=[1/2]_\LL[9/2]_\LL$ (this last expression only requiring the knowledge of
$N^3_{1/2,j_R}$ and not any other $N^\beta_{j_L,j_R}$).

The higher degrees are done in essentially the same way.
For degrees 4 and 5 as with degree
3, the term $S$ in the proof of the claim can actually
be written as a sum of products
of various $[j_R]_\LL$.  In Section~\ref{kkvrefined} we will explain this
observation by
using a refinement of the KKV method.  In general, the calculation can be
implemented on a computer without the need for this algebraic simplification.
In general, the refined PT invariants can be computed algorithmically
by the torus localization method of Section~\ref{pttoric}.

\subsection{Refined invariants for local $\PP^1\times\PP^1$}
\label{p1p1refined}

For curves of bidegree $(d_1,d_2)$, we have $p_a=p_a(d_1,d_2)=(d_1-1)(d_2-1)$. We will
assume that $n^g_{d_1,d_2}=0$ for $g>p_a(d_1,d_2)$.  As
explained in Section~\ref{ptinvariants},
the PT-moduli spaces will be
equal to the relative Hilbert schemes in the cases discussed below, so the
PT invariants can be calculated by hand.  We will also repeatedly use that
the moduli space of curves of bidegree $(d_1,d_2)$ on $\PP^1\times\PP^1$ is
isomorphic to $\PP^{(d_1+1)(d_2+1)-1}$, with virtual Bialynicki-Birula
decomposition or equivariant index
$[((d_1+1)(d_2+1)-1)/2]_\LL$.

We will find it convenient to rewrite $Q^\beta$ with $\beta=(d_1,d_2)$ as
$Q_1^{d_1}Q_2^{d_2}$.

\bigskip\noindent
$(d_1,d_2)=(0,1)$.  We have $p_a=0$, so we get from (\ref{topspinrefined}) and an application of
(\ref{smoothvirmotive}) or the computation of Section~\ref{pttoric}
\[
-\left[\frac12\right]_\LL=P^{\mathrm{r}}_{1,(0,1)}=\sum_{j_R}(-1)^{2j_R}
N^{0,1}_{0,j_R}[j_R]_\LL,
\]
which leads immediately to
\[
N^{0,1}_{0,j_R}=\left\{
\begin{array}{cl}
1 & j_R=1/2\\
0 & {\rm otherwise}
\end{array}\right. ,
\]
in agreement with the B-model result.

\bigskip\noindent
$(d_1,d_2)=(1,0)$.  By symmetry we see immediately that
\[
N^{1,0}_{0,j_R}=\left\{
\begin{array}{cl}
1 & j_R=1/2\\
0 & {\rm otherwise}
\end{array}\right. ,
\]

\bigskip\noindent
$(d_1,d_2)=(1,d_2), d_2\ge 1$.  We have $p_a=0$ and the moduli space of
these curves is $\PP^{2d_2+1}$, so we get as above
\[
-\left[d_2+\frac12\right]_\LL=P^{\mathrm{r}}_{1,(1,d_2)}=
\sum_{j_R}(-1)^{2j_R}
N^{1,d_2}_{0,j_R}[j_R]_\LL,
\]
which leads immediately to
\[
N^{1,d_2}_{0,j_R}=\left\{
\begin{array}{cl}
1 & j_R=d_2+1/2\\
0 & {\rm otherwise}
\end{array}\right. ,
\]
in agreement with the B-model results as far as we have checked.

\bigskip\noindent
$(d_1,d_2)=(d_1,1), d_1\ge 1$.  By symmetry we have immediately
\[
N^{d_1,1}_{0,j_R}=\left\{
\begin{array}{cl}
1 & j_R=d_1+1/2\\
0 & {\rm otherwise}
\end{array}\right. .
\]

\bigskip\noindent
$(d_1,d_2)=(2,2)$.  Now $p_a=1$.  For $j_L=1/2$, we can again use
(\ref{topspinrefined}) and the $\PP^8$ moduli space to get
\[
\left[4\right]_\LL=P^{\mathrm{r}}_{0,(2,2)}=
\sum_{j_R}(-1)^{2j_R}
N^{2,2}_{1/2,j_R}[j_R]_\LL,
\]
which leads immediately to
\[
N^{2,2}_{1/2,j_R}=\left\{
\begin{array}{cl}
1 & j_R=4\\
0 & {\rm otherwise}
\end{array}\right. .
\]
For $j_L=0$, we need to use the product formula and examine the coefficient
of $qQ_1^2Q_2^2$ using the above result for the $N^{2,2}_{j_L,j_R}$.

The moduli space $P_1(X,(2,2))$ is a $\PP^7$ bundle over $\PP^1\times\PP^1$,
with virtual Bialynicki-Birula decomposition or
equivariant index
\begin{equation}
\left[\frac72\right]_\LL\left[\frac12\right]_\LL\left[\frac12\right]_\LL=
\left[\frac92\right]_\LL+2\left[\frac72\right]_\LL+\left[\frac52\right]_\LL.
\label{p122}
\end{equation}

The expansion gives for the coefficient of $qQ_1^2Q_2^2$
\begin{equation}
\left[4\right]_\LL\left[\frac12\right]_\LL+\sum_{j_R}(-1)^{2j_R}
N^{2,2}_{0,j_R}\left[j_R\right]_\LL
\label{co122}
\end{equation}
Equating (\ref{p122}) and (\ref{co122}) gives
\[
N^{2,2}_{0,j_R}=\left\{
\begin{array}{cl}
1 & j_R=7/2,\ 5/2\\
0 & {\rm otherwise}
\end{array}\right. .
\]

As for $\PP^2$, the higher bidegrees are done in essentially the same way,
using the torus localization method of Section~\ref{pttoric} to compute the
refined PT invariants algorithmically by computer instead of by classical
algebraic geometry.

\section{Refinement of the KKV approach}
\label{kkvrefined}

%-----------------------------------------------------------------------------
\subsection{Local $\PP^2$ and Asymptotic Behavior of the Refined PT-invariants}
%-----------------------------------------------------------------------------

The basic idea is that for $k<d-1$, the method of \cite{kkv} says that
we can compute the refined PT invariants $PT^r_{1-p_a+k,d}$
by elementary projective geometry, without the need for the more elaborate toric computation
of Section~\ref{pttoric}.  In addition to the increased simplicity, infinite collections
of results can be put into closed form, providing asymptotic formulae.

Recall from the end of Section~\ref{sec:kkv} that $\cC^{[k]}$ is smooth for $k$ in this
range.  Hence the virtual BB decomposition arises from the usual BB decomposition, and
therefore the refined invariants arise from the Lefschetz action, as was already
well known in the physics literature.

Furthermore, since correction terms come from products of PT invariants of lower degree,
and for any degree $d$ the minimum holomorphic euler characteristic that can occur is
$1-p_a(d)=1-(d-1)(d-2)/2$, we see that the minimum holomorphic euler characteristic with a
correction term is
\begin{equation}
\mathrm{min}\left\{\left(1-p_a(d_1)\right)+\left(1-p_a(d_2)\right)\mid d_1+d_2=d\right\}.
\label{mincorr}
\end{equation}
The minimum of (\ref{mincorr}) occurs when $d_1=1$ or $d_2=1$, in which case it simplifies
to $1-p_a(d)+(d-1)$.  We conclude that there are no correction terms if $k<d-1$, in
which case (\ref{eq:kkvid}) simplifies to
\begin{equation}
H^*\left(\cC^{[k]}\right)=\left(\theta^{p_a-k}\hat{\cH}_\beta\right)_{\su_{\Delta}}\oplus H^*\left(\cC^{[k-2]}\right).
\label{simpkkvid}
\end{equation}

Furthermore (\ref{simpkkvid}), with $\theta^{p_a-k}\hat{\cH}_\beta$ interpreted as being
defined by the refined PT invariants, is rigorously proven by Proposition~\ref{kkvprop}.

To apply (\ref{simpkkvid}), we only need to compute the $PT^r_{1-p_a+k,d}$ for $k<d-1$.
But this is easy: $P_{1-p_a+k}(X,d)$ is a $\PP^{d(d+3)/2-k}$-bundle over
$(\PP^2)^{[k]}$.  So its Lefschetz representation is immediately computed as a product
of the Lefschetz representations of $\PP^{d(d+3)/2-k}$ and $(\PP^2)^{[k]}$.  This is
an equivariant/motivic extension of the method of Section~\ref{sec:kkv}.

Before turning to the asymptotic formulae, we illustrate with low degree examples.  All
the results agree with the B-model methods and the computations of Section~\ref{refinedtobps}.

$d=1$.
Putting $g=0, k=0$ in (\ref{simpkkvid}) we get
\[
\left(\hat{\cH}_1\right)_{\su_{\Delta}}=H^*(\cC^{0})=H^*(\IP^2).
\]
The Lefschetz of $H^*(\IP^2)$ is $[1]$.  Since the left spin can only be $[0]$, we
conclude that the representation is $\hat{\cH}_1=[0,1]$.  Since $\mathrm {Tr}(-1)^{F_R}\hat{\cH}_1
=3[0]=3I_0$, we obtain $n^0_1=3$ and $n^g_1=0$ for $g>0$ for the
GV invariants.

\medskip\noindent
$d=2$.  We similarly have
\[
\left(\hat{\cH}_2\right)_{\su_{\Delta}}=H^*(\cC^{0})=H^*(\IP^5).
\]
we conclude that $\hat{\cH}_1=[0,5/2]$. It follows that $n^0_2=-6$ and $n^g_2=0$ for $g>0$ for the
GV invariants.

\medskip\noindent
{$d=3$.}
Now we can have a left spin of $1/2$.
First (\ref{simpkkvid}) gives
\[
\left(\theta\hat{\cH}_3\right)_{\su_{\Delta}}=H^*(\cC^{[0]})=H^*(P_0(X,3))=H^*(\IP^9).
\]
The Lefschetz of $\IP^9$ is $[9/2]$.
We conclude that
\[
\hat{\cH}_3=[1/2,9/2]\oplus[0,R_0]
\]
for some representation $R_0$ to be determined.

We now apply (\ref{simpkkvid}) again and get
\[
\left(\hat{\cH}_3\right)_{\su_{\Delta}}=H^*(\cC^{[1]})=H^*(P_1(X,3)).
\]
The Lefschetz of this $\IP^8$-bundle over $\IP^2$ is
\[
[4]\otimes[1]=[5]\oplus[4]\oplus[3].
\]
Restricting $[1/2,9/2]\oplus[0,R_0]$ to ${\su}_\Delta$ gives
\[
[1/2]\otimes[9/2]\oplus[0]\otimes[R_0]=[5]\oplus[4]\oplus
R_0
\]
Comparing, we see that $R_0=3$ and conclude that
\[
\hat{\cH}_3=[1/2,9/2]\oplus[0,3].
\]
Then
\[\mathrm {Tr}(-1)^{F_R}\hat{\cH}_3 = - 10 [1/2] + 7 [0] = -10 I_1 + 27 I_0.\]
Therefore
\[
n^g_{3}=\left\{
\begin{array}{cl}
27 & g=0\\
-10 & g=1\\
0 & g\ge2
\end{array}\right.
\]
for the GV invariants.

\medskip\noindent
{$d=4$.}
We start with
\[
\left(\theta^3\hat{\cH}_4\right)_{\su_{\Delta}}=H^*(\cC^{[0]})=H^*(P_{-2}(X,4))=H^*(\IP^{14})
\]
and since $H^*(\IP^{14})=[7]$ we see that
\[
\cH_4=[3/2,7]+[1,R_1]+[1/2,R_{1/2}]+[0,R_0].
\]
Then we get
\[
\theta^2\hat{\cH_4}=H^*(\cC^{[1]})=H^*(P_{-1}(X,4)).
\]
We have seen that $\cC^{[1]}$ is a $\IP^{13}$-bundle over $\IP^2$, with
Lefschetz
\[
[13/2]\otimes[1]=[15/2]+[13/2]+[11/2].
\]
This must be equal to the restriction of
\[
\theta^2\left([3/2,7]+[1,R_1]+[1/2,R_{1/2}]+[0,R_0]\right)=
[1/2,7]\oplus [0,R_1]
\]
to the diagonal, which is
\[([1/2]\otimes[7])+([0]\otimes R_1) =[15/2]+[13/2]+R_1.\]
We infer that $R_1=[11/2]$.

To get $R_{1/2}$, we use
$\hat{\cH_4}=[3/2,7]+[1,11/2]+[1/2,R_{1/2}]+[0,R_0]$ and
\[
\left(\theta\hat{\cH}_4\right)_{\su_{\Delta}}=H^*(\cC^{[2]})-H^*(\cC^{[0]})=H^*(P_{0}(X,4))-
H^*(P_{-2}(X,4)).
\]
Now $P_{0}(X,4)$ is a $\IP^{12}$-bundle over $(\IP^2)^{[2]}$.  The
Betti numbers of the Hilbert scheme are found from
\[
\sum_{m,n} H^{m}\left(\left(\IP^2\right)^{[n]}\right)s^n y^m=\left(
\left(1-y^{2m-2}t^m\right)\left(1-y^{2m}t^m\right)
\left(1-y^{2m+2}t^m\right)\right)^{-1}.
\]
The Lefschetz $\su$ of $(\IP^2)^{[2]}$ is easily deduced from the Betti
numbers of $(\IP^2)^{[2]}$ as $[2]+[1]+[0]$.  So the Lefschetz of $P_0(X,4)$
is
\[
[6]\otimes\left([2]+[1]+[0]\right)=[8]+2[7]+3[6]+2[5]+[4].
\]
Thus
\[
H^*(\cC^{[2]})-H^*(\cC^{[0]})=[8]+[7]+3[6]+2[5]+[4].
\]
Comparing to the restriction to the diagonal of $[1,7]\oplus[1/2,11/2]\oplus
[0,R_{1/2}]$, which is
\[
=[1]\otimes[7]+[1/2]\otimes[11/2]+[0]\otimes R_{1/2}=
[8]+[7]+2[6]+[5]+R_{1/2},
\]
we conclude that $R_{1/2}=[6]+[5]+[4]$.

This is as far as we can get from (\ref{simpkkvid}) for $d=4$.  We briefly digress from
our main development to show how we can complete the calculation by reverting to
(\ref{eq:kkvid}) and computing equivariantly.

We have
\begin{equation}
\left(\hat{\cH}_4\right)_{\su_{\Delta}}=H^*(P_1(X,4))-H^*(P_{-1}(X,4))-H^*(P_0(X,3)\times P_1(X,1)).
\label{eq:h4}
\end{equation}
In \cite{kkv}, the product $P_0(X,3)\times P_1(X,1)$ appeared as a correction
term, but now we see from stable pairs theory that it is a natural
occurrence.  Now $P_1(X,4)$ is a $\IP^{11}$-bundle over $(\IP^2)^{[3]}$.
The Lefschetz $\su$ of $(\IP^2)^{[3]}$ is easily deduced from
the Betti numbers of $(\IP^2)^{[3]}$ as $[3]+[2]+3[1]+[0]$. So we get the
Lefschetz of $P_1(X,4)$ as
\[[\frac{11}2]\otimes\left([3]+[2]+3[1]+[0]\right)\]
\[=[\frac{17}2]+2[\frac{15}2]+5[\frac{13}2]+6[\frac{11}2]+5[\frac{9}2]+2[\frac{7}2]+[\frac{5}2]
\]
and therefore from (\ref{eq:h4}) and previously computed
representations we get
\[
\left(\hat{\cH}_4\right)_{\su_{\Delta}}=[\frac{17}2]+[\frac{15}2]+4[\frac{13}2]+4[\frac{11}2]+4[\frac{9}2]+[\frac{7}2]+[\frac{5}2].
\]
We have to compare to the diagonal restriction, which is
$=[\frac32]\otimes[7]+[1]\otimes[\frac{11}2]+[\frac12]\otimes\left([6]+[5]+4]\right)+[0]\otimes
R_0$.

We solve to get $R_0=[13/2]+[9/2]+[5/2]$, and conclude

\[
\hat{\cH}_4=[3/2,7]+[1,11/2]+[1/2,6+5+4]+[0,13/2+9/2+5/2].
\]

Then
\[\mathrm {Tr}(-1)^{F_R}\hat{\cH}_4 = 15 [3/2] -12 [1] + 33 [1/2] - 30 [0] = 15 I_3 -102 I_2 +231I_1 -192 I_0.\]
Therefore
\[
n^g_{4}=\left\{
\begin{array}{cl}
-192 & g=0\\
231 & g=1\\
-102 & g=2\\
15 & g=3\\
0 & g\ge4
\end{array}\right.
\]
for the GV invariants.

We can now generalize the above computations to compute asymptotic
formulae for the $\su\times\su$ invariants for large $d$.  For degree
$d$, the maximum genus is $g=g(d)=(d-1)(d-2)/2$, so the maximum left
spin is $\left[\frac{g}{2}\right]$.
The basic idea is that for
fixed $k$, we have that for $d$ sufficiently large,
$P_{1-g(d)+k}(X,d)$ is a $\PP^{d(d+3)/2-k}$ bundle over
the Hilbert scheme $(\PP^2)^{[k]}$, so computations can be done uniformly in
$d$.

We start by writing the $\su_L\times\su_R$
representation as
\[
\sum_{i=0}^g\left[\frac{i}{2},R_i\right]
\]
and solving for $R_i$ in decreasing order.

We have
\[
\theta^g\hat{\cH}_d=[0,R_g]=\cC^{[0]}=\IP^{d(d+3)/2},
\]
which has Lefschetz representation $[d(d+3)/4]$.  To simplify notation,
let us define $D:=d(d+3)/2$.  This gives

\begin{equation}
R_{g}=\left[\frac{D}2\right].
\label{eq:rg}
\end{equation}

Note that the bottom row in Table~\ref{bpstable} contains
only the representation $\left[g/2,D/2\right]$ with multiplicity one,
agreeing with (\ref{eq:rg}).

For the second row from the bottom we have from (\ref{simpkkvid})
\[
\theta^{g-1}\hat{\cH}_d=[0,R_{g-1}]+\left[\frac12,R_g\right]=\cC^{[1]}.
\]
By the usual argument, the universal curve is a $\IP^{D-1}$-bundle over $\IP^2$,
with Lefschetz
\[
\left[1\right]\otimes\left[\frac{D-1}2\right]=\left[\frac{D+1}2\right]+\left[\frac{D-1}2\right]+\left[\frac{D-3}2\right],
\]
Restricting $\left[0,R_{g-1}\right]+\left[1/2,R_g\right]$ to the diagonal and using (\ref{eq:rg}) gives
\[
R_{g-1}+\left[\frac{D+1}2\right]+\left[\frac{D-1}2\right].
\]
Equating these last two expressions gives
\begin{equation}
R_{g-1}=\left[\frac{D-3}2\right],
\label{eq:rg1}
\end{equation}
the asymptotic expression for the second to the bottom row (valid for $d\ge3$).

For the next row we have
\begin{equation}
\theta^{g-2}\hat{\cH}_d=\left[0,R_{g-2}\right]+\left[\frac12,R_{g-1}\right]+\left[1,R_g\right]=\cC^{[2]}-\cC^{[0]}.                                                           \label{eq:3rdrow}
\end{equation}
By the usual argument, $\cC^{[2]}$ is a $\IP^{D-2}$-bundle over $(\IP^2)^{[2]}$,
with Lefschetz
\[
\left([2]+[1]+[0]\right)\otimes\left[\frac{D-2}2\right]=\left[\frac{D+2}2\right]+2\left[\frac{D}2\right]+3\left[\frac{D-2}2\right]+2\left[\frac{D-4}2\right]+\left[\frac{D-6}2\right],
\]
so the right hand side of (\ref{eq:3rdrow}) is
\[
\left([2]+[1]+[0]\right)\otimes\left[\frac{D-2}2\right]=\left[\frac{D+2}2\right]+\left[\frac{D}2\right]+3\left[\frac{D-2}2\right]+2\left[\frac{D-4}2\right]+\left[\frac{D-6}2\right].
\]
Restricting the left  hand side of (\ref{eq:3rdrow}) to the diagonal and using (\ref{eq:rg}) and (\ref{eq:rg1}) gives
\[
R_{g-2}+\left(\left[\frac{D-2}2\right]+\left[\frac{D-4}2\right]\right)+\left(\left[\frac{D+2}2\right]+\left[\frac{D}2\right]+\left[\frac{D-2}2\right]\right).
\]
Equating these last two formulas and solving gives
\begin{equation}
R_{g-2}=\left[\frac{D-2}2\right]+\left[\frac{D-4}2\right]+\left[\frac{D-6}2\right],
\label{eq:rg2}
\end{equation}
the third asymptotic row from the bottom, valid for $d\ge4$.

\begin{table}[h]
\centering{
\begin{tabular}[h]{|c|rrrrrrrrrrrrrrrrrrrrrrrrr|}
\hline
$i/k$&\!\!\!-12\!\!\!&\!\!\!-$\frac{23}{2}$\!\!\!&\!\!\!-11
\!\!\!&\!\!\!-$\frac{21}{2}$\!\!\!&\!\!\!-10\!\!\!&\!\!\!-$\frac{19}{2}$\!\!\!&\!\!\!-9\!\!\!&\!\!\!-$\frac{17}{2}$\!\!\!&\!\!\!-8\!\!\!&\!\!\!-$\frac{15}{2}$\!\!\!
&\!\!\!-7\!\!\!&\!\!\!-$\frac{13}{2}$\!\!\!&\!\!\!-6\!\!\!&\!\!\!-$\frac{11}{2}$\!\!\!&\!\!\! -5\!\!\!&\!\!\! -$\frac{9}{2}$\!\!\!&\!\!\! -4\!\!\!&\!\!\! -$\frac{7}{2}$\!\!\!&\!\!\! -3\!\!\!&\!\!\! -$\frac{5}{2} $\!\!\!
&\!\!\! -2\!\!\!&\!\!\! -$\frac{3}{2}$\!\!\!&\!\!\! -1\!\!\! &\!\!\! -$\frac{1}{2}$\!\!\!&\!\!\!0 \\
\hline
$ 0 $&\!\!\!     \!\!\!&\!\!\!         \!\!\!&\!\!\!    \!\!\!&\!\!\!        \!\!\!&\!\!\!       \!\!\!&\!\!\!     \!\!\!&\!\!\!   \!\!\!&\!\!\!   \!\!\!&\!\!\!  \!\!\!&\!\!\!    \!\!\!
&\!\!\!     \!\!\!&\!\!\!     \!\!\!&\!\!\!   \!\!\!&\!\!\!     \!\!\!&\!\!\!  \!\!\!
&\!\!\!        \!\!\!&\!\!\!        \!\!\!&\!\!\!      \!\!\!&\!\!\!      \!\!\!&\!\!\!      \!\!\!
&\!\!\!         \!\!\!&\!\!\!    \!\!\!&\!\!\!    \!\!\!&\!\!\!   \!\!\!& \!\!\!    1 \\
$ -1 $&\!\!\!     \!\!\!&\!\!\!         \!\!\!&\!\!\!    \!\!\!&\!\!\!        \!\!\!&\!\!\!       \!\!\!&\!\!\!     \!\!\!&\!\!\!   \!\!\!&\!\!\!   \!\!\!&\!\!\!  \!\!\!&\!\!\!    \!\!\!
&\!\!\!     \!\!\!&\!\!\!     \!\!\!&\!\!\!   \!\!\!&\!\!\!     \!\!\!&\!\!\!  \!\!\!
&\!\!\!        \!\!\!&\!\!\!        \!\!\!&\!\!\!      \!\!\!&\!\!\!      \!\!\!&\!\!\!      \!\!\!
&\!\!\!         \!\!\!&\!\!\!  1  \!\!\!&\!\!\!    \!\!\!&\!\!\!   \!\!\!& \!\!\!     \\
$ -2 $&\!\!\!     \!\!\!&\!\!\!         \!\!\!&\!\!\!    \!\!\!&\!\!\!        \!\!\!&\!\!\!       \!\!\!&\!\!\!     \!\!\!&\!\!\!   \!\!\!&\!\!\!   \!\!\!&\!\!\!  \!\!\!&\!\!\!    \!\!\!
&\!\!\!     \!\!\!&\!\!\!     \!\!\!&\!\!\!   \!\!\!&\!\!\!     \!\!\!&\!\!\!  \!\!\!
&\!\!\!        \!\!\!&\!\!\!        \!\!\!&\!\!\!      \!\!\!&\!\!\!   1   \!\!\!&\!\!\!      \!\!\!
&\!\!\!     1    \!\!\!&\!\!\!    \!\!\!&\!\!\!  1  \!\!\!&\!\!\!   \!\!\!& \!\!\!     \\
$ -3 $&\!\!\!     \!\!\!&\!\!\!         \!\!\!&\!\!\!    \!\!\!&\!\!\!        \!\!\!&\!\!\!       \!\!\!&\!\!\!     \!\!\!&\!\!\!   \!\!\!&\!\!\!   \!\!\!&\!\!\!  \!\!\!&\!\!\!    \!\!\!
&\!\!\!     \!\!\!&\!\!\!     \!\!\!&\!\!\!   \!\!\!&\!\!\!     \!\!\!&\!\!\!  \!\!\!
&\!\!\!     1   \!\!\!&\!\!\!        \!\!\!&\!\!\!   1   \!\!\!&\!\!\!      \!\!\!&\!\!\! 2     \!\!\!
&\!\!\!         \!\!\!&\!\!\! 1   \!\!\!&\!\!\!   \!\!\!&\!\!\! 1  \!\!\!& \!\!\!     \\
$ -4 $&\!\!\!     \!\!\!&\!\!\!         \!\!\!&\!\!\!    \!\!\!&\!\!\!        \!\!\!&\!\!\!       \!\!\!&\!\!\!     \!\!\!&\!\!\!   \!\!\!&\!\!\!   \!\!\!&\!\!\!  \!\!\!&\!\!\!    \!\!\!
&\!\!\!     \!\!\!&\!\!\!     \!\!\!&\!\!\!   \!\!\!&\!\!\!     \!\!\!&\!\!\! 1 \!\!\!
&\!\!\!        \!\!\!&\!\!\!     3   \!\!\!&\!\!\!      \!\!\!&\!\!\!    3  \!\!\!&\!\!\!      \!\!\!
&\!\!\!       3  \!\!\!&\!\!\!    \!\!\!&\!\!\!  1 \!\!\!&\!\!\!   \!\!\!& \!\!\!     \\
$ -5 $&\!\!\!     \!\!\!&\!\!\!         \!\!\!&\!\!\!    \!\!\!&\!\!\!        \!\!\!&\!\!\!       \!\!\!&\!\!\!     \!\!\!&\!\!\!   \!\!\!&\!\!\!   \!\!\!&\!\!\!  \!\!\!&\!\!\!  1  \!\!\!
&\!\!\!     \!\!\!&\!\!\!  1   \!\!\!&\!\!\!   \!\!\!&\!\!\!   3  \!\!\!&\!\!\!  \!\!\!
&\!\!\!     4   \!\!\!&\!\!\!        \!\!\!&\!\!\!   6   \!\!\!&\!\!\!      \!\!\!&\!\!\! 4     \!\!\!
&\!\!\!         \!\!\!&\!\!\! 2   \!\!\!&\!\!\!   \!\!\!&\!\!\!   \!\!\!& \!\!\!     \\
$ -6 $&\!\!\!     \!\!\!&\!\!\!         \!\!\!&\!\!\!    \!\!\!&\!\!\!        \!\!\!&\!\!\!       \!\!\!&\!\!\!     \!\!\!&\!\!\!   \!\!\!&\!\!\!   \!\!\!&\!\!\! 1 \!\!\!&\!\!\!   \!\!\!
&\!\!\!  1   \!\!\!&\!\!\!     \!\!\!&\!\!\! 3  \!\!\!&\!\!\!     \!\!\!&\!\!\! 5 \!\!\!
&\!\!\!        \!\!\!&\!\!\!     9   \!\!\!&\!\!\!      \!\!\!&\!\!\!   8   \!\!\!&\!\!\!      \!\!\!
&\!\!\!      3   \!\!\!&\!\!\!    \!\!\!&\!\!\! 1  \!\!\!&\!\!\!   \!\!\!& \!\!\!     \\
$ -7 $&\!\!\!     \!\!\!&\!\!\!         \!\!\!&\!\!\!    \!\!\!&\!\!\!    1    \!\!\!&\!\!\!       \!\!\!&\!\!\!   1  \!\!\!&\!\!\!   \!\!\!&\!\!\! 3  \!\!\!&\!\!\!  \!\!\!&\!\!\! 5  \!\!\!
&\!\!\!     \!\!\!&\!\!\! 10    \!\!\!&\!\!\!   \!\!\!&\!\!\!  13   \!\!\!&\!\!\!  \!\!\!
&\!\!\!      16  \!\!\!&\!\!\!        \!\!\!&\!\!\!  11    \!\!\!&\!\!\!      \!\!\!&\!\!\!    6  \!\!\!
&\!\!\!         \!\!\!&\!\!\!   1 \!\!\!&\!\!\!   \!\!\!&\!\!\!   \!\!\!& \!\!\!     \\
$ -8 $&\!\!\!  1   \!\!\!&\!\!\!         \!\!\!&\!\!\!1    \!\!\!&\!\!\!        \!\!\!&\!\!\!   3    \!\!\!&\!\!\!     \!\!\!&\!\!\!  5 \!\!\!&\!\!\!   \!\!\!&\!\!\!  11\!\!\!&\!\!\!   \!\!\!
&\!\!\!   16  \!\!\!&\!\!\!    \!\!\!&\!\!\! 24  \!\!\!&\!\!\!     \!\!\!&\!\!\! 24 \!\!\!
&\!\!\!        \!\!\!&\!\!\!    20    \!\!\!&\!\!\!     \!\!\!&\!\!\!   9   \!\!\!&\!\!\!      \!\!\!
&\!\!\!      3   \!\!\!&\!\!\!    \!\!\!&\!\!\!   \!\!\!&\!\!\!   \!\!\!& \!\!\!     \\
\hline
\end{tabular}
\caption{The asymptotic BPS numbers $N^*_{(g_{max}+i)/2,(D/2+k)}$}}
\label{asymptotictable}
\end{table}

Table~\ref{asymptotictable} gives the asymptotic rows observed in the
B-model calculation with the ordering reversed, top to bottom instead
of bottom to top.  We have just explained the first three rows of
Table~\ref{asymptotictable} from the viewpoint of the refined PT invariants
and found complete agreement,
and we have similarly checked the first six rows and found complete agreement
with the B-model.

\subsection{Local $\PP^1\times\PP^1$ and Asymptotic Behavior of the Refined
PT-invariants}

We now explain the low bidegree cases with $X$ equal to local
$\PP^1\times\PP^1$, using the equivariant refinement of the method of
Section~\ref{sec:kkv}.  Again, we find asymptotic formulae.

As observed earlier in Section~\ref{p1p1refined}, curves of bidegree
$(d_1,d_2)$ have arithmetic genus $p_a=p_a(d_1,d_2)=(d_1-1)(d_2-1)$ and these curves are parametrized by $\PP^{(d_1+1)(d_2+1)-1}$, with
Lefschetz representation  $[((d_1+1)(d_2+1)-1)/2]$.
Without loss of generality, we assume that $d_1\le d_2$.

For $k<d_1$, the method of \cite{kkv} says that
we can compute the refined PT invariants $PT^r_{1-p_a+k,d}$
by elementary projective geometry.

Note that the minimum holomorphic euler characteristic with a
correction term is
\begin{equation}
\mathrm{min}\left\{\left(1-p_a(d_1',d_2')\right)+\left(1-p_a(d_1'',d_2'')
\right)\mid d_1'+d_1''=d_1,\ d_2'+d_2''=d_2\right\}.
\label{mincorrp1p1}
\end{equation}
The minimum of (\ref{mincorrp1p1}) occurs when $(d_1',d_2')=(0,1)$ or $(d_1'',
d_2'')=(0,1)$, in which case it simplifies
to $1-p_a(d_1,d_2)+d_1$.  In particular, there are no correction terms if
$k<d_1$.

Furthermore, it is straightforward to check that $\cC^{[k]}$
is smooth in this range.

To apply (\ref{simpkkvid}), we only need to compute the
$PT^r_{1-p_a+k,(d_1,d_2)}$ for $k<d_1$.
But this is easy: $P_{1-p_a+k}(X,(d_1,d_2))=\cC^{[k]}$ is a
$\PP^{(d_1+1)(d_2+1)-1-k}$-bundle over
$(\PP^2)^{[k]}$.  In particular, it is smooth.
So its Lefschetz representation is immediately computed as a product
of the Lefschetz representations of $\PP^{d(d+3)/2-k}$ and $(\PP^2)^{[k]}$.  This is
an equivariant/motivic extension of the method of Section~\ref{sec:kkv}.

Before turning to the asymptotic formulae, we illustrate with low degree examples.  All
the results agree with the B-model methods and the computations of Section~\ref{refinedtobps}.

\medskip\noindent
$(d_1,d_2)=(0,1)$.  Since $p_a=0$, we have $\hat{\cH}_{(0,1)}=[0,R_0]$ for some
$\su$ representation $R_0$.
Now $P_1(X,(0,1))$ is the moduli space of curves of bidegree $(0,1)$, which is
isomorphic to $\PP^1$ and therefore has Lefschetz representation $[1/2]$.
By Proposition~\ref{kkvprop} with $\beta=(0,1)$ and $r=0$, we get
$\hat{\cH}_{0,(0,1)}=[0,1/2]$.  Since $\mathrm {Tr}(-1)^{F_R}\hat{\cH}_{0,(0,1)}
=-2[0]=-2I_0$, we get $n^0_{(0,1)}=-2$ and $n^g_{(0,1)}=0$ for $g>0$ for the
GV invariants, in agreement with \cite{kkv}, where only the combined invariants
$n^g_d:=\sum_{d_1+d_2=d}n^g_{(d_1,d_2)}$ were computed.

\medskip\noindent
$(d_1,d_2)=(1,1)$.  Since $p_a=0$, we have $\hat{\cH}_{(1,1)}=[0,R_0]$ for some
$\su$ representation $R_0$.
Now $P_1(X,(1,1))$ is the moduli space of curves of bidegree $(1,1)$, which is
isomorphic to $\PP^3$ and therefore has Lefschetz representation $[3/2]$.
This leads as above to $\hat{\cH}_{0,(1,1)}=[0,3/2]$, from which it follows that
$n^0_{(1,1)}=-4$ and $n^g_{(1,1)}=0$ for $g>0$ for the
GV invariants, in agreement with \cite{kkv}.

\medskip\noindent
$(d_1,d_2)=(1,d_2)$.  Since $p_a=0$,
and $P_1(X,(1,d_2))$ is the moduli space of curves of bidegree $(1,d_2)$, which is
isomorphic to $\PP^{2d_2+1}$,
we have as above that $\hat{\cH}_{0,(1,d_2)}=[0,d_2+1/2]$,
from which it follows that
$n^0_{(1,d_2)}=-(2d_2+2)$ and $n^g_{(1,2d_2)}=0$ for $g>0$ for the
GV invariants.

\medskip\noindent
$(d_1,d_2)=(2,2)$.  Now $p_a=1$, so we have $\hat{\cH}_{(2,2)}=[1/2,R_{1/2}]+
[0,R_0]$ for some $\su$ representations $R_{1/2}$ and $R_0$.
We apply Proposition~\ref{kkvprop} with $\beta=(2,2)$ and $r=1$.
Now $P_0(X,(2,2))$ is the moduli space of curves of bidegree $(2,2)$, which is
isomorphic to $\PP^8$, so $R_{1/2}=[4]$.  Since $P_1(X,(2,2))$ is a $\PP^7$-bundle
over $\PP^1\times\PP^1$, its Lefschetz decomposition is
\[
\left[\frac72\right]\left[\frac12\right]\left[\frac12\right]=
\left[\frac92\right]+2\left[\frac72\right]+\left[\frac52\right].
\]
Then Proposition~\ref{kkvprop} with $\beta=(2,2)$ and $r=0$ gives
\[
\left[\frac92\right]+2\left[\frac72\right]+\left[\frac52\right]
=\left[\frac12\right]
\left[4\right]+R_0=\left[\frac92\right]+\left[\frac72\right]+R_0,
\]
so that $R_0=[7/2]+[5/2]$.  Putting this all together, we get
\[
\hat{\cH}_{(2,2)}=\left[\frac12,4\right]+\left[0,\frac72\right]+
\left[0,\frac52\right].
\]
Then
\[
\mathrm{Tr}\left(-1\right)^{F_R}\hat{\cH}_{(2,2)}=9\left[\frac12\right]-14\left[
0\right]=9I_1-32I_0.
\]
It follows that
\[
n^g_{(2,2)}=\left\{
\begin{array}{cl}
-32 & g=0\\
9 & g=1\\
0 & g\ge2
\end{array}\right.
\]
for the
GV invariants, in agreement with \cite{kkv} after combining with
$n^g_{(1,3)}$ (and $n^g_{3,1}$).

\medskip\noindent
$(d_1,d_2)=(2,3)$.  Now $p_a=2$, so we have $\hat{\cH}_{(2,3)}=[1,R_1]+[1/2,R_{1/2}]+
[0,R_0]$ for some $\su$ representations $R_1,\ R_{1/2}$ and $R_0$.
We apply Proposition~\ref{kkvprop} with $\beta=(2,3)$ and $r=2$.
Now $P_{-1}(X,(2,3))$ is the moduli space of curves of bidegree $(2,3)$, which is
isomorphic to $\PP^{11}$, so $R_{1}=[11/2]$.

Since $\cC=P_0(X,(2,3))$ is a $\PP^{10}$-bundle
over $\PP^1\times\PP^1$, its Lefschetz representation is
\[
\left[5\right]\left[\frac12\right]\left[\frac12\right]=
\left[6\right]+2\left[5\right]+\left[4\right].
\]
Then Proposition~\ref{kkvprop} with $\beta=(2,3)$ and $r=1$ gives
\[
\left[6\right]+2\left[5\right]+\left[4\right]
=\left[\frac12\right]
\left[\frac{11}2\right]+R_{1/2}=\left[6\right]+\left[5\right]+R_{1/2},
\]
so that $R_{1/2}=[5]+[4]$.

Finally, we have to compute $\cC^{[2]}=P_1(X,(2,3))$, which is a $\PP^9$ bundle
over $(\PP^1\times\PP^1)^{[2]}$.  While we still have smoothness, we
will nevertheless have a correction term, since in this case $k=2=d_1$.

We compute the Betti numbers of the Hilbert
schemes of $\PP^1\times\PP^1$ by the generating function
\[
\sum_{m,n} H^{m}\left(\left(\PP^1\times\PP^1\right)^{[n]}\right)s^n y^m=
\prod_m\left(
\left(1-y^{2m-2}t^m\right)\left(1-y^{2m}t^m\right)^2
\left(1-y^{2m+2}t^m\right)\right)^{-1}.
\]
This gives the Lefschetz of $(\PP^1\times\PP^1)^{[2]}$ as $[2]+2[1]+3[0]$.
This implies that we get $[9/2]([2]+2[1]+3[0])$ for $P_1(X,(2,3))$, which
expands to
\[
\left[\frac{13}2\right]+3\left[\frac{11}2\right]+6\left[\frac92\right]
+3\left[\frac72\right]+\left[\frac52\right].
\]
Applying Proposition~\ref{kkvprop} with $\beta=(2,3)$ and $r=0$ gives for the
left hand side
\[
\left[\frac{13}2\right]+3\left[\frac{11}2\right]+6\left[\frac92\right]
+3\left[\frac72\right]+\left[\frac52\right]-\left[\frac{11}2\right]=
\left[\frac{13}2\right]+2\left[\frac{11}2\right]+6\left[\frac92\right]
+3\left[\frac72\right]+\left[\frac52\right],
\]
while for the right hand side we get, ignoring correction terms for the moment
\[
\left[1\right]\left[\frac{11}2\right]+\left[\frac12\right]\left(
\left[5\right]+\left[4\right]\right)+R_0=\left[\frac{13}2\right]+
2\left[\frac{11}2\right]+3\left[\frac92\right]+\left[\frac72\right]+R_0
\]
so that $R_0=3[9/2]+2[7/2]+[5/2]$.  However there is a correction due to
reducible curves $C'\cup C''$, where $C'$ and $C''$ have bidegrees $(2,2)$
and $(1,0)$ respectively.  The correction term may be recognized either from
the method of \cite{kkv} or by the product formula (\ref{motivicproduct}).
Either way, the correction is $[1/2][4]=[9/2]+[7/2]$,
coming from the moduli space
$\PP^1\times\PP^8$ of the pair of curves.  This gives the corrected value
\[
R_0=2[9/2]+[7/2]+[5/2].
\]

Putting this all together, we get
\[
\hat{\cH}_{(2,3)}=\left[1,\frac{11}2\right]+\left[\frac12,5\right]
+\left[\frac12,4\right]+2\left[0,\frac92\right]+\left[0,\frac72\right]+
\left[0,\frac52\right].
\]
Then
\[
\mathrm{Tr}\left(-1\right)^{F_R}\hat{\cH}_{(2,3)}=-12\left[1\right]
+20\left[\frac12\right]-34\left[
0\right]=-12I_2+68I_1-110I_0.
\]
It follows that
\[
n^g_{(2,3)}=\left\{
\begin{array}{cl}
-110 & g=0\\
68 & g=1\\
-12&g=2\\
0 & g\ge3
\end{array}\right.
\]
for the
GV invariants, again agreeing with \cite{kkv} for $d=d_1+d_2=5$.

\smallskip We can now turn to the asymptotic formulae for sufficiently
large $(d_1,d_2)$.

We start by writing the $\su_L\times\su_R$
representation as
\[
\sum_{i=0}^{p_a}\left[\frac{i}{2},R_i\right]
\]
and solving for $R_i$ in decreasing order.

We put $D=(d_1+1)(d_2+1)-1$.

We have
\[
\theta^{p_a}\hat{\cH}_d=[0,R_{p_a}]=\cC^{[0]}=\IP^{D},
\]
which has Lefschetz representation $[D/2]$.
This gives

\begin{equation}
R_{p_a}=\left[\frac{D}2\right].
\label{eq:rgp1p1}
\end{equation}

For the second row from the bottom we have from (\ref{simpkkvid})
\[
\theta^{p_a-1}\hat{\cH}_d=[0,R_{p_a-1}]+\left[\frac12,R_{p_a}\right]=\cC^{[1]}.
\]
By the usual argument, the universal curve is a $\IP^{D-1}$-bundle over
$\IP^1\times\PP^1$,
with Lefschetz
\[
\left[\frac{D-1}2\right]\left[\frac12\right]\left[\frac12\right]=
\left[\frac{D+1}2\right]+2\left[\frac{D-1}2\right]+\left[\frac{D-3}2\right].
\]
Restricting $\left[0,R_{p_a-1}\right]+\left[1/2,R_{p_a}\right]$ to the diagonal and using (\ref{eq:rgp1p1}) gives
\[
R_{p_a-1}+\left[\frac{D+1}2\right]+\left[\frac{D-1}2\right].
\]
Equating these last two expressions gives
\begin{equation}
R_{p_a-1}=\left[\frac{D-1}2\right]+\left[\frac{D-3}2\right]
\label{eq:rgp1p11}
\end{equation}
the asymptotic expression for the second to the bottom row.

We content ourselves with one more row; the general cases are similar.
For the next row we have
\begin{equation}
\theta^{p_a-2}\hat{\cH}_{d_1,d_2}=\left[0,R_{p_a-2}\right]+
\left[\frac12,R_{p_a-1}\right]+\left[1,R_{p_a}\right]=\cC^{[2]}-\cC^{[0]}.                                                           \label{eq:3rdrowp1}
\end{equation}
By the usual argument, $\cC^{[2]}$ is a $\IP^{D-2}$-bundle over $(\IP^1\times
\IP^1)^{[2]}$,
with Lefschetz
\[
\left([2]+2[1]+3[0]\right)\otimes\left[\frac{D-2}2\right]=
\left[\frac{D+2}2\right]+3\left[\frac{D}2\right]+6\left[\frac{D-2}2\right]+3\left[\frac{D-4}2\right]+\left[\frac{D-6}2\right],
\]
so the right hand side of (\ref{eq:3rdrowp1}) is
\[
\left[\frac{D+2}2\right]+2\left[\frac{D}2\right]+6\left[\frac{D-2}2\right]+3\left[\frac{D-4}2\right]+\left[\frac{D-6}2\right].
\]
Restricting the left  hand side of (\ref{eq:3rdrowp1}) to the diagonal and using (\ref{eq:rgp1p1}) and (\ref{eq:rgp1p11}) gives
\[
R_{p_a-2}+\left(\left[\frac{D}2\right]+2\left[\frac{D-2}2\right]+\left[\frac{D-4}2\right]\right)+\left(\left[\frac{D+2}2\right]+\left[\frac{D}2\right]+\left[\frac{D-2}2\right]\right).
\]
Equating these last two formulas and solving gives
\begin{equation}
R_{g-2}=3\left[\frac{D-2}2\right]+2\left[\frac{D-4}2\right]+\left[\frac{D-6}2\right],
\label{eq:rgp1p12}
\end{equation}
the third asymptotic row from the bottom.

\subsection{Wall crossing}
In this section, we explain how to understand the correction terms in the refined KKV approach via a wall crossing on stable pairs in the case of local $\PP^2$. In \cite{CC1}, a wall crossing phenomenon in the moduli spaces of stable pairs is studied. We alter the notion of stable pairs by introducing the stability parameter denoted by $\alpha$.
\begin{defn}
Let $\alpha$ be a positive rational number. An $\alpha$-stable pair on $X$ is a pair $(\cF, s)$ of a sheaf $\cF$ and a nonzero section $s\in H^0(\cF)$ such that
\begin{itemize}
\item $\cF$ is of pure of dimension 1
\item For all proper nonzero subsheaves $\cF'$ of $\cF$, we have
\begin{equation}\label{eq:alphass}
\frac{\chi(\cF')+\epsilon(s,\cF')\alpha}{r(\cF')} <  \frac{\chi(\cF)+\alpha}{r(\cF)},
\end{equation}
where $r(\cF)$ is the leading coefficient of Hilbert polynomial $\chi(\cF(m))$ and $\epsilon(s,\cF')=1$ if $s$ factors through $\cF'$ and zero otherwise.
\end{itemize}
\end{defn}

Let $X$ be local $\PP^2$ and let $M^\alpha(d,n)$ denote the moduli space of $\alpha$-stable pairs $(\cF,s)$ on local $\PP^2$ with $\mathrm{ch}_2(\cF)=d$ and $\chi(\cF)=n$. One can see that a stable pair as in Section \ref{ptinvariants} can be considered as an $\alpha$-semistable pair for sufficiently large $\alpha$, which we will denote by $\alpha=\infty$. In other words,  $M^{\infty}(d,n)=P_n(X,d)$. At the other extreme, when $\alpha$ is sufficiently close to zero, or $\alpha=0^+$, a pair $(\cF,s)$ is $\alpha$-stable pair if and only if the sheaf $\cF$ itself is a stable sheaf. Hence, we have a connection to the Hilbert space $\hat{\cH}_d$. We have \cite{CC1}
\[
\left(\hat{\cH}_d\right)_{\su_{\Delta}}= H^*(M^{0^+}(d,1)) - H^*(M^{0^+}(d,-1)).
\]
This formula is very similar to \eqref{updatedkkv} with $r=0$. The only difference is that we have replaced $\alpha=\infty$ with $\alpha=0^+$ and removed all correction terms. We claim that in terms of virtual motives, the correction terms in \eqref{updatedkkv} are exactly the wall crossing contributions from $M^{\infty}(d,n)$ to $M^{0^+}(d,n)$ for $d\le 5$.

Wall crossing occurs at the values of $\alpha$ for which there exist strictly semistable pairs. In general, there are only finitely many such walls and the moduli spaces remain unchanged for values of $\alpha$ in between walls.

If a pair $(\cF,s)$ become strictly semistable, we have an exact sequence of the form
\[0\to (\cF',s')\to (\cF,s)\to (\cF'',s'')\to 0. \]
On the one side of the wall, $(\cF,s)$ is stable and as our stability parameter $\alpha$ passes to the other side of the wall, this exact sequence destabilizes $(\cF,s)$. So, we lose those pairs from the moduli space. Instead, new pairs $(\tilde{\cF},\tilde{s})$ defined by the flipped exact sequence
\[0\to (\cF'',s'')\to (\tilde{\cF},\tilde{s})\to (\cF',s')\to 0 \]
become stable. So, by computing Ext groups corresponding to each exact sequence, we can see what happens as we cross the wall.

\medskip\noindent
$d=1,2,$ and $3$. There were no correction term in the KKV computation and one can easily see that there are no walls by an elementary calculation.

\medskip\noindent
$d=4$. The correction term for $\left(\hat{\cH}_4\right)_{\su_{\Delta}}$ in KKV approach is $-H^*(P_0(X,3)\times P_1(X,1))$, which in terms of the virtual motive is $\left[\frac92\right]_\LL \left[1\right]_\LL$ after an appropriate sign change. By an elementary calculation, one can see there is no wall for $M^{\alpha}(4,-1)$ and a unique wall at $\alpha=3$ for $M^{\alpha}(4,1)$. The strictly semistable pairs in $M^{3}(4,1)$ are of type $$(1,(3,0))\oplus (0,(1,1)),$$
where $(1, (d,n))$ (resp. $(0, (d,n))$) denotes the pairs $(\cF,s)$ with a nonzero (resp. zero) section $s$ and $\mathrm{ch}_2(F)=d$ and $\chi(F)=n$.
Note that $\alpha$-stable pairs of type $(1,(3,0))$ are parametrized by $\PP^9$ for any $\alpha$, and stable pairs of type $(0,(1,1))$ are parametrized by $\PP^2$. By the Riemann-Roch theorem and \cite[Corollary 1.6]{mhe}, we can compute the extension group defined on the category of pairs.
\begin{align}
  \Ext^1((0, (1,1)), (1, (3,0)))\simeq \CC^3 \label{extensionpi}\\
  \Ext^1((1, (3,0)),(0, (1,1)))\simeq \CC^4 \label{extensionpt}
\end{align}
The extension given by an element in \eqref{extensionpt} is stable when $\alpha>3$ and becomes unstable when $\alpha<3$. The extension given by an element in \eqref{extensionpi} behaves in the other way. So, at the wall as we cross from $\alpha=\infty$ to $\alpha=0^+$, the $\PP^3$-bundle on $\PP^9\times\PP^2$ is replaced by $\PP^2$-bundle on $\PP^9\times\PP^2$. This gives a geometric wall crossing contribution $-\LL^3 [\PP^9][\PP^2]$. To get the contribution to the virtual motive, we multiply $(-\LL^{-1/2})^{\dim P_1(X,4)}=-\LL^{-\frac{17}{2}}$, which yields $\left[\frac92\right]_\LL \left[1\right]_\LL$. This matches with the correction term.

\medskip\noindent
$d=5$. The correction term for $\left(\hat{\cH}_5\right)_{\su_{\Delta}}$ is
\begin{align*}
-&H^*(P_{-2}(X,4)\times P_{3}(X,1))-H^*(P_{-1}(X,4)\times P_{2}(X,1))-H^*(P_{0}(X,4)\times P_{1}(X,1))\\
-&H^*(P_{0}(X,3)\times P_{1}(X,3))+H^*(P_{-2}(X,4)\times P_{1}(X,1)).
\end{align*}
Possible wall crossing terms for $M^\alpha(5,1)$ and $M^\alpha(5,-1)$ are as follows.
\begin{center}
\begin{tabular}{|l|l|l|}
\hline
$\alpha$
& Splitting type & Associated correction term \\
\hline
\multicolumn{3}{|l|}{Wall crossing for $M^\alpha(5,1)$} \\
\hline
14&
$(1,(4,-2))\oplus (0,(1,3))$&$-H^*(P_{-2}(X,4)\times P_{3}(X,1))$\\
\hline
9&
$(1,(4,-1))\oplus (0,(1,2))$&$-H^*(P_{-1}(X,4)\times P_{2}(X,1))$\\
\hline
4&
$(1,(4,0))\oplus (0,(1,1))$&$-H^*(P_{0}(X,4)\times P_{1}(X,1))$\\
\hline
$\frac{3}{2}$&
$(1,(3,0))\oplus (0,(2,1))$&$-H^*(P_{0}(X,3)\times P_{1}(X,3))$\\
\hline
\multicolumn{3}{|l|}{Wall crossing for $M^\alpha(5,-1)$} \\
\hline
$6$&
$(1,(4,-2))\oplus (0,(1,1))$&$+H^*(P_{-2}(X,4)\times P_{1}(X,1))$\\
\hline
\end{tabular}
\end{center}

We explain the wall crossing for $M^\alpha(5,1)$ at $\alpha=14$. The computation for other walls is similar.
As before one can compute
\begin{align*}
  \Ext^1((0, (1,3)), (1, (4,-2)))\simeq \CC^4 \\
  \Ext^1((1, (4,-2)),(0, (1,3)))\simeq \CC^7,
\end{align*}
and $\alpha$-stable pairs of type $(1,(4,-2))$ are parametrized by $\PP^{14}$ for any $\alpha$, and stable pairs of type $(0,(1,3))$ are parametrized by $\PP^2$. Hence as we cross wall from $\alpha=\infty$ to $9<\alpha<14$, the $\PP^6$-bundle on $\PP^{14}\times\PP^2$ is replaced by $\PP^3$-bundle on $\PP^{14}\times\PP^2$. Hence, the wall crossing contribution here is $-(\LL^6+\LL^5+\LL^4)[\PP^{14}][\PP^2]$. After multiplying $(-\LL^{-1/2})^{\dim P_1(X,5)}=\LL^{-\frac{26}{2}}$, we get $-\left[1\right]_\LL\left[7\right]_\LL \left[1\right]_\LL$. This matches with the associated correction term $-H^*(P_{-2}(X,4)\times P_{3}(X,1))=-\left[7\right]_\LL\left[1\right]_\LL \left[1\right]_\LL,$ as  $P_{3}(X,1)$ is a $\PP^2$-bundle over $\PP^2$.

\section{Conclusions}

We have described refined stable pair invariants and shown that the
information of those invariants up to a fixed degree is equivalent
to knowing the $\su\times\su$ BPS invariants up to that degree.

Several interesting questions remain for future work.  We conjecture
that there is a purely motivic description of our refined stable pair
invariants which would make clearer the connection of our work
to the motivic stable pair invariants of~\cite{mmns}.  This would
provide an more precise
interpretation of (\ref{motivicproduct}) as a generalization
of the motivic product formulae of~\cite{bbs,mmns}.

It would also be interesting to use other twistings to define new
invariants. Perhaps those new invariants will be related to other related
mathematical
invariants, either those arising from a change in stability condition or
by choosing different quiver descriptions of the PT moduli spaces.

To calculate the 5d BPS index, the B-model approach using
the refined holomorphic anomaly equation combined with the
direct integration approach~\cite{Huang:2009md}\cite{Haghighat:2008gw}
to non-compact geometries is the most efficient method.

For the right choice of the ${\cal R}$ symmetry group it expresses
the 5d index $Z$ e.g.\ for local Calabi-Yau manifolds based on del
Pezzo surfaces in terms of quasimodular forms of subgroups of
$PSL(2,\mathbb{Z})$. It is independent of the question
whether the geometry has a toric realization~\cite{HuangKlemm}.

For local ${\cal O}(-2,-2)\rightarrow \mathbb{P}^1\times\mathbb{P}^1$
the approach yields in one stroke the refined 5d M-theory index,
Nekrasov's partition function of N=2 d=4 theory, the refined
Chern-Simons partition function on $L(2,1)$ as well as a
refined version of the partition function of N=6 d=3 ABJM theory.

Both approaches described in this paper need to be generalized
to incorporate open string boundary conditions or more
general Wilson lines than the one in $L(2,1)$. For the B-model
a remodeled and refined version describing Ooguri-Vafa
invariants would be highly desirable and for A-model one
would like to give a more stringent mathematical definition
of the moduli space related to invariants that the refined
vertex computes. Finally, wallcrossing should apply
more broadly than in the special cases described in
Section~\ref{kkvrefined}.

\bigskip\noindent
{\bf Acknowledgments.}  It is a pleasure to thank J.~Bryan, K.~Chung,
D.~Freed, M.~Mari\~no, G.~Moore, N.~Nekrasov, A.~Okounkov, and
R.~Pandharipande for helpful conversations.
The research of J.C.~and S.K.~was partially supported by NSF grants DMS-05-55678 and DMS-12-01089.
The research of A.K. is partially supported by the DFG grant KL2271/1-1.
S.K.~thanks the Max Planck Institut f\"ur Mathematik and
the Bethe Center for Theoretical Physics in Bonn for their hospitality, and
A.K~thanks the Department of Mathematics at the University of Illinois
at Urbana-Champaign for hospitality.


\begin{thebibliography}{9}

\bibitem{Aganagic:2011mi}
  M.~Aganagic, M.~C.~N.~Cheng, R.~Dijkgraaf, D.~Krefl and C.~Vafa,
  ``Quantum Geometry of Refined Topological Strings,'' arXiv:1105.0630 [hep-th].
  %%CITATION = ARXIV:1105.0630;%%

\bibitem{Aganagic:2002wv}
  M.~Aganagic, A.~Klemm, M.~Marino and C.~Vafa,
  ``Matrix model as a mirror of Chern-Simons theory,''
  JHEP {\bf 0402}, 010 (2004),
  hep-th/0211098.
  %%CITATION = HEP-TH/0211098;%%

\bibitem{Aganagic:2012au}
  M.~Aganagic and K.~Schaeffer,
  ``Orientifolds and the Refined Topological String,''
  JHEP {\bf 1209}, 084 (2012),
  arXiv:1202.4456 [hep-th].
  %%CITATION = ARXIV:1202.4456;%%.


\bibitem{Aganagic:2012hs}
  M.~Aganagic and S.~Shakirov,
  ``Refined Chern-Simons Theory and Topological String,''
  arXiv:1210.2733 [hep-th].
  %%CITATION = ARXIV:1210.2733;%%.

\bibitem{behrend}
K.~Behrend, ``Donaldson-{T}homas type invariants via microlocal geometry.'' {Annals of Mathematics}, {170}:{1307--1338}, (2009).

\bibitem{bbs} K.~Behrend, J.~Bryan, and B.~Szendroi, ``Motivic degree zero Donaldson-Thomas invariants'', to appear in Inv.\ Math., arXiv:0909.5088.

\bibitem{BCOV}
  M.~Bershadsky, S.~Cecotti, H.~Ooguri, C.~Vafa,
  ``Kodaira-Spencer theory of gravity and exact results for quantum string amplitudes,''
  Commun.\ Math.\ Phys.\  {\bf 165 } (1994),  311-428.
  hep-th/9309140.

\bibitem{CC1}
J.~Choi and K.~Chung, ``Moduli spaces of $\alpha$-stable pairs and wall-crossing on $\mathbb{P}^2$'', arXiv:1210.2499.

\bibitem{DS}
A.~Dimca and B.~Szendr\"{o}i, ``The Milnor fibre of the Pfaffian and the Hilbert scheme of four points on $\mathbb{C}^3$''. Math. Res. Lett. {\bf 16} (2009), 1037-1055.

\bibitem{Drukker:2010nc}
  N.~Drukker, M.~Marino and P.~Putrov,
  ``From weak to strong coupling in ABJM theory,''
  Commun.\ Math.\ Phys.\  {\bf 306}, 511 (2011),
 arXiv:1007.3837 [hep-th].
  %%CITATION = ARXIV:1007.3837;%%

\bibitem{fg} B.~Fantechi and L.~G\"ottsche, ``Riemann-Roch theorems and elliptic genus for virtually smooth Schemes'', arXiv:0706.0988.

\bibitem{GSY}
  D.~Gaiotto, A.~Strominger, X.~Yin,
  ``New connections between 4D and 5D black holes,''
  JHEP {\bf 0602}, 024 (2006),
  hep-th/0503217.


\bibitem{GV2}
  R.~Gopakumar, C.~Vafa,
``M theory and topological strings 2,'' hep-th/9812127.

\bibitem{gs} L.~G\"ottsche and V.~Shende , ``Refined curve counting on complex surfaces,'' arXiv/1208.1973 [math.AG].


\bibitem{Haghighat:2008gw}
  B.~Haghighat, A.~Klemm and M.~Rauch,
  ``Integrability of the holomorphic anomaly equations,''
  JHEP {\bf 0810}, 097 (2008)
  [arXiv:0809.1674 [hep-th]].
  %%CITATION = ARXIV:0809.1674;%%




\bibitem{mhe} M.~He. ``Espaces de Modules de syst\`emes coh\'erents''. Internat. J. of Math. {\bf 7} (1998), 545-598.

\bibitem{Hori:2000kt}
  K.~Hori and C.~Vafa,
  ``Mirror symmetry,''
  hep-th/0002222.
  %%CITATION = HEP-TH/0002222;%%


\bibitem{HKK}
  M.~-x.~Huang, A.~-K.~Kashani-Poor, A.~Klemm,
  ``The Omega deformed B-model for rigid N=2 theories,''
  arXiv:1109.5728 [hep-th].


\bibitem{Huang:2009md}
  M.~-x.~Huang and A.~Klemm,
  ``Holomorphicity and Modularity in Seiberg-Witten Theories with Matter,''
  JHEP {\bf 1007}, 083 (2010)
  [arXiv:0902.1325 [hep-th]].
  %%CITATION = ARXIV:0902.1325;%%

\bibitem{HK2}
  M.~-x.~Huang, A.~Klemm,
  ``Direct integration for general $\Omega$ backgrounds,''
  arXiv:1009.1126 [hep-th].


\bibitem{Huang:2006si}
  M.~-x.~Huang and A.~Klemm,
  %``Holomorphic Anomaly in Gauge Theories and Matrix Models,''
  JHEP {\bf 0709}, 054 (2007)
  [hep-th/0605195].
  %%CITATION = HEP-TH/0605195;%%

\bibitem{HuangKlemm}
  M.~-x.~Huang, A.~Klemm,
  ``Refined BPS index on local del Pezzo surfaces and half K3,'' to appear.



\bibitem{HKQ}
  M.~-x.~Huang, A.~Klemm, S.~Quackenbush,
  ``Topological string theory on compact Calabi-Yau: Modularity and boundary conditions,''
  Lect.\ Notes Phys.\  {\bf 757}, 45-102 (2009),
  hep-th/0612125.

\bibitem{IKV}
  A.~Iqbal, C.~Kozcaz, C.~Vafa,
  ``The Refined topological vertex,''
  JHEP {\bf 0910}, 069 (2009), hep-th/0701156.

\bibitem{Iqbal:2012xm}
  A.~Iqbal and C.~Vafa,
  ``BPS Degeneracies and Superconformal Index in Diverse Dimensions,''
  arXiv:1210.3605 [hep-th].
  %%CITATION = ARXIV:1210.3605;%%



\bibitem{Iqbal:2012mt}
  A.~Iqbal and C.~Kozcaz,
  ``Refined Topological Strings and Toric Calabi-Yau Threefolds,''
  arXiv:1210.3016 [hep-th].
  %%CITATION = ARXIV:1210.3016;%%


\bibitem{Kachru:1995fv}
  S.~Kachru, A.~Klemm, W.~Lerche, P.~Mayr and C.~Vafa,
  ``Nonperturbative results on the point particle limit of N=2 heterotic string compactifications,''
  Nucl.\ Phys.\ B {\bf 459}, 537 (1996), hep-th/9508155.
    %%CITATION = HEP-TH/9508155;%%

\bibitem{katz04}
  S.~H.~Katz,
  ``Gromov-Witten, Gopakumar-Vafa, and Donaldson-Thomas invariants of Calabi-Yau threefolds,'' math/0408266 [math-AG].

\bibitem{katz08}
S.~Katz, ``Genus zero Gopakumar-Vafa invariants of contractible curves,'' J.\ Diff.\
Geom.\ {\bf 79} (2008), 185--195.

\bibitem{KKV1}
  S.~H.~Katz, A.~Klemm and C.~Vafa,
  ``Geometric engineering of quantum field theories,''
  Nucl.\ Phys.\ B {\bf 497}, 173 (1997), hep-th/9609239.
  %%CITATION = HEP-TH/9609239;%%

\bibitem{kkv}
  S.~H.~Katz, A.~Klemm, C.~Vafa,
  ``M theory, topological strings and spinning black holes,''
  Adv.\ Theor.\ Math.\ Phys.\  {\bf 3}, 1445-1537 (1999), hep-th/9910181.

\bibitem{Klemm:2012ii}
  A.~Klemm, M.~Marino, M.~Schiereck and M.~Soroush,
  ``ABJM Wilson loops in the Fermi gas approach,'' arXiv:1207.0611 [hep-th].
  %%CITATION = ARXIV:1207.0611;%%

\bibitem{KK}
  D.~Krefl, J.~Walcher,
  ``Extended Holomorphic Anomaly in Gauge Theory,''
  Lett.\ Math.\ Phys.\  {\bf 95}, 67-88 (2011), arXiv:1007.0263 [hep-th].

\bibitem{KK2}
  D.~Krefl and J.~Walcher,
  ``Shift versus Extension in Refined Partition Functions,''
  arXiv:1010.2635 [hep-th].
  %%CITATION = ARXIV:1010.2635;%%

\bibitem{Krefl:2012ei}
  D.~Krefl, J.~Walcher and J.~Walcher,
  ``ABCD of Beta Ensembles and Topological Strings,''
  arXiv:1207.1438 [hep-th].
  %%CITATION = ARXIV:1207.1438;%%

\bibitem{Marino:2009jd}
  M.~Marino and P.~Putrov,
  ``Exact Results in ABJM Theory from Topological Strings,''
  JHEP {\bf 1006}, 011 (2010), arXiv:0912.3074 [hep-th].
  %%CITATION = ARXIV:0912.3074;%%

\bibitem{Marino:2002fk}
  M.~Marino,
  ``Chern-Simons theory, matrix integrals,
  and perturbative three manifold invariants,''
  Commun.\ Math.\ Phys.\  {\bf 253}, 25 (2004)
  [hep-th/0207096].
  %%CITATION = HEP-TH/0207096;%%

\bibitem{mnop1}
D.~Maulik, N.~Nekrasov, A.~Okounkov, and R.~Pandharipande, ``Gromov-Witten
theory and Donaldson-Thomas theory, I.'' \textit{Compos. Math.}, 142(5):1263--1285
(2006).

\bibitem{mmns}
A.~Morrison, S.~Mozgovoy, K.~Nagao, and B.~Szendroi,
``Motivic Donaldson-Thomas invariants of the conifold and the refined topological vertex'', arXiv:1107.5017 [math].


%\cite{Nakayama:2011be}
\bibitem{Nakayama:2011be}
 Y.~Nakayama and H.~Ooguri,
``Comments on Worldsheet Description of the Omega Background,''
Nucl.\ Phys.\ B {\bf 856}, 342 (2012)
[arXiv:1106.5503 [hep-th]].
%%CITATION = ARXIV:1106.5503;%%


\bibitem{Nekrasov:2002qd}
  N.~A.~Nekrasov,
  ``Seiberg-Witten prepotential from instanton counting,''
  Adv.\ Theor.\ Math.\ Phys.\  {\bf 7} 831 (2004), hep-th/0206161.
  %%CITATION = HEP-TH/0206161;%%

\bibitem{no} N.~Nekrasov and A.~Okounkov, ``The index of M-theory,'' in
preparation.

%\cite{Nekrasov:2011bc}
\bibitem{Nekrasov:2011bc}
N.~Nekrasov, A.~Rosly and S.~Shatashvili,
``Darboux coordinates, Yang-Yang functional, and gauge theory,''
Nucl.\ Phys.\ Proc.\ Suppl.\  {\bf 216}, 69 (2011)
[arXiv:1103.3919 [hep-th]].
%%CITATION = ARXIV:1103.3919;%%

\bibitem{Nekrasov:2010ka}
  N.~Nekrasov and E.~Witten,
  ``The Omega Deformation, Branes, Integrability, and Liouville Theory,''
  JHEP {\bf 1009}, 092 (2010)
  [arXiv:1002.0888 [hep-th]].



\bibitem{pt} R.~Pandharipande and R.P.~Thomas, ``Curve counting via stable
pairs in the derived category,'' Inv.\ Math.\ {bf 178} 407--447 (2009), arXiv:0707.2348 [math.AG].

\bibitem{ptvertex} R.~Pandharipande and R.P~Thomas, ``The 3-fold vertex via
stable pairs,'' Geom.\ \& Top.\ {\bf 13}, 1835--1876 (2009),
arXiv:0709.3823 [math.AG].

\bibitem{ptbps} R.~Pandharipande and R.P~Thomas,
``Stable Pairs and BPS Invariants,''
Jour.\ AMS {\bf 23}, 267--297 (2010), arXiv:0711.3899 [math.AG].

\bibitem{Witten:1993yc}
  E.~Witten,
  ``Phases of N=2 theories in two-dimensions,''
  Nucl.\ Phys.\ B {\bf 403}, 159 (1993), hep-th/9301042.
  %%CITATION = HEP-TH/9301042;%%


\end{thebibliography}
\end{document}